\documentclass[a4paper,english,USenglish]{lipics-v2016}

\usepackage{microtype}
\usepackage{xcolor}
\usepackage{xspace}
\usepackage{mathpartir}
\usepackage{prooftree}
\usepackage{calrsfs,mathbbol}
\usepackage{extarrows} % \xlongrightarrow
\usepackage{stmaryrd} % \llbracket, \rrbracket
\usepackage{cmll} % \parr, \with
\usepackage{etoolbox} % \ifblank
\usepackage{mathtools} % \mathclap
\usepackage{comment}

\newif\ifproofs
\proofstrue
\proofsfalse

\newif\ifcomments
\commentstrue
\commentsfalse

\newif\iflong
\longtrue
\longfalse

\newcommand{\ttemark}{\texttt{\upshape!}}
\newcommand{\ttqmark}{\texttt{\upshape?}}

\newcommand{\ttbar}{\texttt{\upshape|}}

\newcommand{\ttunderscore}{\texttt{\upshape\_}}
\newcommand{\ttat}{\texttt{\upshape@}}

\definecolor{mymagenta}{rgb}{0.5,0,0.5}
\definecolor{myred}{rgb}{0.3,0,0}
\definecolor{mygreen}{rgb}{0,0.2,0}
\definecolor{myblue}{rgb}{0,0,0.6}
\definecolor{highlight}{gray}{0.9}
\definecolor{listing}{gray}{0.925}

\usepackage{listings}
\lstdefinestyle{INLINE}{
  keywordstyle=\color{myblue},
}
\lstdefinestyle{DISPLAY}{
  style=INLINE,
  basicstyle=\tt\small,
  numberstyle=\tiny\sffamily\color{myred},
  numbersep=1em,
  numbers=left,
  escapeinside={//}{\^^M},
  backgroundcolor=\color{highlight},
}
\lstdefinestyle{FLOAT}{
  style=DISPLAY,
  float,
  captionpos=b,
}
\newcommand{\InputScala}[1][]{\lstinputlisting[language=Scala,style=FLOAT,#1]}
\newcommand{\SI}{\lstinline[language=Scala,style=INLINE]}

\newcommand{\mkkeyword}[1]{\mathtt{\color{myblue}#1}}
\newcommand{\mkconstant}[1]{\mathtt{\color{mymagenta}#1}}
\newcommand{\seqof}[1]{\overline{#1}}
\newcommand{\set}[2][]{\{#2\}\ifblank{#1}{}{_{#1}}}
\newcommand{\bag}[2][]{\langle#2\rangle\ifblank{#1}{}{_{#1}}}
\newcommand{\pair}[2]{(#1,#2)}
\newcommand{\upair}[2]{\set{#1,#2}}

\newcommand{\proofcase}[1]{\medskip\noindent Case #1.}

\newcommand{\proofrule}[1]{\proofcase{\refrule{#1}}}

\newcommand{\parens}[1]{(#1)}
\newcommand{\bigparens}[1]{\left(#1\right)}

\newcommand{\cf}{\emph{cf.}\xspace}
\newcommand{\etal}{\emph{et al.}\xspace}

\newcommand{\DiscussRule}[1]{\subparagraph*{#1.}}

\newcommand{\True}{\mkconstant{True}}

\newcommand{\mkint}[1]{\mkconstant{#1}}

\newcommand{\VarSystem}{\mathit{system}}
\newcommand{\TagPrintBool}{\Tag[print\ttunderscore bool]}
\newcommand{\TagPrintInt}{\Tag[print\ttunderscore int]}

\newcommand{\Mailbox}{\MailboxA}
\newcommand{\MailboxA}{a}
\newcommand{\MailboxB}{b}
\newcommand{\MailboxC}{c}

\newcommand{\Mailboxes}{\MailboxesA}
\newcommand{\MailboxesA}{\seqof\MailboxA}

\newcommand{\MailboxesC}{\seqof\MailboxC}

\newcommand{\Var}[1][x]{\mathit{#1}}
\newcommand{\VarX}{\Var[x]}
\newcommand{\VarY}{\Var[y]}

\newcommand{\VarClient}{\Var[client]}

\newcommand{\Vars}{\VarsX}
\newcommand{\VarsX}{\seqof\VarX}

\newcommand{\Name}{\NameU}
\newcommand{\NameU}{u}
\newcommand{\NameV}{v}
\newcommand{\NameW}{w}

\newcommand{\Names}{\NamesU}
\newcommand{\NamesU}{\seqof\NameU}
\newcommand{\NamesV}{\seqof\NameV}

\newcommand{\Tag}[1][m]{\mathtt{\color{myred}#1}}
\newcommand{\TagA}{\Tag[A]}
\newcommand{\TagB}{\Tag[B]}
\newcommand{\TagC}{\Tag[C]}
\newcommand{\TagReceive}{\Tag[receive]}
\newcommand{\TagSend}{\Tag[send]}
\newcommand{\TagLeft}{\Tag[left]}
\newcommand{\TagRight}{\Tag[right]}
\newcommand{\TagAcquire}{\Tag[acquire]}
\newcommand{\TagRelease}{\Tag[release]}
\newcommand{\TagReply}{\Tag[reply]}
\newcommand{\TagPut}{\Tag[put]}
\newcommand{\TagGet}{\Tag[get]}
\newcommand{\TagResult}{\Tag[result]}

\newcommand{\TagResolve}{\Tag[resolve]}
\newcommand{\TagCredit}{\Tag[credit]}
\newcommand{\TagDebit}{\Tag[debit]}
\newcommand{\TagStop}{\Tag[stop]}

\newcommand{\Process}{\ProcessP}
\newcommand{\ProcessP}{P}
\newcommand{\ProcessQ}{Q}
\newcommand{\ProcessR}{R}

\newcommand{\Guard}{\GuardG}
\newcommand{\GuardG}{G}
\newcommand{\GuardH}{H}

\newcommand{\Hole}{[~]}

\newcommand{\Done}{\mkkeyword{done}}
\newcommand{\Fail}[1]{\mkkeyword{fail}\ifblank{#1}{}{~#1}}
\newcommand{\Free}[2][\ttdot]{\mkkeyword{free}\ifblank{#2}{}{~#2}#1}
\newcommand{\Send}[3]{%
  #1\ttemark#2{}\ifblank{#3}{}{[#3]}%
}
\newcommand{\Receive}[4][\ttdot]{%
  #2\ttqmark#3\ifblank{#4}{}{(#4)}#1%
}
\newcommand{\New}[1]{(\nu#1)}
\newcommand{\parop}{\mathbin\ttbar}
\newcommand{\sumop}{+}

\newcommand{\IfX}[3]{\mkkeyword{if}#1\mkkeyword{then}#2\mkkeyword{else}#3}
\newcommand{\If}[3]{\IfX{~#1~}{~#2~}{~#3}}

\newcommand{\invoke}[2]{#1{}[#2]}
\newcommand{\define}[2]{#1(#2) \triangleq}
\newcommand{\RecVar}[1][X]{\mathsf{\color{mygreen}#1}}

\newcommand{\ProcessContext}{\mathcal{C}}

\newcommand{\rulename}[1]{\text{\upshape\scriptsize[\textsc{#1}]}}
\newcommand{\defrule}[1]{\hypertarget{rule:#1}{\rulename{#1}}}
\newcommand{\refrule}[1]{\hyperlink{rule:#1}{\rulename{#1}}}

\newenvironment{lines}[1][]{
  \begin{array}[#1]{@{}l@{}}
  }{
  \end{array}
}

\newcommand{\Bag}{\BagA}
\newcommand{\BagA}{\mathsf{A}}
\newcommand{\BagB}{\mathsf{B}}

\newcommand{\bunion}{\uplus}

\newcommand{\dgraph}{\dgraphA}
\newcommand{\dgraphA}{\varphi}
\newcommand{\dgraphB}{\psi}

\newcommand{\gempty}{\emptyset}
\newcommand{\gunion}{\sqcup}
\newcommand{\gedge}[2]{\upair{#1}{#2}}

\newcommand{\grel}[1]{\mathsf{dep}(#1)}
\newcommand{\glabel}[2]{#1 - #2}

\newcommand{\gimplies}{\Rightarrow}

\newcommand{\Type}{\TypeT}
\newcommand{\TypeT}{\tau}
\newcommand{\TypeS}{\sigma}
\newcommand{\TypeR}{\rho}

\newcommand{\Types}{\seqof\Type}
\newcommand{\TypesT}{\seqof\TypeT}
\newcommand{\TypesS}{\seqof\TypeS}
\newcommand{\TypesR}{\seqof\TypeR}

\newcommand{\tnat}{\mkkeyword{nat}}
\newcommand{\tint}{\mkkeyword{int}}
\newcommand{\tbool}{\mkkeyword{bool}}

\newcommand{\In}{\ttqmark}
\newcommand{\Out}{\ttemark}

\newcommand{\SessionType}{\SessionTypeT}
\newcommand{\SessionTypeT}{T}
\newcommand{\SessionTypeS}{S}

\newcommand{\End}{\mkkeyword{end}}
\newcommand{\SBranch}{\mathbin{\with}}
\newcommand{\SChoice}{\mathbin{\oplus}}
\newcommand{\SFork}{{\otimes}}
\newcommand{\SJoin}{{\parr}}

\newcommand{\co}[1]{\overline{#1}}

\newcommand{\Pattern}{\PatternE}
\newcommand{\PatternE}{E}
\newcommand{\PatternF}{F}
\newcommand{\PatternG}{G}

\newcommand{\Patterns}{\seqof\Pattern}

\newcommand{\MessageType}{\mathtt{M}}

\newcommand{\tzero}{\mathbb{0}}
\newcommand{\tone}{\mathbb{1}}
\newcommand{\tmul}{\cdot}
\newcommand{\tMul}{\prod}
\newcommand{\tsum}{+}
\newcommand{\tstar}{^{\ast}}

\newcommand{\tmessage}[2]{#1{[}#2{]}}

\newcommand{\EmptyContext}{\emptyset}
\newcommand{\Context}{\ContextA}
\newcommand{\ContextA}{\Gamma}
\newcommand{\ContextB}{\Delta}

\newcommand{\wtp}[3]{#1 \vdash #2\ifblank{#3}{}{{} :: #3}}
\newcommand{\wtg}[2]{#1 \vdash #2}

\newcommand{\fn}{\mathsf{fn}}
\newcommand{\bn}{\mathsf{bn}}

\newcommand{\dom}{\mathsf{dom}}
\newcommand{\sem}[1]{\llbracket#1\rrbracket}
\newcommand{\subst}[2]{\set{#1/#2}}

\newcommand{\pder}[2]{#1/#2}
\newcommand{\cmul}{\mathbin{\|}}
\newcommand{\size}{\mathsf{size}}

\newcommand{\subp}[1][]{\sqsubseteq\ifblank{#1}{}{_{#1}}}
\newcommand{\supp}[1][]{\sqsupseteq\ifblank{#1}{}{_{#1}}}
\newcommand{\eqp}[1][]{\simeq\ifblank{#1}{}{_{#1}}}

\newcommand{\subt}{\leqslant}
\newcommand{\supt}{\geqslant}
\newcommand{\eqt}{\lessgtr}
\newcommand{\eqdef}{\stackrel{\text{\tiny def}}=}

\newcommand{\red}{\rightarrow}
\newcommand{\nred}{\arrownot\red}

\newcommand{\gred}[2]{\xlongrightarrow{\glabel{#1}{#2}}}

\newcommand{\rrel}{\mathrel{\mathcal{R}}}
\newcommand{\df}[2]{#1 \vDash #2}

\newcommand{\marginnote}[2]{%
  \ifcomments%
  {\makebox[0pt]{\color{magenta}$^\bigstar$}}%
  \marginpar{\parbox{2.5cm}{\flushleft\tiny\sffamily\textbf{#1}: #2}}%
  \fi%
}
\newcommand{\Luca}[1]{\marginnote{Luca}{\color{magenta}#1}}
\newcommand{\Ugo}[1]{\marginnote{Ugo}{\color{magenta}#1}}

\theoremstyle{note}
\newtheorem*{notation}{Notation}

\theoremstyle{plain}
\newtheorem{proposition}[theorem]{Proposition}

\makeatletter
\newtheorem*{rep@theorem}{\rep@title}
\newcommand{\newreptheorem}[2]{%
\newenvironment{rep#1}[2][]{%
 \def\rep@title{\bfseries #2 \ref{##2}}%
 \begin{rep@theorem}[##1]}%
 {\end{rep@theorem}}}
\makeatother

\newreptheorem{theorem}{Theorem}

\newcommand{\eoe}{\hfill$\blacksquare$}

\title{Mailbox Types for Unordered Interactions}

\iffalse
\author{Authors omitted for review}
\authorrunning{}
\affil{Affiliations omitted for review}
\Copyright{The authors}
\else
\author{Ugo de'Liguoro}
\author{Luca Padovani}
\affil{Universit\`a di Torino, Dipartimento di Informatica, Torino, Italy}
\authorrunning{U. de'Liguoro and L. Padovani}
\Copyright{Ugo de'Liguoro and Luca Padovani}
\fi



\EventEditors{John Q. Open and Joan R. Acces}
\EventNoEds{2}
\EventLongTitle{42nd Conference on Very Important Topics (CVIT 2016)}
\EventShortTitle{CVIT 2016}
\EventAcronym{CVIT}
\EventYear{2016}
\EventDate{December 24--27, 2016}
\EventLocation{Little Whinging, United Kingdom}
\EventLogo{}
\SeriesVolume{42}
\ArticleNo{23}
\begin{document}

\maketitle

\begin{abstract}
  We propose a type system for reasoning on protocol conformance and
  deadlock freedom in networks of processes that communicate through
  unordered mailboxes.
  We model these networks in the mailbox calculus, a mild extension
  of the asynchronous $\pi$-calculus with first-class mailboxes and
  selective input.
  The calculus subsumes the actor model and allows us to analyze
  networks with dynamic topologies and varying number of processes
  possibly mixing different concurrency abstractions.
  Well-typed processes are deadlock free and never fail because of
  unexpected messages. For a non-trivial class of them, junk freedom
  is also guaranteed.
  We illustrate the expressiveness of the calculus and of the type
  system by encoding instances of non-uniform, concurrent objects,
  binary sessions extended with joins and forks, and some known
  actor benchmarks.
\end{abstract}

%%% Local Variables:
%%% mode: latex
%%% TeX-master: "main"
%%% End:

\section{Introduction}
\label{sec:introduction}

Message passing is a key mechanism used to coordinate concurrent
processes.
The order in which a process consumes messages may coincide with the
order in which they arrive at destination (\emph{ordered
  processing}) or may depend on some intrinsic property of the
messages themselves, such as their priority, their tag, or the shape
of their content (\emph{out-of-order or selective processing}).
Ordered message processing is common in networks of processes
connected by point-to-point \emph{channels}.
Out-of-order message processing is common in networks of processes
using \emph{mailboxes}, into which processes concurrently store
messages and from which one process selectively receives messages.
This communication model is typically found in the various
implementations of actors~\cite{HewittEtAl73,Agha86} such as
Erlang~\cite{Armstrong13}, Scala and Akka
actors~\cite{HallerSommers11}, CAF~\cite{CharoussetHiesgenSchmidt16}
and Kilim~\cite{SrinivasanMycroft08}.
Non-uniform, concurrent
objects~\cite{RavaraVasconcelos00,Puntigam01,CrafaPadovani17} are
also examples of out-of-order message processors. For example, a
busy lock postpones the processing of any $\Tag[acquire]$ message
until it is released by its current owner.
Out-of-order message processing adds further complexity to the
challenging task of concurrent and parallel application development:
storing a message into the wrong mailbox or at the wrong time,
forgetting a message in a mailbox, or relying on the presence of a
particular message that is not guaranteed to be found in a mailbox
are programming mistakes that are easy to do and hard to detect
without adequate support from the language and its development
tools.

\InputScala[label=lst:account,caption=An example of Scala actor taken from the \texttt{Savina} benchmark suite~\cite{ImamSarkar14}.]{Account.scala}

The Scala actor in Listing~\ref{lst:account}, taken from the
\texttt{Savina} benchmark suite~\cite{ImamSarkar14}, allows us to
illustrate some of the subtle pitfalls that programmers must
carefully avoid when dealing with out-of-order message processing.
The \SI{process} method matches messages found in the actor's
mailbox according to their type.
If a message of type \SI{DebitMessage} is found, then \SI{balance}
is incremented by the deposited amount and the actor requesting the
operation is notified with a \SI{ReplyMessage}
(lines~\ref{begin.account.debit}--\ref{end.account.debit}).
If a message of type \SI{CreditMessage} is found, \SI{balance} is
decremented by the amount that is transferred to \SI{recipient}
(lines~\ref{begin.account.credit}--\ref{account.credit.send}). Since
the operation is meant to be atomic, the actor temporarily changes
its behavior and waits for a \SI{ReplyMessage} from \SI{recipient}
signalling that the transfer is complete, before notifying
\SI{sender} in turn
(lines~\ref{account.credit.wait}--\ref{end.account.credit}).
A message of type \SI{StopMessage} terminates the actor
(line~\ref{account.stop}).

Note how the correct execution of this code depends on some key
assumptions:
\begin{itemize}
\item \SI{ReplyMessage} should be stored in the actor's mailbox only
  when the actor is involved in a transaction, or else the message
  would trigger the ``catch all'' clause that throws a ``unsupported
  message'' exception
  (lines~\ref{begin.account.unknown}--\ref{end.account.unknown}).
\item No debit or credit message should be in the actor's mailbox by
  the time it receives \SI{StopMessage}, or else some critical
  operations affecting the balance would not be performed.
\item Two distinct accounts should not try to simultaneously
  initiate a transaction with each other. If this were allowed, each
  account could consume the credit message found in its own mailbox
  and then deadlock waiting for a reply from the other account
  (lines~\ref{account.credit.wait}--\ref{end.account.credit}).
\end{itemize}

Static analysis techniques that certify the validity of assumptions
like these can be valuable for developers.  For example, session
types~\cite{HuttelEtAl16} have proved to be an effective formalism
for the enforcement of communication protocols and have been applied
to a variety of programming paradigms and
languages~\cite{AnconaEtAl16}, including those based on mailbox
communications~\cite{MostrousVasconcelos11,CharalambidesDingesAgha16,Fowler16,NeykovaYoshida16}.
However, session types are specifically designed to address
point-to-point, ordered interactions over
channels~\cite{Honda93}. Retrofitting them to a substantially
different communication model calls for some inevitable compromises
on the network topologies that can be addressed and forces
programmers to give up some of the flexibility offered by unordered
message processing.

Another aspect that complicates the analysis of actor systems is
that the pure actor model as it has been originally
conceived~\cite{HewittEtAl73,Agha86} does not accurately reflect the
actual practice of actor programming. In the pure actor model, each
actor owns a single mailbox and the only synchronization mechanism
is message reception from such mailbox.
However, it is a known fact that the implementation of complex
coordination protocols in the pure actor model is
challenging~\cite{VarelaAgha01,TasharofiDingesJohnson13,ImamSarkar14bis,ChatterjeeEtAl16}. These
difficulties have led programmers to mix the actor model with
different concurrency
abstractions~\cite{ImamSarkar12,TasharofiDingesJohnson13}, to extend
actors with controlled forms of synchronization~\cite{VarelaAgha01}
and to consider actors with multiple/first-class
mailboxes~\cite{Haller12,ImamSarkar14bis,ChatterjeeEtAl16}.
In fact, popular implementations of the actor model feature
disguised instances of multiple/first-class mailbox usage, even if
they are not explicitly presented as such: in Akka, the messages
that an actor is unable to process immediately can be temporarily
stashed into a different mailbox~\cite{Haller12}; in Erlang, hot
code swapping implies transferring at runtime the input capability
on a mailbox from a piece of code to a different
one~\cite{Armstrong13}.

In summary, there is still a considerable gap between the scope of
available approaches used to analyze mailbox-based communicating
systems and the array of features used in programming these
systems. To help narrowing this gap, we make the following
contributions:
\begin{itemize}
\item We introduce \emph{mailbox types}, a new kind of behavioral
  types with a simple and intuitive semantics embodying the
  unordered nature of mailboxes. Mailbox types allow us to describe
  mailboxes subject to selective message processing as well as
  mailboxes concurrently accessed by several
  processes. Incidentally, mailbox types also provide precise
  information on the size and reachability of mailboxes that may
  lead to valuable code optimizations.
\item We develop a mailbox type system for the \emph{mailbox
    calculus}, a mild extension of the asynchronous
  $\pi$-calculus~\cite{SangiorgiWalker01} featuring tagged messages,
  selective inputs and first-class mailboxes.  The mailbox calculus
  allows us to address a broad range of systems with dynamic
  topology and varying number of processes possibly using a mixture
  of concurrency models (including multi-mailbox actors) and
  abstractions (such as locks and futures).
\item We prove three main properties of well-typed processes: the
  absence of failures due to unexpected messages (\emph{mailbox
    conformance}); the absence of pending activities and messages in
  irreducible processes (\emph{deadlock freedom}); for a non-trivial
  class of processes, the guarantee that every message can be
  eventually consumed (\emph{junk freedom}).
\item We illustrate the expressiveness of mailbox types by
  presenting well-typed encodings of known concurrent objects (locks
  and futures) and actor benchmarks (atomic transactions and
  master-workers parallelism) and of binary sessions extended with
  forks and joins. In discussing these examples, we emphasize the
  impact of out-of-order message processing and of first-class
  mailboxes.
\end{itemize}

\subparagraph*{Structure of the paper.}
We start from the definition of the \emph{mailbox calculus} and of
the properties we expect from well-typed processes
(Section~\ref{sec:model}).
We introduce mailbox types (Section~\ref{sec:types}) and dependency
graphs (Section~\ref{sec:graphs}) for tracking mailbox dependencies
in processes that use more than one. Then, we present the typing
rules (Section~\ref{sec:rules}) and the soundness results of the
type system (Section~\ref{sec:properties}).
In the latter part of the paper, we discuss a few more complex
examples (Section~\ref{sec:examples}), related work
(Section~\ref{sec:related}) and ideas for further developments
(Section~\ref{sec:conclusions}).
Additional definitions and proofs can be found in
Appendices~\ref{sec:extra_properties}--\ref{sec:example_readers_writer}.

%%% Local Variables:
%%% mode: latex
%%% TeX-master: "main"
%%% End:

\newcommand{\DefLock}{\RecVar[Lock]}
\newcommand{\DefFreeLock}{\RecVar[FreeLock]}
\newcommand{\DefBusyLock}{\RecVar[BusyLock]}
\newcommand{\DefUser}{\RecVar[User]}
\newcommand{\DefFuture}{\RecVar[Future]}
\newcommand{\DefResolvedFuture}{\RecVar[Present]}
\newcommand{\DefAccount}{\RecVar[Account]}

\newcommand{\VarLock}{\Var[lock]}
\newcommand{\VarOwner}{\Var[owner]}
\newcommand{\VarFuture}{\Var[f]}
\newcommand{\VarSelf}{\Var[self]}
\newcommand{\VarBalance}{\Var[balance]}
\newcommand{\VarAmount}{\Var[amount]}
\newcommand{\VarSender}{\Var[sender]}
\newcommand{\VarRecipient}{\Var[recipient]}
\newcommand{\VarAlice}{\Var[alice]}
\newcommand{\VarBank}{\Var[bank]}
\newcommand{\VarCarol}{\Var[carol]}

\section{The Mailbox Calculus}
\label{sec:model}

\begin{table}
  \begin{center}
    \[
      \begin{array}[t]{@{}rr@{~}c@{~}ll@{}}
        \textbf{Process} & \ProcessP, \ProcessQ
        & ::= & \Done & \text{(termination)}
        \\
        & & | & \invoke\RecVar\Names & \text{(invocation)}
        \\
        & & | & \Guard & \text{(guarded process)}
        \\
        & & | & \Send\Name\Tag\NamesV & \text{(stored message)}
        \\
        & & | & \ProcessP \parop \ProcessQ & \text{(parallel composition)}
        \\
        & & | & \New\Mailbox\Process & \text{(mailbox restriction)}
        \\\\
        \textbf{Guard} & \GuardG, \GuardH
        & ::= & \Fail\Name & \text{(runtime error)}
        \\
        & & | & \Free\Name\Process & \text{(mailbox deletion)}
        \\
        & & | & \Receive\Name\Tag\Vars\Process & \text{(selective receive)}
        \\
        & & | & \GuardG \sumop \GuardH & \text{(guard composition)}
      \end{array}
    \]
  \end{center}
  \caption{\label{tab:processes} Syntax of the mailbox calculus.}
\end{table}

We assume given an infinite set of \emph{variables} $\VarX$,
$\VarY$, an infinite set of \emph{mailbox names} $\MailboxA$,
$\MailboxB$, a set of \emph{tags} $\Tag$ and a finite set of
\emph{process variables} $\RecVar$.
We let $\NameU$, $\NameV$ range over variables and mailbox names
without distinction.
Throughout the paper we write $\seqof{e}$ for possibly empty
sequences $e_1, \dots, e_n$ of various entities. For example,
$\Names$ stands for a sequence $\Name_1,\dots,\Name_n$ of names and
$\set\Names$ for the corresponding set.

The syntax of the \emph{mailbox calculus} is shown in
Table~\ref{tab:processes}.
The term $\Done$ represents the terminated process that performs no
action.
The term $\Send\Name\Tag\NamesV$ represents a message stored in
mailbox $\Name$. The message has tag $\Tag$ and arguments $\NamesV$.
The term $\ProcessP \parop \ProcessQ$ represents the parallel
composition of $\ProcessP$ and $\ProcessQ$ and
$\New\Mailbox\Process$ represents a restricted mailbox $\Mailbox$
with scope $\Process$.
The term $\invoke\RecVar\Names$ represents the invocation of the
process named $\RecVar$ with parameters $\Names$. For each process
variable $\RecVar$ we assume that there is a corresponding
\emph{global process definition} of the form
$\define\RecVar\Vars\Process$.
A guarded process $\Guard$ is a composition of actions.
The action $\Fail\Name$ represents the process that fails with an
error for having received an unexpected message from mailbox
$\Name$.
The action $\Free\Name\Process$ represents the process that deletes
the mailbox $\Name$ if it is empty and then continues as $\Process$.
The action $\Receive\Name\Tag\Vars\Process$ represents the process
that receives an $\Tag$-tagged message from mailbox $\Name$ then
continues as $\Process$ with $\Vars$ replaced by the message's
arguments.
A compound guard $\GuardG \sumop \GuardH$ offers all the actions
offered by $\GuardG$ and $\GuardH$. We assume that all actions in
the same guard refer to the same mailbox $\Name$.
The notions of free and bound names of a process $\Process$ are
standard and respectively denoted by $\fn(\Process)$ and
$\bn(\Process)$.

The operational semantics of the mailbox calculus is mostly
conventional. We use the structural congruence relation $\equiv$
defined below to rearrange equivalent processes:
\[
  \begin{array}{@{}*{3}{r@{~}c@{~}l}@{}}
    \Fail\Mailbox \sumop \Guard & \equiv & \Guard
    &
    \GuardG \sumop \GuardH & \equiv & \GuardH \sumop \GuardG
    &
    \GuardG \sumop (\GuardH \sumop \GuardH') & \equiv & (\GuardG \sumop \GuardH) \sumop \GuardH'
    \\
    \Done \parop \Process & \equiv & \Process
    &
    \ProcessP \parop \ProcessQ & \equiv & \ProcessQ \parop \ProcessP
    &
    \ProcessP \parop (\ProcessQ \parop \ProcessR)
    & \equiv &
    (\ProcessP \parop \ProcessQ) \parop \ProcessR
    \\
    & & &
    \New\MailboxA\New\MailboxB\Process
    & \equiv &
    \New\MailboxB\New\MailboxA\Process
    &
    \New\Mailbox\ProcessP \parop \ProcessQ
    & \equiv &
    \New\Mailbox(\ProcessP \parop \ProcessQ)
    \quad\text{if $\Mailbox \not\in \fn(\ProcessQ)$}
  \end{array}
\]

Structural congruence captures the usual commutativity and
associativity laws of action and process compositions, with
$\Fail{}$ and $\Done$ acting as the respective units.  Additionally,
the order of mailbox restrictions is irrelevant and the scope of a
mailbox may shrink or extend dynamically.
The reduction relation $\red$ is inductively defined by the rules
\[
  \begin{array}{rr@{~}c@{~}ll}
    \defrule{r-read} &
    \Send\Mailbox\Tag{\seqof\MailboxC}
    \parop
    \Receive\Mailbox\Tag\Vars\Process \sumop \Guard
    & \red &
    \Process\subst{\seqof\MailboxC}\Vars
    \\
    \defrule{r-free} &
    \New\Mailbox\parens{\Free\Mailbox\Process \sumop \Guard}
    & \red &
    \Process
    \\
    \defrule{r-def} &
    \invoke\RecVar{\seqof\MailboxC}
    & \red &
    \Process\subst{\seqof\MailboxC}\Vars
    & \text{if $\define\RecVar\Vars\Process$}
    \\
    \defrule{r-par} &
    \ProcessP \parop \ProcessR
    & \red &
    \ProcessQ \parop \ProcessR
    & \text{if $\ProcessP \red \ProcessQ$}
    \\
    \defrule{r-new} &
    \New\Mailbox\ProcessP
    & \red &
    \New\Mailbox\ProcessQ
    & \text{if $\ProcessP \red \ProcessQ$}
    \\
    \defrule{r-struct} &
    \ProcessP
    & \red &
    \ProcessQ
    & \text{if $\ProcessP \equiv \ProcessP' \red \ProcessQ' \equiv \ProcessQ$}
  \end{array}
\]
where $\Process\subst\MailboxesC\Vars$ denotes the usual
capture-avoiding replacement of the variables $\Vars$ with the
mailbox names $\MailboxesC$.
Rule~\refrule{r-read} models the selective reception of an
$\Tag$-tagged message from mailbox $\Mailbox$, which erases all the
other actions of the guard.
Rule~\refrule{r-free} is triggered when the process is ready to
delete the empty mailbox $\Mailbox$ and no more messages can be
stored in $\Mailbox$ because there are no other processes in the
scope of $\Mailbox$.
Rule~\refrule{r-def} models a process invocation by replacing the
process variable $\RecVar$ with the corresponding definition.
Finally, rules~\refrule{r-par}, \refrule{r-new} and
\refrule{r-struct} close reductions under parallel compositions,
name restrictions and structural congruence.
We write $\red^*$ for the reflexive and transitive closure of
$\red$, we write $\ProcessP \nred^* \ProcessQ$ if not
$\ProcessP \red^* \ProcessQ$ and $\Process \nred$ if
$\ProcessP \nred \ProcessQ$ for all $\ProcessQ$.

Hereafter, we will occasionally use numbers and conditionals in
processes. These and other features can be either encoded or added
to the calculus without difficulties.

\begin{example}[lock]
  \label{ex:lock}
  In this example we model a lock as a process that waits for
  messages from a $\VarSelf$ mailbox in which acquisition and
  release requests are stored. The lock is either free or busy. When
  in state free, the lock nondeterministically consumes an
  $\TagAcquire$ message from $\VarSelf$. This message indicates the
  willingness to acquire the lock by another process and carries a
  reference to a mailbox into which the lock stores a $\TagReply$
  notification. When in state busy, the lock waits for a
  $\TagRelease$ message indicating that it is being released:
  \[
    \begin{array}{@{}r@{~}c@{~}l@{}}
      \DefFreeLock(\VarSelf) & \triangleq &
      \Free\VarSelf\Done
      \\ & \sumop &
      \Receive\VarSelf\TagAcquire\VarOwner
      \invoke\DefBusyLock{\VarSelf,\VarOwner}
      \\ & \sumop &
      \Receive\VarSelf\TagRelease{}
      \Fail\VarSelf
      \\
      \DefBusyLock(\VarSelf,\VarOwner) & \triangleq &
      \Send\VarOwner\TagReply\VarSelf \parop
      \Receive\VarSelf\TagRelease{}\invoke\DefFreeLock\VarSelf
    \end{array}
  \]

  Note the presence of the $\Free[]\VarSelf$ guard in the definition
  of $\DefFreeLock$ and the lack thereof in $\DefBusyLock$. In the
  former case, the lock manifests the possibility that no process is
  willing to acquire the lock, in which case it deletes the mailbox
  and terminates. In the latter case, the lock manifests its
  expectation to be eventually released by its current owner.
  Also note that $\DefFreeLock$ fails if it receives a $\TagRelease$
  message. In this way, the lock manifests the fact that it can be
  released only if it is currently owned by a process.
  A system where two users $\VarAlice$ and $\VarCarol$ compete for
  acquiring $\VarLock$ can be modeled as the process
  \begin{equation}
    \label{eq:alice_carol}
    \New\VarLock
    \New\VarAlice
    \New\VarCarol\parens{
      \invoke\DefFreeLock\VarLock
      \parop
      \invoke\DefUser{\VarAlice, \VarLock}
      \parop
      \invoke\DefUser{\VarCarol, \VarLock}
    }
  \end{equation}
  where
  \[
    \DefUser(\VarSelf, \VarLock) \triangleq
    \Send\VarLock\TagAcquire\VarSelf \parop
    \Receive\VarSelf\TagReply{l}
    \parens{
      \Send{l}\TagRelease{}
      \parop
      \Free\VarSelf\Done
    }
  \]

  Note that $\DefUser$ uses the reference $l$ -- as opposed to
  $\VarLock$ -- to release the acquired lock. As we will see in
  Section~\ref{sec:rules}, this is due to the fact that it is this
  particular reference to the lock's mailbox -- and not $\VarLock$
  itself -- that carries the capability to release the lock.
  \eoe
\end{example}

\begin{example}[future variable]
  \label{ex:future}
  A future variable is a one-place buffer that stores the result of
  an asynchronous computation. The content of the future variable is
  set once and for all by the producer once the computation
  completes. This phase is sometimes called \emph{resolution} of the
  future variable. After the future variable has been resolved, its
  content can be retrieved any number of times by the consumers. If
  a consumer attempts to retrieve the content of the future variable
  beforehand, the consumer suspends until the variable is resolved.
  We can model a future variable thus:
  \[
    \begin{array}{@{}r@{~}c@{~}l@{}}
      \DefFuture(\VarSelf) & \triangleq & \Receive\VarSelf\TagResolve\Var\invoke\DefResolvedFuture{\VarSelf,\Var}
      \\
      \DefResolvedFuture(\VarSelf, \Var) & \triangleq &
      \Free\VarSelf\Done
      \\ & \sumop &
      \Receive\VarSelf\TagGet\VarSender(\Send\VarSender\TagReply\Var \parop \invoke\DefResolvedFuture{\VarSelf,\Var})
      \\ & \sumop &
      \Receive\VarSelf\TagResolve{}
      \Fail\VarSelf
    \end{array}
  \]

  The process $\DefFuture$ represents an unresolved future variable,
  which waits for a $\TagResolve$ message from the producer. Once
  the variable has been resolved, it behaves as specified by
  $\DefResolvedFuture$, namely it satisfies an arbitrary number of
  $\TagGet$ messages from consumers but it no longer accepts
  $\TagResolve$ messages.
  \eoe
\end{example}

\begin{example}[bank account]
  \label{ex:account}
  Below we see the process definition corresponding to the actor
  shown in Listing~\ref{lst:account}. The structure of the term
  follows closely that of the Scala code:
  \[
    \begin{array}{@{}rcl@{}}
      \DefAccount(\VarSelf, \VarBalance) & \triangleq &
      \Receive\VarSelf\TagDebit{\VarAmount,\VarSender}
      \\ & &
      \Send\VarSender\TagReply{}
      \parop
      \invoke\DefAccount{\VarSelf, \VarBalance + \VarAmount}
      \\
      & \sumop &
      \Receive\VarSelf\TagCredit{\VarAmount,\VarRecipient,\VarSender}
      \\ & &
      \Send\VarRecipient\TagDebit{\VarAmount,\VarSelf} \parop {}
      \\ & &
      \Receive\VarSelf\TagReply{}
      \parens{
        \Send\VarSender\TagReply{}
        \parop
        \invoke\DefAccount{\VarSelf, \VarBalance + \VarAmount}
      }
      \\
      & \sumop &
      \Receive\VarSelf\TagStop{}
      \Free\VarSelf\Done
      \\
      & \sumop &
      \Receive\VarSelf\TagReply{}
      \Fail\VarSelf
    \end{array}
  \]

  The last term of the guarded process, which results in a failure,
  corresponds to the catch-all clause in Listing~\ref{lst:account}
  and models the fact that a $\TagReply$ message is not expected to
  be found in the account's mailbox unless the account is involved
  in a transaction. The $\TagReply$ message is received and handled
  appropriately in the $\TagCredit$-guarded term.

  We can model a deadlock in the case two distinct bank accounts
  attempt to initiate a transaction with one another. Indeed, we
  have
  \[
    \begin{lines}
      \invoke\DefAccount{\VarAlice, 10} \parop
      \Send\VarAlice\TagCredit{2, \VarCarol, \VarBank} \parop {}
      \\
      \invoke\DefAccount{\VarCarol, 15} \parop
      \Send\VarCarol\TagCredit{5, \VarAlice, \VarBank}
    \end{lines}
    \red^*
    \begin{lines}
      \Send\VarCarol\TagDebit{2,\VarAlice}
      \parop
      \Receive[]\VarAlice\TagReply{}\dots \parop {}
      \\
      \Send\VarAlice\TagDebit{5,\VarCarol}
      \parop
      \Receive[]\VarCarol\TagReply{}\dots
    \end{lines}
  \]
  where both $\VarAlice$ and $\VarCarol$ ignore the incoming
  $\TagDebit$ messages, whence the deadlock.
  \eoe
\end{example}

We now provide operational characterizations of the properties
enforced by our typing discipline. We begin with mailbox
conformance, namely the property that a process never fails because
of unexpected messages.  To this aim, we define a process context
$\ProcessContext$ as a process in which there is a single occurrence
of an \emph{unguarded hole} $\Hole$:
\[
  \ProcessContext
  ~~::=~~
  \Hole ~~\mid~~
  \ProcessContext \parop \Process ~~\mid~~
  \Process \parop \ProcessContext ~~\mid~~
  \New\Mailbox\ProcessContext
\]

The hole is ``unguarded'' in the sense that it does not occur
prefixed by an action. As usual, we write
$\ProcessContext[\Process]$ for the process obtained by replacing
the hole in $\ProcessContext$ with $\Process$. Names may be captured
by this replacement. A mailbox conformant process never reduces to a
state in which the only action of a guard is $\Fail{}$:

\begin{definition}
  \label{def:mailbox_conformance}
  We say that $\ProcessP$ is \emph{mailbox conformant} if
  $\ProcessP \nred^* \ProcessContext[\Fail\Mailbox]$ for all
  $\ProcessContext$ and $\Mailbox$.
\end{definition}

Looking at the placement of the $\Fail\Name$ actions in earlier
examples we can give the following interpretations of mailbox
conformance: a lock is never released unless it has been acquired
beforehand (Example~\ref{ex:lock}); a future variable is never
resolved twice (Example~\ref{ex:future}); an account will not be
notified of a completed transaction (with a $\TagReply$ message)
unless it is involved in an ongoing transaction
(Example~\ref{ex:account}).

We express deadlock freedom as the property that all irreducible
residuals of a process are (structurally equivalent to) the
terminated process:

\begin{definition}
  \label{def:deadlock_freedom}
  We say that $\ProcessP$ is \emph{deadlock free} if
  $\ProcessP \red^* \ProcessQ \nred$ implies
  $\ProcessQ \equiv \Done$.
\end{definition}

According to Definition~\ref{def:deadlock_freedom}, if a
deadlock-free process halts we have that: (1) there is no
sub-process waiting for a message that is never produced; (2) every
mailbox is empty. Clearly, this is not the case for the transaction
between $\VarAlice$ and $\VarCarol$ in Example~\ref{ex:account}.

\begin{example}[deadlock]
  \label{ex:deadlock}
  Below is another example of deadlocking process using $\DefFuture$
  from Example~\ref{ex:future}, obtained by resolving a future
  variable with the value it does not contain yet:
  \begin{equation}
    \label{eq:future_deadlock}
    \New\VarFuture
    \New\MailboxC\parens{
      \invoke\DefFuture\VarFuture
      \parop
      \Send\VarFuture\TagGet\MailboxC
      \parop
      \Receive\MailboxC\TagReply\Var
      \Free\MailboxC
      \Send\VarFuture\TagPut\Var
    }
  \end{equation}

  Notice that attempting to retrieve the content of a future
  variable not knowing whether it has been resolved is
  legal. Indeed, $\DefFuture$ does not fail if a $\TagGet$ message
  is present in the future variable's mailbox before it is
  resolved. Thus, the deadlocked process above is mailbox conformant
  but also an instance of undesirable process that will be ruled out
  by our static analysis technique (\cf
  Example~\ref{ex:future_deadlock}). We will need dependency graphs
  in addition to types to flag this process as ill typed.
  \eoe
\end{example}

A property stronger than deadlock freedom is fair termination.  A
fairly terminating process is a process whose residuals always have
the possibility to terminate. Formally:

\begin{definition}
  \label{def:terminating}
  We say that $\ProcessP$ is \emph{fairly terminating} if
  $\ProcessP \red^* \ProcessQ$ implies $\ProcessQ \red^* \Done$.
\end{definition}

An interesting consequence of fair termination is that it implies
\emph{junk freedom} (also known as lock
freedom~\cite{Kobayashi02,Padovani14B}) namely the property that
every message can be eventually consumed. Our type system does not
guarantee fair termination nor junk freedom in general, but it does
so for a non-trivial sub-class of well-typed processes that we
characterize later on.

%%% Local Variables:
%%% mode: latex
%%% TeX-master: "main"
%%% End:

\section{A Mailbox Type System}
\label{sec:type_system}

In this section we detail the type system for the mailbox
calculus. We start from the syntax and semantics of mailbox types
(Section~\ref{sec:types}) and of dependency graphs
(Section~\ref{sec:graphs}), the mechanism we use to track mailbox
dependencies. Then we present the typing rules
(Section~\ref{sec:rules}) and the properties of well-typed processes
(Section~\ref{sec:properties}).

% !TEX root = main.tex
\subsection{Mailbox Types}
\label{sec:types}

\begin{table}
  \begin{center}
    \[
      \begin{array}{rrcll}
        \textbf{Type} & \TypeT, \TypeS
        & ::= & \In\Pattern & \text{(input)}
        \\
        & & | & \Out\Pattern & \text{(output)}
        \\\\
        \textbf{Pattern} & \PatternE, \PatternF
        & ::= & \tzero & \text{(unreliable mailbox)}
        \\
        & & | & \tone & \text{(empty mailbox)}
        \\
        & & | & \tmessage\Tag\Types & \text{(atom)}
        \\
        & & | & \PatternE \tsum \PatternF & \text{(sum)}
        \\
        & & | & \PatternE \tmul \PatternF & \text{(product)}
        \\
        & & | & \Pattern\tstar & \text{(exponential)}
      \end{array}
    \]
  \end{center}
  \caption{\label{tab:types} Syntax of mailbox types and patterns.}
\end{table}

The syntax of mailbox types and patterns is shown in
Table~\ref{tab:types}.
Patterns are \emph{commutative regular expressions}~\cite{Conway71}
describing the configurations of messages stored in a mailbox.
An atom $\tmessage\Tag\Types$ describes a mailbox containing a
single message with tag $\Tag$ and arguments of type $\Types$. We
let $\MessageType$ range over atoms and abbreviate
$\tmessage\Tag\Types$ with $\Tag$ when $\Types$ is the empty
sequence.
Compound patterns are built using sum ($\PatternE \tsum \PatternF$),
product ($\PatternE \tmul \PatternF$) and exponential
($\Pattern\tstar$).
The constants $\tone$ and $\tzero$ respectively describe the empty
and the unreliable mailbox. There is no configuration of messages
stored in an unreliable mailbox, not even the empty one. We will use
the $\tzero$ pattern for describing mailboxes from which an
unexpected message has been received.
Let us look at a few simple examples.
The pattern $\TagA \tsum \TagB$ describes a mailbox that contains
either an $\TagA$ message or a $\TagB$ message, but not both,
whereas the pattern $\TagA \tsum \tone$ describes a mailbox that
either contains a $\TagA$ message or is empty.
The pattern $\TagA \tmul \TagB$ describes a mailbox that contains
both an $\TagA$ message and also a $\TagB$ message. Note that
$\TagA$ and $\TagB$ may be equal, in which case the mailbox contains
\emph{two} $\TagA$ messages.
Finally, the pattern $\TagA\tstar$ describes a mailbox that contains
an arbitrary number (possibly zero) of $\TagA$ messages.

A \emph{mailbox type} consists of a \emph{capability} (either $\In$
or $\Out$) paired with a pattern. The capability specifies whether
the pattern describes messages to be received from ($\In$) or stored
in ($\Out$) the mailbox. Here are some examples:
A process using a mailbox of type $\Out\TagA$ \emph{must} store an
$\TagA$ message into the mailbox, whereas a process using a mailbox
of type $\In\TagA$ is guaranteed to receive an $\TagA$ message from
the mailbox.
A process using a mailbox of type $\Out(\TagA \tsum \tone)$
\emph{may} store an $\TagA$ message into the mailbox, but is not
obliged to do so.
A process using a mailbox of type $\Out\parens{\TagA \tsum \TagB}$
decides whether to store an $\TagA$ message or a $\TagB$ message in
the mailbox, whereas a process using a mailbox of type
$\In\parens{\TagA \tsum \TagB}$ must be ready to receive both kinds
of messages.
A process using a mailbox of type $\In\parens{\TagA \tmul \TagB}$ is
guaranteed to receive both an $\TagA$ message and a $\TagB$ message
and may decide in which order to do so. A process using a mailbox of
type $\Out\parens{\TagA \tmul \TagB}$ must store both $\TagA$ and
$\TagB$ into the mailbox.
A process using a mailbox of type $\Out{\TagA\tstar}$ decides how
many $\TagA$ messages to store in the mailbox, whereas a process
using a mailbox of type $\In{\TagA\tstar}$ must be prepared to
receive an arbitrary number of $\TagA$ messages.

To cope with possibly infinite types we interpret the productions in
Table~\ref{tab:types} coinductively and consider as types the
regular trees~\cite{Courcelle83} built using those productions. We
require every infinite branch of a type tree to go through
infinitely many atoms. This strengthened contractiveness condition
allows us to define functions inductively on the structure of
patterns, provided that these functions do not recur into argument
types (\cf Definitions~\ref{def:subp} and~\ref{def:der}).

The semantics of patterns is given in terms of sets of multisets of
atoms. Because patterns include types, the given semantics is
parametric in the subtyping relation, which will be defined next:

\begin{definition}[subpattern]
  \label{def:subp}
  The \emph{configurations} of $\Pattern$ are inductively defined by
  the following equations, where $\BagA$ and $\BagB$ range over
  multisets $\bag{\seqof\MessageType}$ of atoms and $\bunion$
  denotes multiset union:
  \[
    \begin{array}[c]{@{}r@{~}c@{~}l@{}}
      \sem\tzero & \eqdef & \emptyset
      \\
      \sem\tone & \eqdef & \set{\bag{}}
    \end{array}
    \qquad
    \begin{array}[c]{@{}r@{~}c@{~}l@{}}
      \sem{\PatternE \tsum \PatternF} & \eqdef & \sem\PatternE \cup \sem\PatternF
      \\
      \sem{\PatternE \tmul \PatternF} & \eqdef & \set{ \BagA \bunion \BagB \mid \BagA \in \sem\PatternE, \BagB \in \sem\PatternF }
    \end{array}
    \qquad
    \begin{array}[c]{@{}r@{~}c@{~}l@{}}
      \sem\MessageType & \eqdef & \set{\bag\MessageType}
      \\
      \sem{\Pattern\tstar} & \eqdef & \sem\tone \cup \sem\Pattern \cup
      \sem{\Pattern\tmul\Pattern} \cup \cdots
    \end{array}
  \]

  Given a preorder relation $\rrel$ on types, we write
  $\PatternE \subp[\rrel] \PatternF$ if
  $\bag[i\in I]{\tmessage{\Tag_i}{\TypesT_i}} \in \sem\PatternE$
  implies
  $\bag[i\in I]{\tmessage{\Tag_i}{\TypesS_i}} \in \sem\PatternF$ and
  $\TypesT_i \rrel \TypesS_i$ for every $i\in I$.
  We write $\eqp[\rrel]$ for ${\subp[\rrel]} \cap {\supp[\rrel]}$.
\end{definition}

For example,
$\sem{\TagA \tsum \TagB} = \set{ \bag\TagA, \bag\TagB }$ and
$\sem{\TagA \tmul \TagB} = \set{ \bag{\TagA, \TagB} }$.
It is easy to see that $\subp[\rrel]$ is a pre-congruence with
respect to all the connectives and that it includes all the known
laws of commutative Kleene algebra~\cite{Conway71}: both $\tsum$ and
$\tmul$ are commutative and associative, $\tsum$ is idempotent and
has unit $\tzero$, $\tmul$ distributes over $\tsum$, it has unit
$\tone$ and is absorbed by $\tzero$.  Also observe that
$\subp[\rrel]$ is related covariantly to $\rrel$, that is
$\TypesT \rrel \TypesS$ implies
$\tmessage\Tag\TypesT \subp[\rrel]
\tmessage\Tag\TypesS$.

We now define subtyping. As types may be infinite, we resort to
coinduction:

\begin{definition}[subtyping]
  \label{def:subt}
  We say that $\rrel$ is a \emph{subtyping relation} if
  $\TypeT \rrel \TypeS$ implies either
  \begin{enumerate}
  \item\label{subt:input} $\TypeT = \In\PatternE$ and
    $\TypeS = \In\PatternF$ and $\PatternE \subp[\rrel] \PatternF$,
    or
  \item\label{subt:output} $\TypeT = \Out\PatternE$ and
    $\TypeS = \Out\PatternF$ and $\PatternF \subp[\rrel] \PatternE$.
  \end{enumerate}
  We write $\subt$ for the largest subtyping relation and say that
  $\TypeT$ is a \emph{subtype} of $\TypeS$ (and $\TypeS$ a
  \emph{supertype} of $\TypeT$) if $\TypeT \subt \TypeS$.
  We write $\eqt$ for ${\subt} \cap {\supt}$, $\subp$ for
  $\subp[\subt]$ and $\eqp$ for $\eqp[\subt]$.
\end{definition}

Items~\ref{subt:input} and~\ref{subt:output} respectively correspond
to the usual covariant and contravariant rules for channel types
with input and output capabilities~\cite{PierceSangiorgi96}.  For
example, $\Out\parens{\TagA \tsum \TagB} \subt \Out\TagA$ because a
mailbox of type $\Out\parens{\TagA \tsum \TagB}$ is more permissive
than a mailbox of type $\Out\TagA$. Dually,
$\In\TagA \subt \In\parens{\TagA \tsum \TagB}$ because a mailbox of
type $\In\TagA$ provides stronger guarantees than a mailbox of type
$\In\parens{\TagA \tsum \TagB}$.
Note that
$\Out\parens{\TagA \tmul \TagB} \eqt \Out\parens{\TagB \tmul \TagA}$
and
$\In\parens{\TagA \tmul \TagB} \eqt \In\parens{\TagB \tmul \TagA}$,
to witness the fact that the order in which messages are stored in a
mailbox is irrelevant.

Mailbox types whose patterns are in particular relations with the
constants $\tzero$ and $\tone$ will play special roles, so we
introduce some corresponding terminology.

\begin{definition}[type and name classification]
  \label{def:classification}
  We say that (a name whose type is) $\Type$ is:
  \begin{itemize}
  \item \emph{relevant} if $\Type \not\subt \Out\tone$ and
    \emph{irrelevant} otherwise;
  \item \emph{reliable} if $\Type \not\subt \In\tzero$ and
    \emph{unreliable} otherwise;
  \item \emph{usable} if $\Out\tzero \not\subt \Type$ and
    \emph{unusable} otherwise.
  \end{itemize}
\end{definition}

A relevant name \emph{must} be used, whereas an irrelevant name may
be discarded because not storing any message in the mailbox it
refers to is allowed by its type. All mailbox types with input
capability are relevant.
A reliable mailbox is one from which no unexpected message has been
received. All names with output capability are reliable.
A usable name \emph{can} be used, in the sense that there exists a
construct of the mailbox calculus that expects a name with that
type. All mailbox types with input capability are usable, but
$\In(\TagA \tmul \tzero)$ is unreliable.
Both $\Out\Tag[A]$ and $\Out\parens{\tone \tsum \Tag[A]}$ are
usable. The former type is also relevant because a process using a
mailbox with this type must (eventually) store an $\Tag[A]$ message
in it. On the contrary, the latter type is irrelevant, since not
using the mailbox is a legal way of using it.

\begin{center}
  \fbox{
    \parbox{0.95\textwidth}{%

      \textit{Henceforth we assume that all types are usable and
        that all argument types are also reliable. That is, we ban
        all types like $\Out\tzero$ or $\Out(\tzero\tmul\Tag)$ and
        all types like $\In\tmessage\Tag{\In\tzero}$ or
        $\Out\tmessage\Tag{\In\tzero}$.
        Example~\ref{ex:global_assumptions} in
        Appendix~\ref{sec:prop_subtyping} discusses the technical
        motivation for these assumptions.}

    }
  }
\end{center}

\begin{example}[lock type]
  \label{ex:lock_type}
  The mailbox used by the lock (Example~\ref{ex:lock}) will have
  several different types, depending on the viewpoint we take
  (either the lock itself or one of its users) and on the state of
  the lock (whether it is free or busy).
  As we can see from the definition of $\DefFreeLock$, a free lock
  waits for an $\TagAcquire$ message which is supposed to carry a
  reference to another mailbox into which the capability to release
  the lock is stored. Since the lock is meant to have several
  concurrent users, it is not possible in general to predict the
  number of $\TagAcquire$ messages in its mailbox.  Therefore, the
  mailbox of a free lock has type
  \[
    \In\tmessage\TagAcquire{\Out\tmessage\TagReply{\Out\TagRelease}}\tstar
  \]
  from the viewpoint of the lock itself. When the lock is busy, it
  expects to find one $\TagRelease$ message in its mailbox, but in
  general the mailbox will also contain $\TagAcquire$ messages
  corresponding to pending acquisition requests. So, the mailbox of
  a busy lock has type
  \[
    \In\parens{\TagRelease\tmul\tmessage\TagAcquire{\Out\tmessage\TagReply{\Out\TagRelease}}\tstar}
  \]
  indicating that the mailbox contains (or will eventually contain)
  a single $\TagRelease$ message along with arbitrarily many
  $\TagAcquire$ messages.

  Prospective owners of the lock may have references to the lock's
  mailbox with type
  $\Out\tmessage\TagAcquire{\Out\tmessage\TagReply{\Out\TagRelease}}$
  or
  $\Out\tmessage\TagAcquire{\Out\tmessage\TagReply{\Out\TagRelease}}\tstar$
  depending on whether they acquire the lock exactly once (just like
  $\VarAlice$ and $\VarCarol$ in Example~\ref{ex:lock}) or several
  times. Other intermediate types are possible in the case of users
  that acquire the lock a bounded number of times.
  The current owner of the lock will have a reference to the lock's
  mailbox of type $\Out\TagRelease$. This type is relevant, implying
  that the owner must eventually release the lock.
  \eoe
\end{example}

%%% Local Variables:
%%% mode: latex
%%% TeX-master: "main"
%%% End:

\subsection{Dependency Graphs}
\label{sec:graphs}

We use \emph{dependency graphs} for tracking dependencies between
mailboxes. Intuitively, there is a dependency between $\NameU$ and
$\NameV$ if either $\NameV$ is the argument of a message in mailbox
$\NameU$ or $\NameV$ occurs in the continuation of a process waiting
for a message from $\NameU$.
Dependency graphs have names as vertices and undirected
edges. However, the usual representation of graphs does not account
for the fact that mailbox names may be restricted and that the
multiplicity of dependencies matters. Therefore, we define
dependency graphs using the syntax below:
\[
  \textbf{Dependency Graph}
  \qquad
  \dgraphA, \dgraphB ~~::=~~
  \gempty ~~\mid~~
  \gedge\NameU\NameV ~~\mid~~
  \dgraphA \gunion \dgraphB ~~\mid~~
  \New\Mailbox\dgraph
\]

The term $\gempty$ represents the empty graph which has no vertices
and no edges.
The unordered pair $\gedge\NameU\NameV$ represents the graph made of
a single edge connecting the vertices $\NameU$ and $\NameV$.
The term $\dgraphA\gunion\dgraphB$ represents the union of
$\dgraphA$ and $\dgraphB$ whereas $\New\Mailbox\dgraph$ represents
the same graph as $\dgraph$ except that the vertex $\Mailbox$ is
restricted.
The usual notions of free and bound names apply to dependency
graphs. We write $\fn(\dgraph)$ for the free names of $\dgraph$.

\begin{table}
  \[
    \begin{array}{@{}c@{}}
      \inferrule{}{
        \gedge\NameU\NameV \gred\NameU\NameV \gempty
      }
      ~~\defrule{g-axiom}
      \qquad
      \inferrule{
        \dgraphA \gred\NameU\NameV \dgraphA'
      }{
        \dgraphA \gunion \dgraphB \gred\NameU\NameV \dgraphA' \gunion \dgraphB
      }
      ~~\defrule{g-left}
      \qquad
      \inferrule{
        \dgraphB \gred\NameU\NameV \dgraphB'
      }{
        \dgraphA \gunion \dgraphB \gred\NameU\NameV \dgraphA \gunion \dgraphB'
      }
      ~~\defrule{g-right}
      \\\\
      \inferrule{
        \dgraphA \gred\NameU\NameV \dgraphB
        \\
        \Mailbox\ne\NameU,\NameV
      }{
        \New\Mailbox\dgraphA \gred\NameU\NameV \New\Mailbox\dgraphB
      }
      ~~\defrule{g-new}
      \qquad
      \inferrule{
        \dgraphA \gred\NameU\NameW \dgraphB
        \\
        \dgraphB \gred\NameW\NameV \dgraphA'
      }{
        \dgraphA \gred\NameU\NameV \dgraphA'
      }
      ~~\defrule{g-trans}
    \end{array}
  \]
  \caption{\label{tab:dgraph} Labelled transitions of dependency graphs.}
\end{table}

To define the semantics of a dependency graph we use the labelled
transition system of Table~\ref{tab:dgraph}.  A label
$\glabel\NameU\NameV$ represents a path connecting $\NameU$ with
$\NameV$. So, a relation $\dgraphA \gred\NameU\NameV \dgraphA'$
means that $\NameU$ and $\NameV$ are connected in $\dgraphA$ and
$\dgraphA'$ describes the residual edges of $\dgraphA$ that have not
been used for building the path between $\NameU$ and $\NameV$.  The
paths of $\dgraph$ are built from the edges of $\dgraph$ (\cf
\refrule{g-axiom}) connected by shared vertices (\cf
\refrule{g-trans}). Restricted names cannot be observed in labels,
but they may contribute in building paths in the graph (\cf
\refrule{g-new}).

\begin{definition}[graph acyclicity and entailment]
  \label{def:grel}
  Let
  $\grel\dgraphA \eqdef \set{ \pair\NameU\NameV \mid
    \exists\dgraph': \dgraph \gred\NameU\NameV \dgraph' }$ be the
  \emph{dependency relation} generated by $\dgraph$.
  We say that $\dgraph$ is \emph{acyclic} if $\grel\dgraph$ is
  irreflexive.
  We say that $\dgraphA$ \emph{entails} $\dgraphB$, written
  $\dgraphA \gimplies \dgraphB$, if
  $\grel\dgraphB \subseteq \grel\dgraphA$.
\end{definition}

Note that $\gunion$ is commutative, associative and has $\gempty$ as
unit with respect to $\grel\cdot$ (see
Appendix~\ref{sec:prop_graphs}).
These properties of dependency graphs are key to prove that typing
is preserved by structural congruence on processes.
Note also that $\gunion$ is \emph{not} idempotent. Indeed,
$\gedge\NameU\NameV \gunion \gedge\NameU\NameV$ is cyclic whereas
$\gedge\NameU\NameV$ is not. The following example motivates the
reason why the multiplicity of dependencies is important.

\begin{example}
  \label{ex:multiplicity}
  Consider the reduction
  \[
    \ProcessP \eqdef
    \New\MailboxA
    \New\MailboxB\parens{
      \Send\MailboxA\TagA\MailboxB
      \parop
      \Send\MailboxA\TagB\MailboxB
      \parop
      \Receive\MailboxA\TagA\VarX
      \Receive\MailboxA\TagB\VarY
      \Free\MailboxA
      \Send\VarX\Tag\VarY
    }
    \red^*
    \New\MailboxB\Send\MailboxB\Tag\MailboxB
    \nred
  \]
  and observe that $\ProcessP$ stores two messages in the mailbox
  $\MailboxA$, each containing a reference to the mailbox
  $\MailboxB$. The two variables $\VarX$ and $\VarY$, which were
  syntactically different in $\ProcessP$, have been unified into
  $\MailboxB$ in the reduct, which is deadlocked.
  Unlike previous examples of deadlocked processes, which resulted
  from mutual dependencies between different mailboxes, in this case
  the deadlock is caused by the same dependency
  $\gedge\MailboxA\MailboxB$ arising twice.
  \eoe
\end{example}

%%% Local Variables:
%%% mode: latex
%%% TeX-master: "main"
%%% End:

%!TEX root = main.tex

\subsection{Typing Rules}
\label{sec:rules}

\begin{table}
  \begin{center}
    \[
      \begin{array}{@{}c@{}}
        \multicolumn{1}{@{}l@{}}{\textbf{Typing rules for processes}\hfill\framebox{$\wtp\Context\Process{\smash\dgraph}$}}
        \\\\
        \inferrule{
          % \mathstrut
        }{
          \wtp\EmptyContext\Done\gempty
        }
        ~~\defrule{t-done}
        \qquad
        \inferrule{
          \RecVar : (\Vars : \Types; \dgraph)
          % \\
          % \text{$\Types$ usable}
        }{
          \wtp{
            \Names : \Types
          }{
            \invoke\RecVar\Names
          }{
            \dgraph\subst\Names\Vars
          }
        }
        ~~\defrule{t-def}
        \qquad
        \inferrule{
          \wtp{\Context, \Mailbox : \In\tone}\Process\dgraph
        }{
          \wtp\Context{\New\Mailbox\Process}{\New\Mailbox\dgraph}
        }
        ~~\defrule{t-new}
        \\\\
        \inferrule{
          % \mathstrut
          % \text{$\Types$ usable}
        }{
          \wtp{
            \NameU : \Out{\tmessage\Tag\TypesT},
            \NamesV : \TypesT
          }{
            \Send\NameU\Tag\NamesV
          }{
            \gedge\NameU{\set\NamesV}
          }
        }
        ~~\defrule{t-msg}
        \qquad
        \inferrule{
          \wtg{
            \Name : \In\Pattern,
            \Context
          }{
            \Guard
          }
          \\
          \vDash \Pattern
        }{
          \textstyle
          \wtp{
            \Name : \In\Pattern,
            \Context
          }{
            \Guard
          }{
            \gedge\Name{\dom(\Context)}
          }
        }
        ~~\defrule{t-guard}
        \\\\
        \inferrule{
          \wtp{\Context_i}{\Process_i}{\dgraph_i}~{}^{(i=1,2)}
        }{
          \wtp{
            \Context_1 \cmul \Context_2
          }{
            \Process_1 \parop \Process_2
          }{
            \dgraph_1 \gunion \dgraph_2
          }
        }
        ~~\defrule{t-par}
        \qquad
        \inferrule{
          \wtp\ContextB\Process\dgraphB
          \\
          \ContextA \subt \ContextB
          \\
          \dgraphA \gimplies \dgraphB
        }{
          \wtp\ContextA\Process\dgraphA
        }
        ~~\defrule{t-sub}
        \\\\
        \multicolumn{1}{@{}l@{}}{\textbf{Typing rules for guards}\hfill\framebox{$\wtg\Context\Guard$}}
        \\\\
        \inferrule{
          % \mathstrut
        }{
          \wtg{
            \Name : \In\tzero,
            \Context
          }{
            \Fail\Name
          }
        }
        ~~\defrule{t-fail}
        \qquad
        \inferrule{
          \wtp\Context\Process\dgraph
        }{
          \wtg{
            \Name : \In\tone, \Context
          }{
            \Free\Name\Process
          }
        }
        ~~\defrule{t-free}
        \\\\
        \inferrule{
          \wtp{
            \Name : \In\PatternE,
            \Context,
            \Vars : \Types
          }{
            \Process
          }{
            \dgraph
          }
          % \\
          % \text{$\Types$ usable}
        }{
          \wtg{
            \Name : \In(\tmessage\Tag\Types\tmul\PatternE),
            \Context
          }{
            \Receive\Name\Tag\Vars\Process
          }
        }
        ~~\defrule{t-in}
        \qquad
        \inferrule{
          \wtg{
            \Name : \In\Pattern_i,
            \Context
          }{
            \Guard_i
          }
          ~{}^{(i=1,2)}
        }{
          \wtg{
            \Name : \In(\Pattern_1 \tsum \Pattern_2),
            \Context
          }{
            \Guard_1 \sumop \Guard_2
          }
        }
        ~~\defrule{t-branch}
      \end{array}
    \]
  \end{center}
  \caption{
    \label{tab:inference}
    Typing rules.
  }
\end{table}

We use \emph{type environments} for tracking the type of free names
occurring in processes. A type environment is a partial function
from names to types written as $\Names : \Types$ or
$\Name_1 : \Type_1, \dots, \Name_n : \Type_n$. We let $\Gamma$ and
$\Delta$ range over type environments, we write $\dom(\Gamma)$ for
the domain of $\Gamma$ and $\Gamma, \Delta$ for the union of
$\Gamma$ and $\Delta$ when
$\dom(\Gamma) \cap \dom(\Delta) = \emptyset$.
We say that $\Context$ is reliable if so are all the types in its
range.

Judgments for processes have the form $\wtp\Context\Process\dgraph$,
meaning that $\Process$ is well typed in $\Context$ and yields the
dependency graph $\dgraph$. Judgments for guards have the form
$\wtg\Context\Guard$, meaning that $\Guard$ is well typed in
$\Context$.
We say that a judgment $\wtp\Context\Process\dgraph$ is well formed
if $\fn(\dgraph) \subseteq \dom(\Context)$\Luca{Se si toglie
  l'entailment da \refrule{t-sub} questa condizione non serve pi\`u
  e si semplificano le regole e la descrizione. Considerare questa
  modifica dopo aver ricontrollato le prove.} and $\dgraph$ is
acyclic.  Each process typing rule has an implicit side condition
requiring that its conclusion is well formed.
For each global process definition $\define\RecVar\Vars\Process$ we
assume that there is a corresponding \emph{global process
  declaration} of the form $\RecVar : (\Vars : \Types; \dgraph)$.
We say that the definition is \emph{consistent} with the
corresponding declaration if $\wtp{\Vars : \Types}\Process\dgraph$.
Hereafter, all process definitions are assumed to be consistent.
We now discuss the typing rules in detail, introducing auxiliary
notions and notation as we go along.

\DiscussRule{Terminated process}
According to the rule \refrule{t-done}, the terminated process
$\Done$ is well typed in the empty type environment and yields no
dependencies. This is motivated by the fact that $\Done$ does not
use any mailbox. Later on we will introduce a subsumption rule
\refrule{t-sub} that allows us to type $\Done$ in any type
environment with irrelevant names.\Ugo{irrelevant types?}\Luca{\`E
  indifferente per la terminologia introdotta nella
  Definizione~\ref{def:classification}}

\DiscussRule{Message}
Rule \refrule{t-msg} establishes that a message
$\Send\NameU\Tag\NamesV$ is well typed provided that the mailbox
$\NameU$ allows the storing of an $\Tag$-tagged message with
arguments of type $\Types$ and the types of $\NamesV$ are indeed
$\Types$.
The subsumption rule \refrule{t-sub} will make it possible to use
arguments whose type is a \emph{subtype} of the expected ones.
A message $\Send\NameU\Tag\NamesV$ establishes dependencies between
the target mailbox $\NameU$ and all of the arguments $\NamesV$. We
write $\gedge\NameU{\set{\NameV_1,\dots,\NameV_n}}$ for the
dependency graph
$\gedge\NameU{\NameV_1}\gunion \cdots \gunion
\gedge\NameU{\NameV_n}$ and use $\gempty$ for the empty graph union.

\DiscussRule{Process invocation}
The typing rule for a process invocation $\invoke\RecVar\Names$
checks that there exists a global definition for $\RecVar$ which
expects exactly the given number and type of parameters. Again,
rule~\refrule{t-sub} will make it possible to use parameters whose
types are subtypes of the expected ones.
A process invocation yields the same dependencies as the
corresponding process definition, with the appropriate substitutions
applied.

\DiscussRule{Guards}
Guards are used to match the content of a mailbox and possibly
retrieve messages from it.
According to rule~\refrule{t-fail}, the action $\Fail\Name$ matches
a mailbox $\Name$ with type $\In\tzero$, indicating that an
unexpected message has been found in the mailbox. The type
environment may contain arbitrary associations, since the
$\Fail\Name$ action causes a runtime error.
Rule~\refrule{t-free} states that the action $\Free\Name\Process$
matches a mailbox $\Name$ with type $\In\tone$, indicating that the
mailbox is empty. The continuation is well typed in the residual
type environment $\Context$.
An input action $\Receive\Name\Tag\Vars\Process$ matches a mailbox
$\Name$ with type $\In(\tmessage\Tag\Types \tmul \PatternE)$ that
guarantees the presence of an $\Tag$-tagged message possibly along
with other messages as specified by $\PatternE$. The continuation
$\Process$ must be well typed in an environment where the mailbox
has type $\In\PatternE$, which describes the content of the mailbox
after the $\Tag$-tagged message has been removed. Associations for
the received arguments $\Vars$ are also added to the type
environment.
A compound guard $\Guard_1 \sumop \Guard_2$ offers the actions
offered by $\Guard_1$ and $\Guard_2$ and therefore matches a mailbox
$\Name$ with type $\In(\Pattern_1 \tsum \Pattern_2)$, where
$\Pattern_i$ is the pattern that describes the mailbox matched by
$\Guard_i$. Note that the residual type environment $\Context$ is
the same in both branches, indicating that the type of other
mailboxes used by the guard cannot depend on that of $\Name$.

The judgments for guards do not yield any dependency graph. This is
compensated by the rule~\refrule{t-guard}, which we describe next.

\DiscussRule{Guarded processes}
Rule~\refrule{t-guard} is used to type a guarded process $\Guard$,
which matches some mailbox $\Name$ of type $\In\Pattern$ and
possibly retrieves messages from it.
As we have seen while discussing guards, $\Pattern$ is supposed to
be a pattern of the form $\Pattern_1 \tsum \cdots \tsum \Pattern_n$
where each $\Pattern_i$ is either $\tzero$, $\tone$ or of the form
$\MessageType \tmul \PatternF$. However, only the patterns
$\Pattern$ that are in \emph{normal form} are suitable to be used in
this typing rule and the side condition $\df{}\Pattern$ checks that
this is indeed the case. We motivate the need of a normal form by
means of a simple example.

Suppose that our aim is to type a process
$\Receive\Name\TagA{}\ProcessP \sumop \Receive\Name\TagB{}\ProcessQ$
that consumes either an $\TagA$ message or a $\TagB$ message from
$\Name$, whichever of these two messages is matched first in
$\Name$, and then continues as $\ProcessP$ or $\ProcessQ$
correspondingly. Suppose also that the type of $\Name$ is
$\In\Pattern$ with
$\Pattern \eqdef \TagA\tmul\TagC \tsum \TagB\tmul\TagA$, which
allows the rules for guards to successfully type check the process.
As we have seen while discussing rule~\refrule{t-in}, $\ProcessP$
and $\ProcessQ$ must be typed in an environment where the type of
$\Name$ has been updated so as to reflect the fact that the consumed
message is no longer in the mailbox. In this particular case, we
might be tempted to infer that the type of $\Name$ in $\ProcessP$ is
$\In\TagC$ and that the type of $\Name$ in $\ProcessQ$ is
$\In\TagA$. Unfortunately, the type $\In\TagC$ does not accurately
describe the content of the mailbox after $\TagA$ has been consumed
because, according to $\Pattern$, the $\TagA$ message may be
accompanied by \emph{either} a $\TagB$ message \emph{or} by a
$\TagC$ message, whereas $\In\TagC$ only accounts for the second
possibility.
Thus, the appropriate pattern to be used for typing this process is
$\TagA \tmul (\TagB \tsum \TagC) \tsum \TagB \tmul \TagA$, where the
fact that $\TagB$ may be found after consuming $\TagA$ is made
explicit.  This pattern and $\Pattern$ are equivalent as they
generate exactly the same set of valid configurations. Yet,
$\TagA \tmul (\TagB \tsum \TagC) \tsum \TagB \tmul \TagA$ is in
normal form whereas $\Pattern$ is not.  In general the normal form
is not unique. For example, also the patterns
$\TagB \tmul \TagA \tsum \TagC \tmul \TagA$ and
$\TagA \tmul (\TagB \tsum \TagC)$ are in normal form and equivalent
to $\Pattern$ and can be used for typing processes that consume
messages from $\Name$ in different orders or with different
priorities.

The first ingredient for defining the notion of pattern normal form
is that of pattern residual $\pder\Pattern\MessageType$, which
describes the content of a mailbox that initially contains a
configuration of messages described by $\Pattern$ and from which we
remove a single message with type $\MessageType$:

\begin{definition}[pattern residual]
  \label{def:der}
  The \emph{residual} of a pattern $\Pattern$ with respect to an
  atom $\MessageType$, written $\pder\Pattern\MessageType$, is
  inductively defined by the following equations:
  \[
    \begin{array}{r@{~}c@{~}ll}
      \pder\tzero\MessageType = \pder\tone\MessageType & \eqdef & \tzero
      \\
      \pder{(\Pattern\tstar)}\MessageType & \eqdef & \pder\Pattern\MessageType \tmul \Pattern\tstar
    \end{array}
    ~
    \begin{array}{r@{~}c@{~}ll}
      \pder{\tmessage\Tag\TypesT}{\tmessage\Tag\TypesS} & \eqdef & \tone & \text{if $\TypesT \subt \TypesS$}
      \\
      \pder{\tmessage\Tag\TypesT}{\tmessage{\Tag'}\TypesS} & \eqdef & \tzero & \text{if $\Tag \ne \Tag'$}
    \end{array}
    ~
    \begin{array}{r@{~}c@{~}ll}
      \pder{(\PatternE \tsum \PatternF)}\MessageType & \eqdef & \pder\PatternE\MessageType \tsum \pder\PatternF\MessageType
      \\
      \pder{(\PatternE \tmul \PatternF)}\MessageType & \eqdef & \pder\PatternE\MessageType \tmul \PatternF \tsum \PatternE \tmul \pder\PatternF\MessageType
    \end{array}
  \]
\end{definition}

If we take the pattern $\PatternE$ discussed earlier we have
$\pder\Pattern\TagA = \tone \tmul \TagC \tsum \TagA \tmul \tzero
\tsum \tzero \tmul \tone \tsum \TagB \tmul \tone \eqp \TagB \tsum
\TagC$.
The pattern residual operator is closely related to Brzozowski's
derivative in a commutative Kleene
algebra~\cite{Brzozowski64,HopkinsKozen99}. Unlike Brzozowski's
derivative, the pattern residual is a partial operator:
$\pder\Pattern{\tmessage\Tag\TypesS}$ is defined provided that the
$\TypesS$ are supertypes of all types $\TypesT$ found in
$\Tag$-tagged atoms within $\Pattern$. This condition has a natural
justification: when choosing the message to remove from a mailbox
containing a configuration of messages described by $\Pattern$, only
the tag $\Tag$ of the message -- and not the type of its arguments
-- matters. Thus, $\TypesS$ faithfully describe the received
arguments provided that they are supertypes of \emph{all} argument
types of \emph{all} $\Tag$-tagged message types in $\Pattern$. For
example, assuming $\tnat \subt \tint$, we have that
$\pder{(\tmessage\Tag\tint \tsum
  \tmessage\Tag\tnat)}{\tmessage\Tag\tint}$ is defined whereas
$\pder{(\tmessage\Tag\tint \tsum
  \tmessage\Tag\tnat)}{\tmessage\Tag\tnat}$ is not.

We use the notion of pattern residual to define pattern normal
forms:

\begin{definition}[pattern normal form]
  \label{def:nf}
  We say that a pattern $\Pattern$ is in \emph{normal form}, written
  $\df{}\Pattern$, if $\df\Pattern\Pattern$ is derivable by the
  following axioms and rules:
  \[
    \df\Pattern\tzero
    \qquad
    \df\Pattern\tone
    \qquad
    \inferrule{
      \PatternF \eqp \pder\PatternE\MessageType
    }{
      \df\PatternE{\MessageType\tmul\PatternF}
    }
    \qquad
    \inferrule{
      \df{\PatternE}{\PatternF_1}
      \\
      \df{\PatternE}{\PatternF_2}
    }{
      \df\PatternE{\PatternF_1 \tsum \PatternF_2}
    }
  \]
\end{definition}

Essentially, the judgment $\df{}\PatternE$ verifies that $\PatternE$
is expressed as a sum of $\tzero$, $\tone$ and
$\MessageType\tmul\PatternF$ terms where $\PatternF$ is (equivalent
to) the residual of $\PatternE$ with respect to $\MessageType$.

A guarded process yields all the dependencies between the mailbox
$\Name$ being used and the names occurring free in the
continuations, because the process will not be able to exercise the
capabilities on these names until the message from $\Name$ has been
received.

\DiscussRule{Parallel composition}
Rule~\refrule{t-par} deals with parallel compositions of the form
$\Process_1 \parop \Process_2$. This rule accounts for the fact that
the same mailbox $\Name$ may be used in both $\Process_1$ and
$\Process_2$ according to different types.  For example,
$\Process_1$ might store an $\Tag[A]$ message into $\Name$ and
$\Process_2$ might store a $\Tag[B]$ message into $\Name$. In the
type environment for the parallel composition as a whole we must be
able to express with a single type the combined usages of $\Name$ in
$\Process_1$ and $\Process_2$.  This is accomplished by introducing
an operator that combines types:

\begin{definition}[type combination]
  \label{def:type_combination}
  We write $\TypeT \cmul \TypeS$ for the \emph{combination} of
  $\TypeT$ and $\TypeS$, where $\cmul$ is the partial symmetric
  operator defined as follows:
  \[
    \Out\PatternE \cmul \Out\PatternF \eqdef \Out\parens{\PatternE \tmul \PatternF}
    \qquad
    \Out\PatternE \cmul \In\parens{\PatternE \tmul \PatternF} \eqdef \In\PatternF
    \qquad
    \In\parens{\PatternE \tmul \PatternF} \cmul \Out\PatternE \eqdef \In\PatternF
  \]
\end{definition}

Continuing the previous example, we have
$\Out\TagA \cmul \Out\TagB = \Out(\TagA \tmul \TagB)$ because
storing one $\TagA$ message and one $\TagB$ message in $\Name$ means
storing an overall configuration of messages described by the
pattern $\TagA \tmul \TagB$.
When $\Name$ is used for both input and output operations, the
combined type of $\Name$ describes the overall balance of the
mailbox. For example, we have
$\Out\TagA \cmul \In(\TagA \tmul \TagB) = \In\TagB$: if we combine a
process that stores an $\TagA$ message into $\Name$ with another
process that consumes both an $\TagA$ message and a $\TagB$ message
from the same mailbox in some unspecified order, then we end up with
a process that consumes a $\TagB$ message from $\Name$.

Notice that $\cmul$ is a partial operator in that not all type
combinations are defined. It might be tempting to relax $\cmul$ in
such a way that
$\Out(\TagA \tmul \TagB) \cmul \In\TagA = \Out\TagB$, so as to
represent the fact that the combination of two processes results in
an excess of messages that must be consumed by some other process.
However, this would mean allowing different processes to consume
messages from the same mailbox, which is not safe in general (see
Example~\ref{ex:linear_input}).
For the same reason, the combination of $\In\PatternE$ and
$\In\PatternF$ is always undefined regardless of $\PatternE$ and
$\PatternF$.
Operators akin to $\cmul$ for the combination of channel types are
commonly found in substructural type systems for the (linear)
$\pi$-calculus~\cite{SangiorgiWalker01,Padovani14B}. Unlike these
systems, in our case the combination concerns also the content of a
mailbox in addition to the capabilities for accessing it.

\begin{example}
  \label{ex:linear_input}
  Suppose that we extend the type combination operator so that
  $\In\parens{\PatternE \tmul \PatternF} = \In\PatternE \cmul
  \In\PatternF$.  To see why this extension would be dangerous,
  consider the process
  \[
    \begin{array}{l}
      \bigparens{
        \Send\Name\TagA\True
        \parop
        \Receive\Name\TagA\VarX
        \parens{
          \Send\VarSystem\TagPrintBool\VarX
          \parop
          \Free\Name
          \Done
        }
      }
      \parop
      {}
      \\
      \bigparens{
        \Send\Name\TagA{\mkint2}
        \parop
        \Receive\Name\TagA\VarY
        \parens{
          \Send\VarSystem\TagPrintInt\VarY
          \parop
          \Free\Name
          \Done
        }
      }
    \end{array}
  \]

  Overall, this process stores into $\Name$ a combination of
  messages that matches the pattern
  $\tmessage\TagA\tbool \tmul \tmessage\TagA\tint$ and retrieves
  from $\Name$ the same combination of messages. Apparently, $\Name$
  is used in a balanced way.
  However, there is no guarantee that the $\Send\Name\TagA\True$
  message is received by the process at the top and that the
  $\Send\Name\TagA{\mkint2}$ message is received by the process at
  the bottom. In fact, the converse may happen because only the tag
  of a message -- not the type or value of its arguments -- is used
  for matching messages in the mailbox calculus.
  \eoe
\end{example}

We now extend type combination to type environments in the expected
way:

\begin{definition}[type environment combination]
  \label{def:cmul}
  We write $\ContextA \cmul \ContextB$ for the \emph{combination} of
  $\ContextA$ and $\ContextB$, where $\cmul$ is the partial operator
  inductively defined by the equations:
  \[
    \Context \cmul \EmptyContext \eqdef \Context
    \qquad
    \EmptyContext \cmul \Context \eqdef \Context
    \qquad
    (\Name : \TypeT, \ContextA) \cmul (\Name : \TypeS, \ContextB) \eqdef \Name : \TypeT \cmul \TypeS, (\ContextA \cmul \ContextB)
  \]
\end{definition}

With this machinery in place, rule~\refrule{t-par} is
straightforward to understand and the dependency graph of
$\Process_1 \parop \Process_2$ is simply the union of the dependency
graphs of $\Process_1$ and $\Process_2$.

% The
% (implicit) side condition on the judgment in the conclusion requires
% the resulting dependency graph to be acyclic.  In particular, a
% parallel composition of processes cannot yield multiple dependencies
% between the same pair of mailboxes (see
% Example~\ref{ex:multiplicity}).

\DiscussRule{Mailbox restriction}
Rule~\refrule{t-new} establishes that the process creating a new
mailbox $\Mailbox$ with scope $\Process$ is well typed provided that
the type of $\Mailbox$ is $\In\tone$. This means that every message
stored in the mailbox $\Mailbox$ by (a sub-process of) $\Process$ is
also consumed by (a sub-process of) $\Process$.
The dependency graph of the process is the same as that of
$\Process$, except that $\Mailbox$ is restricted.

\DiscussRule{Subsumption}
As we have anticipated earlier in a few occasions, the subsumption
rule \refrule{t-sub} allows us to rewrite types in the type
environment and to introduce associations for irrelevant names.  The
rule makes use of the following notion of subtyping for type
environments:

\begin{definition}[subtyping for type environments]
  We say that $\ContextA$ is a \emph{subtype environment} of
  $\ContextB$ if $\ContextA \subt \ContextB$, where $\subt$ is the
  least preorder on type environments such that:
  \[
    \inferrule{
      \mathstrut
    }{
      \Name : \Out\tone, \Context \subt \Context
    }
    \qquad
    \inferrule{
      \TypeT \subt \TypeS
    }{
      \Name : \TypeT, \Context \subt \Name : \TypeS, \Context
    }
  \]
\end{definition}

Intuitively, $\ContextA \subt \ContextB$ means that $\ContextA$
provides more capabilities than $\ContextB$. For example,
$\NameU : \Out(\Tag[A] \tsum \Tag[B]), \NameV : \Out\tone \subt
\NameU : \Out\Tag[A]$ since a process that is well typed in the
environment $\Name : \Out\Tag[A]$ stores an $\Tag[A]$ message into
$\Name$, which is also a valid behavior in the environment
$\NameU : \Out(\Tag[A] \tsum \Tag[B]), \NameV : \Out\tone$ where
$\NameU$ has more capabilities (it is also possible to store a
$\Tag[B]$ message into $\NameU$) and there is an irrelevant name
$\NameV$ not used by the process.

{ Rule~\refrule{t-sub} also allows us to replace the
  dependency graph yielded by $\Process$ with another one that
  generates a superset of dependencies. In general, the dependency
  graph should be kept as small as possible to minimize the
  possibility of yielding mutual dependencies (see
  \refrule{t-par}). The replacement allowed by \refrule{t-sub} is
  handy for technical reasons, but not necessary. The point is that
  the residual of a process typically yields fewer dependencies than
  the process itself, so we use \refrule{t-sub} to enforce the
  invariance of dependency graphs across reductions.  }

\begin{example}
  \label{ex:lock_derivation}
  We show the full typing derivation for $\DefFreeLock$ and
  $\DefBusyLock$ defined in Example~\ref{ex:lock}.  Our objective is
  to show the consistency of the global process declarations
  \[
    \DefFreeLock : (\VarSelf : \TypeT; \emptyset)
    \qquad
    \DefBusyLock : (\VarSelf : \TypeT, \VarOwner : \TypeR; \gedge\VarSelf\VarOwner)
  \]
  where $\TypeT \eqdef \In\tmessage\TagAcquire\TypeR\tstar$ and
  $\TypeR \eqdef \Out{\tmessage\TagReply{\Out\TagRelease}}$.  In the
  derivation trees below we rename $\VarSelf$ as $\VarX$ and
  $\VarOwner$ and $\VarY$ to resonably fit the derivations within
  the page limits. We start from the body of $\DefBusyLock$, which
  is simpler, and obtain
  \[
    \begin{prooftree}
      \[
        \justifies
        \wtp{
          \VarX : \Out\TagRelease,
          \VarY : \TypeR
        }{
          \Send\VarY\TagReply\VarX
        }{
          \gedge\VarY\VarX
        }
        \using\refrule{t-msg}
      \]
      \[
        \[
          \[
            \justifies
            \wtp{
              \VarX : \TypeT
            }{
              \invoke\DefFreeLock\VarX
            }{
              \gempty
            }
            \using\refrule{t-def}
          \]
          \justifies
          \wtg{
            \VarX : \TypeS
          }{
            \Receive\VarX\TagRelease{}
            \invoke\DefFreeLock\VarX
          }
          \using\refrule{t-in}
        \]
        \justifies
        \wtp{
          \VarX : \TypeS
        }{
          \Receive\VarX\TagRelease{}
          \invoke\DefFreeLock\VarX
        }{
          \gempty
        }
        \using\refrule{t-guard}
      \]
      \justifies
      \wtp{
        \VarX : \TypeT, \VarY : \TypeR
      }{
        \Send\VarY\TagReply\VarX
        \parop
        \Receive\VarX\TagRelease{}
        \invoke\DefLock\VarX
      }{
        \gedge\VarY\VarX \gunion \gempty
      }
      \using\refrule{t-par}
    \end{prooftree}
  \]
  where
  $\TypeS \eqdef \In{(\TagRelease \tmul
    \tmessage\TagAcquire\TypeR\tstar)}$.

  Concerning $\DefFreeLock$, the key step is rewriting the pattern
  of $\TypeT$ in a normal form that matches the branching structure
  of the process. To this aim, we use the property
  $\PatternE\tstar \eqp \tone \tsum \PatternE \tmul \PatternE\tstar$
  and the fact that $\tzero$ is absorbing for the product
  connective:
  \[
    \begin{prooftree}
      \[
        \[
          \qquad\qquad
          \vdots
          \qquad\qquad
          \[
            \[
              \[
                \justifies
                \wtg{
                  \VarX : \In\tzero
                }{
                  \Fail\VarX
                }
                \using\refrule{t-fail}
              \]
              \justifies
              \wtp{
                \VarX : \In\tzero
              }{
                \Fail\VarX
              }{
                \gempty
              }
              \using\refrule{t-guard}
            \]
            \using\refrule{t-guard}
            \justifies
            \wtg{
              \VarX : \In(\TagRelease\tmul\tzero)
            }{
              \Receive\VarX\TagRelease{}
              \Fail\VarX
            }
            \using\refrule{t-in}
          \]
          \justifies
          \wtg{
            \VarX :
            \In\parens{
              \tone
              \tsum
              \tmessage\TagAcquire\TypeR \tmul \tmessage\TagAcquire\TypeR\tstar
              \tsum
              \TagRelease \tmul \tzero
            }
          }{
            {}
            \cdots
            \sumop
            \Receive\VarX\TagRelease{}
            \Fail\VarX
          }
          \using\refrule{t-branch}
        \]
        \justifies
        \wtp{
          \VarX :
          \In\parens{
            \tone
            \tsum
            \tmessage\TagAcquire\TypeR \tmul \tmessage\TagAcquire\TypeR\tstar
            \tsum
            \TagRelease \tmul \tzero
          }
        }{
          {}
          \cdots
          \sumop
          \Receive\VarX\TagRelease{}
          \Fail\VarX
        }{
          \gempty
        }
        \using\refrule{t-guard}
      \]
      \justifies
      \wtp{
        \VarX : \TypeT
      }{
        \Free\VarX\Done
        \sumop
        \cdots
        \sumop
        \Receive\VarX\TagRelease{}
        \Fail\VarX
      }{
        \gempty
      }
      \using\refrule{t-sub}
    \end{prooftree}
  \]

  The elided sub-derivation concerns the first two branches of
  $\DefFreeLock$ and is as follows:
  \[
    \begin{prooftree}
      \[
        \[
          \justifies
          \wtp{
            \EmptyContext
          }{
            \Done
          }{
            \gempty
          }
          \using\refrule{t-done}
        \]
        \justifies
        \wtg{
          \VarX : \In\tone
        }{
          \Free\VarX\Done
        }
        \using\refrule{t-free}
      \]
      \[
        \[
          \justifies
          \wtp{
            \VarX : \TypeT,
            \VarY : \TypeR
          }{
            \invoke\DefBusyLock{\VarX,\VarY}
          }{
            \gedge\VarX\VarY
          }
          \using\refrule{t-def}
        \]
        \justifies
        \wtg{
          \VarX :
          \In\tmessage\TagAcquire\TypeR \tmul \tmessage\TagAcquire\TypeR\tstar
        }{
          \Receive\VarX\TagAcquire\VarY
          \invoke\DefBusyLock{\VarX,\VarY}
        }
        \using\refrule{t-in}
      \]
      \justifies
      \wtg{
        \VarX :
        \In\parens{
          \tone
          \tsum
          \tmessage\TagAcquire\TypeR \tmul \tmessage\TagAcquire\TypeR\tstar
        }
      }{
        \Free\VarX\Done
        \sumop
        \Receive\VarX\TagAcquire\VarY
        \invoke\DefBusyLock{\VarX,\VarY}
      }
      \using\refrule{t-branch}
    \end{prooftree}
  \]

  The process~\eqref{eq:alice_carol}, combining an instance of the
  lock and the users $\VarAlice$ and $\VarCarol$, is also well
  typed. As we will see at the end of Section~\ref{sec:properties},
  this implies that both $\VarAlice$ and $\VarCarol$ are able to
  acquire the lock, albeit in some unspecified order.
  \eoe
\end{example}

\begin{example}
  \label{ex:future_deadlock}
  In this example we show that the
  process~\eqref{eq:future_deadlock} of Example~\ref{ex:future} is
  ill typed. In order to do so, we assume the global process
  declaration
  \[
    \DefFuture : (\VarSelf : \In(\tmessage\TagPut\tint \tmul
    \tmessage\TagGet{\Out\tmessage\TagReply\tint}\tstar); \gempty)
  \]
  which can be shown to be consistent with the given definition for
  $\DefFuture$.
  In the derivation below we use the pattern
  $\PatternF \eqdef \tmessage\TagPut\tint \tmul
  \tmessage\TagGet\TypeR\tstar$ and the types
  $\TypeT \eqdef \Out\tmessage\TagPut\tint$,
  $\TypeS \eqdef \In\parens{\tmessage\TagReply\tint \tmul \tone}$
  and $\TypeR \eqdef \Out\tmessage\TagReply\tint$:
  \[
    \begin{prooftree}
      \[
        \[
          \[
            \justifies
            \wtp{
              \VarFuture : \Out\tmessage\TagGet\TypeR,
              \MailboxC : \TypeR
            }{
              \Send\VarFuture\TagGet\MailboxC
            }{
              \gedge\VarFuture\MailboxC
            }
            \using\refrule{t-msg}
          \]
          \[
            \[
              \[
                \[
                  \[
                    \justifies
                    \wtp{
                      \VarFuture : \TypeT, \Var : \tint
                    }{
                      \Send\VarFuture\TagPut\Var
                    }{
                      \gempty
                    }
                    \using\refrule{t-msg}
                  \]
                  \justifies
                  \wtg{
                    \VarFuture : \TypeT, \MailboxC : \In\tone, \Var : \tint
                  }{
                    \Free\MailboxC
                    \Send\VarFuture\TagPut\Var
                  }
                  \using\refrule{t-free}
                \]
                \justifies
                \wtp{
                  \VarFuture : \TypeT, \MailboxC : \In\tone, \Var : \tint
                }{
                  \Free\MailboxC
                  \Send\VarFuture\TagPut\Var
                }{
                  \gedge\MailboxC\VarFuture
                }
                \using\refrule{t-guard}
              \]
              \justifies
              \wtg{
                \VarFuture : \TypeT, \MailboxC : \TypeS
              }{
                \Receive\MailboxC\TagReply\Var
                \Free\MailboxC
                \Send\VarFuture\TagPut\Var
              }
              \using\refrule{t-in}
            \]
            \justifies
            \wtp{
              \VarFuture : \TypeT, \MailboxC : \TypeS
            }{
              \Receive\MailboxC\TagReply\Var
              \Free\MailboxC
              \Send\VarFuture\TagPut\Var
            }{
              \gedge\MailboxC\VarFuture
            }
            \using\refrule{t-guard}
          \]
          \justifies
          \wtp{
            \VarFuture : \Out\parens{
              \tmessage\TagGet\TypeR
              \tmul
              \tmessage\TagPut\tint
            },
            \MailboxC : \In\tone
          }{
            \Send\VarFuture\TagGet\MailboxC
            \parop
            \Receive\MailboxC\TagReply\Var
            \Free\MailboxC
            \Send\VarFuture\TagPut\Var
          }{
            {-}
          }
          \using\refrule{t-par}
        \]
        \justifies
        \wtp{
          \VarFuture : \Out\PatternF,
          \MailboxC : \In\tone
        }{
          \Send\VarFuture\TagGet\MailboxC
          \parop
          \Receive\MailboxC\TagReply\Var
          \Free\MailboxC
          \Send\VarFuture\TagPut\Var
        }{
          {-}
        }
        \using\refrule{t-sub}
      \]
      \justifies
      \wtp{
        \VarFuture : \Out\PatternF
      }{
        \New\MailboxC\parens{
          \Send\VarFuture\TagGet\MailboxC
          \parop
          \Receive\MailboxC\TagReply\Var
          \Free\MailboxC
          \Send\VarFuture\TagPut\Var
        }
      }{
        {-}
      }
      \using\refrule{t-new}
    \end{prooftree}
  \]

  In attempting this derivation we have implicitly extended the
  typing rules so that names with type $\tint$ do not contribute in
  generating any significant dependency. The critical point of the
  derivation is the application of \refrule{t-par}, where we are
  composing two parallel processes that yield a circular dependency
  between $\MailboxC$ and $\VarFuture$. In the process on the left
  hand side, the dependency $\gedge\VarFuture\MailboxC$ arises
  because $\MailboxC$ is sent as a reference in a message targeted
  to $\VarFuture$. In the process on the right hand side, the
  dependency $\gedge\MailboxC\VarFuture$ arises because there are
  guards concerning the mailbox $\MailboxC$ that block an output
  operation on the mailbox $\VarFuture$.
  \eoe
\end{example}

\begin{example}[non-deterministic choice]
  \label{ex:choice}
  Different input actions in the same guard can match messages with
  the same tag. This feature can be used to encode in the mailbox
  calculus the non-deterministic choice between $\Process_1$ and
  $\Process_2$ as the process
  \begin{equation}
    \label{eq:choice}
    \New\Mailbox\parens{
      \Send\Mailbox\Tag{}
      \parop
      \Receive\Mailbox\Tag{}
      \Free\Mailbox
      \Process_1
      \sumop
      \Receive\Mailbox\Tag{}
      \Free\Mailbox
      \Process_2
    }
  \end{equation}
  provided that $\wtp\Context{\Process_i}{\dgraph_i}$ for $i=1,2$.
  That is, $\Process_1$ and $\Process_2$ must be well typed in the
  same type environment. Below is the typing derivation for
  \eqref{eq:choice}
  \[
    \begin{prooftree}
      \[
        \[
          \justifies
          \wtp{
            \Mailbox : \Out\Tag
          }{
            \Send\Mailbox\Tag{}
          }{
            \gempty
          }
          \using\refrule{t-msg}
        \]
        \[
          \[
            \[
              \[
                \[
                  \[
                    \wtp\Context{
                      \Process_i
                    }{
                      \dgraph_i
                    }
                    \justifies
                    \wtg{
                      \Context, \Mailbox : \In\tone
                    }{
                      \Free\Mailbox
                      \Process_i
                    }
                    \using\refrule{t-free}
                  \]
                  \justifies
                  \wtp{
                    \Context, \Mailbox : \In\tone
                  }{
                    \Free\Mailbox
                    \Process_i
                  }{
                    \dgraph
                  }
                  \using\refrule{t-guard}
                \]
                \justifies
                \wtg{
                  \Context, \Mailbox : \In(\Tag \tmul \tone)
                }{
                  \Receive\Mailbox\Tag{}
                  \Free\Mailbox
                  \Process_i
                }
                \using\refrule{t-in}, i=1,2
              \]
              \justifies
              \wtg{
                \Context, \Mailbox : \In(\Tag \tmul \tone \tsum \Tag \tmul \tone)
              }{
                \Receive\Mailbox\Tag{}
                \Free\Mailbox
                \Process_1
                \sumop
                \Receive\Mailbox\Tag{}
                \Free\Mailbox
                \Process_2
              }
              \using\refrule{t-branch}
            \]
            \justifies
            \wtp{
              \Context, \Mailbox : \In(\Tag \tmul \tone \tsum \Tag \tmul \tone)
            }{
              \Receive\Mailbox\Tag{}
              \Free\Mailbox
              \Process_1
              \sumop
              \Receive\Mailbox\Tag{}
              \Free\Mailbox
              \Process_2
            }{
              \dgraph
            }
            \using\refrule{t-guard}
          \]
          \justifies
          \wtp{
            \Context, \Mailbox : \In(\Tag \tmul \tone)
          }{
            \Receive\Mailbox\Tag{}
            \Free\Mailbox
            \Process_1
            \sumop
            \Receive\Mailbox\Tag{}
            \Free\Mailbox
            \Process_2
          }{
            \dgraph
          }
          \using\refrule{t-sub}
        \]
        \justifies
        \wtp{
          \Context, \Mailbox : \In\tone
        }{
          \Send\Mailbox\Tag{}
          \parop
          \Receive\Mailbox\Tag{}
          \Free\Mailbox
          \Process_1
          \sumop
          \Receive\Mailbox\Tag{}
          \Free\Mailbox
          \Process_2
        }{
          \dgraph
        }
        \using\refrule{t-par}
      \]
      \justifies
      \wtp\Context{
        \New\Mailbox\parens{
          \Send\Mailbox\Tag{}
          \parop
          \Receive\Mailbox\Tag{}
          \Free\Mailbox
          \Process_1
          \sumop
          \Receive\Mailbox\Tag{}
          \Free\Mailbox
          \Process_2
        }
      }{
        \New\Mailbox\dgraph
      }
      \using\refrule{t-new}
    \end{prooftree}
  \]
  where $\dgraph \eqdef \gedge\Mailbox{\dom(\Context)}$.  The key
  step is the application of \refrule{t-sub}, which exploits the
  idempotency of $\tsum$ (in patterns) to rewrite $\Tag \tmul \tone$
  as the equivalent pattern
  $\Tag \tmul \tone \tsum \Tag \tmul \tone$.
  \eoe
\end{example}

%%% Local Variables:
%%% mode: latex
%%% TeX-master: "main"
%%% End:

\subsection{Properties of well-typed processes}
\label{sec:properties}

In this section we state the main properties enjoyed by well-typed
processes. As usual, subject reduction is instrumental for all of
the results that follow as it guarantees that typing is preserved by
reductions:

\begin{theorem}
  \label{thm:sr}
  If $\Context$ is reliable and $\wtp\Context\ProcessP\dgraph$ and
  $\ProcessP \red \ProcessQ$, then $\wtp\Context\ProcessQ\dgraph$.
\end{theorem}

Interestingly, Theorem~\ref{thm:sr} seems to imply that the types of
the mailboxes used by a process do not change.  In sharp contrast,
other popular behavioral typing disciplines (session types in
particular), are characterized by a subject reduction result in
which types reduce along with processes.
Theorem~\ref{thm:sr} also seems to contradict the observations made
earlier concerning the fact that the mailboxes used by a process may
have different types (Example~\ref{ex:lock_type}).
The type preservation guarantee assured by Theorem~\ref{thm:sr} can
be explained by recalling that the type environment $\Context$ in a
judgment $\wtp\Context\ProcessP\dgraph$ already takes into account
the overall balance between the messages stored into and consumed
from the mailbox used by $\ProcessP$ (see
Definition~\ref{def:cmul}). In light of this observation,
Theorem~\ref{thm:sr} simply asserts that well-typed processes are
steady state: they never produce more messages than those that are
consumed, nor do they ever try to consume more messages than those
that are produced.

A practically relevant consequence of Theorem~\ref{thm:sr} is that,
by looking at the type $\In\PatternE$ of the mailbox $\Mailbox$ used
by a guarded process $\Process$ (rule~\refrule{t-guard}), it is
possible to determine \emph{bounds} to the number of messages that
can be found in the mailbox as $\Process$ waits for a message to
receive.  In particular, if every configuration of $\PatternE$
contains at most $k$ atoms with tag $\Tag$, then at runtime
$\Mailbox$ contains at most $\Tag$-tagged messages. As a special
case, a mailbox of type $\In\tone$ is guaranteed to be empty and can
be statically deallocated. Note that the bounds may change after
$\Process$ receives a message. For example, a free lock is
guaranteed to have no $\TagRelease$ messages in its mailbox, and
will have at most one when it is busy (see
Example~\ref{ex:lock_derivation}).

The main result concerns the soundness of the type system,
guaranteeing that well-typed (closed) processes are both mailbox
conformant and deadlock free:

\begin{theorem}
  \label{thm:soundness}
  If $\wtp\EmptyContext\ProcessP\dgraph$, then $\ProcessP$ is
  mailbox conformant and deadlock free.
\end{theorem}

Fair termination and junk freedom are not guaranteed by our typing
discipline in general. The usual counterexamples include processes
that postpone indefinitely the use of a mailbox with a relevant
type. For instance, the $\Tag$ message in the well-typed process
$\New\Mailbox\parens{ \Send\Mailbox\Tag{} \parop
  \invoke\RecVar\Mailbox }$ where
$\define\RecVar\Var{\invoke\RecVar\Var}$ is never consumed because
$\Mailbox$ is never used for an input operation.

Nevertheless, fair termination is guaranteed for the class of
finitely unfolding processes:

\begin{theorem}
  \label{thm:fair_termination}
  We say that $\ProcessP$ is \emph{finitely unfolding} if all
  maximal reductions of $\ProcessP$ use \refrule{r-def} finitely
  many times.  If $\wtp\EmptyContext\ProcessP\dgraph$ and
  $\ProcessP$ is finitely unfolding, then $\ProcessP$ is fairly
  terminating.
\end{theorem}

The class of finitely unfolding processes obviously includes all
finite processes (those not using process invocations) but also many
recursive processes. For example, every process of the form
$\New\Mailbox\parens{ \Send\Mailbox\Tag{} \parop \cdots \parop
  \Send\Mailbox\Tag{} \parop \invoke\RecVar\Mailbox }$ where
$\RecVar(\Var) \triangleq \Receive\Var\Tag{}\invoke\RecVar\Var
\sumop \Free\Var\Done$ is closed, well typed and finitely unfolding
regardless of the number of $\Tag$ messages stored in $\Mailbox$,
hence is fairly terminating and junk free by
Theorem~\ref{thm:fair_termination}.

%%% Local Variables:
%%% mode: latex
%%% TeX-master: "main"
%%% End:

%%% Local Variables:
%%% mode: latex
%%% TeX-master: "main"
%%% End:

\newcommand{\VarPool}{\Var[pool]}

\section{Examples}
\label{sec:examples}

In this section we discuss a few more examples that illustrate the
expressiveness of the mailbox calculus and of its type system. We
consider a variant of the bank account shown in
Listing~\ref{lst:account} (Section~\ref{sec:example_account_await}),
the case of master-workers parallelism
(Section~\ref{sec:example_master_workers}) and the encoding of
binary sessions extended with forks and joins
(Sections~\ref{sec:example_binary_sessions}
and~\ref{sec:example_forks_joins}).

\newcommand{\VarFutureLong}{\Var[future]}

\InputScala[label=lst:account_await,caption=An Akka actor using futures from the \texttt{Savina} benchmark suite~\cite{ImamSarkar14}.]{AccountAkkaAwait.scala}

\subsection{Actors using futures}
\label{sec:example_account_await}

Many Scala programs combine actors with
futures~\cite{TasharofiDingesJohnson13}.
As an example, Listing~\ref{lst:account_await} shows an alternative
version of the \SI{Account} actor in Akka that differes from
Listing~\ref{lst:account} in the handling of \SI{CreditMessages}
(lines~\ref{account.future.begin}--\ref{account.future.end}). The
\SI{future} variable created here is initialized asynchronously with
the result of the debit operation invoked on \SI{recipient}. To make
sure that each transaction is atomic, the actor waits for the
variable to be resolved (line~\ref{account.future.end}) before
notifying \SI{sender} that the operation has been completed.

This version of \SI{Account} is arguably simpler than the one in
Listing~\ref{lst:account}, if only because the actor has a unique
top-level behavior. One way of modeling this implementation of
\SI{Account} in the mailbox calculus is to use $\DefFuture$,
discussed in Example~\ref{ex:future}.
A simpler modeling stems from the observation that \SI{future} in
Listing~\ref{lst:account_await} is used for a one-shot
synchronization.  A future variable with this property is akin to a
mailbox from which the value of the resolved variable is retrieved
exactly once. Following this approach we obtain the process below:
\[
  \begin{array}{@{}rcl@{}}
    \DefAccount(\VarSelf, \VarBalance) & \triangleq &
    \Receive\VarSelf\TagDebit{\VarAmount,\VarSender}
    \\ & &
    \Send\VarSender\TagReply{}
    \parop
    \invoke\DefAccount{\VarSelf, \VarBalance + \VarAmount}
    \\
    & \sumop &
    \Receive\VarSelf\TagCredit{\VarAmount,\VarRecipient,\VarSender}
    \\ & &
    \New\VarFutureLong\bigparens{
      \begin{lines}
        \Send\VarRecipient\TagDebit{\VarAmount,\VarFutureLong}
        \parop {}
        \\
        \Receive\VarFutureLong\TagReply{}
        \Free\VarFutureLong
        \\
        \parens{
          \Send\VarSender\TagReply{}
          \parop
          \invoke\DefAccount{\VarSelf, \VarBalance + \VarAmount}
        }
      \end{lines}
    }
    \\
    & \sumop &
    \Receive\VarSelf\TagStop{}
    \Free\VarSelf\Done
    \\
    & \sumop &
    \Receive\VarSelf\TagReply{}
    \Fail\VarSelf
  \end{array}
\]

Compared to the process in Example~\ref{ex:account}, here the
notification from the $\VarRecipient$ account is received from the
mailbox $\VarFutureLong$, which is created locally during the
handling of the $\TagCredit$ message. The rest of the process is the
same as before. This definition of $\DefAccount$ and the one in
Example~\ref{ex:account} can both be shown to be consistent with the
declaration
\[
  \DefAccount : (\VarSelf :
  \In\parens{
    \tmessage\TagDebit{\tint,\TypeR}\tstar
    \tmul
    \tmessage\TagCredit{\tint,\Out\tmessage\TagDebit{\tint,\TypeR},\TypeR}\tstar
    \tsum
    \TagStop}, \VarBalance : \tint; \gempty)
\]
where $\TypeR \eqdef \Out\TagReply$.  In particular, the
dependencies between $\VarSelf$ and $\VarFutureLong$ that originate
in this version of $\DefAccount$ are not observable from outside
$\DefAccount$ itself.

The use of multiple mailboxes and the interleaving of blocking
operations on them may increase the likelyhood of programming
mistakes causing mismatched communications and/or
deadlocks. However, these errors can be detected by a suitable
typing discipline such the one proposed in this paper.
Types can also be used to mitigate the runtime overhead resulting
from the use of multiple mailboxes. Here, for example, the typing of
$\VarFutureLong$ guarantees that this mailbox is used for receiving
a \emph{single} message and that $\VarFutureLong$ is empty by the
time $\Free[]\VarFutureLong$ is performed. A clever compiler can
take advantage of this information to statically optimize both the
allocation and the deallocation of this mailbox.

%%% Local Variables:
%%% mode: latex
%%% TeX-master: "main"
%%% End:

% \input{example_quicksort}

\subsection{Master-workers parallelism}
\label{sec:example_master_workers}

\newcommand{\DefAvailable}{\RecVar[Available]}
\newcommand{\DefWorker}{\RecVar[Worker]}
\newcommand{\DefCreatePool}{\RecVar[CreatePool]}
\newcommand{\DefCollectResults}{\RecVar[CollectResults]}
\newcommand{\VarWorker}{\Var[worker]}
\newcommand{\TagTask}{\Tag[task]}
\newcommand{\TagWork}{\Tag[work]}
\newcommand{\TagOne}{\Tag[one]}

In this example we model a \emph{master} process that receives tasks
to perform from a \emph{client}. For each task, the master creates a
pool of \emph{workers} and assigns each worker a share of work.  The
master waits for all partial results from the workers before sending
the final result back to the client and making itself available
again.
The number of workers may depend on some quantity possibly related
to the task to be performed and that is known at runtime only.

Below we define three processes corresponding to the three states in
which the master process can be, and we leave $\DefWorker$
unspecified:
\[
  \begin{array}{@{}r@{~}c@{~}l@{}}
    \DefAvailable(\VarSelf) & \triangleq &
    \Receive\VarSelf\TagTask\VarClient
    \New\VarPool\invoke\DefCreatePool{\VarSelf,\VarPool,\VarClient}
    \\
    & \sumop & \Free\VarSelf\Done
    \\
    \DefCreatePool(\VarSelf,\VarPool,\VarClient) & \triangleq &
    \IfX{~\textit{more workers needed}~}{
      \\ & & \quad
      \New\VarWorker(
      \Send\VarWorker\TagWork\VarPool
      \parop
      \invoke\DefWorker\VarWorker)
      \parop
      {}
      \\ & & \quad
      \invoke\DefCreatePool{\VarSelf,\VarPool,\VarClient}
      \\ & &
    }{
      \\ & & \quad
      \invoke\DefCollectResults{\VarSelf, \VarPool, \VarClient}
    }
    \\
    \DefCollectResults(\VarSelf,\VarPool,\VarClient) & \triangleq &
    \Receive\VarPool\TagResult{}
    \invoke\DefCollectResults{\VarSelf,\VarPool,\VarClient}
    \\
    & \sumop &
    \Free\VarPool\parens{
      \Send\VarClient\TagResult{}
      \parop
      \invoke\DefAvailable\VarSelf
    }
  \end{array}
\]

The ``$\If{\mathit{condition}}\ProcessP\ProcessQ$'' form used here
can be encoded in the mailbox calculus and is typed similarly to the
non-deterministic choice of Example~\ref{ex:choice}.
These definitions can be shown to be consistent with the following
declarations:
\[
  \begin{array}{@{}r@{~}c@{~}l@{}}
    \DefAvailable & : & (\VarSelf : \In\tmessage\TagTask{\Out\TagResult}\tstar; \gempty)
    \\
    \DefCreatePool, \DefCollectResults & : & (\VarSelf : \In\tmessage\TagTask{\Out\TagResult}\tstar,
    \VarPool : \In\TagResult\tstar, \VarClient : \Out\TagResult;
    \\
    & & \phantom{(}
    \gedge\VarPool\VarSelf \gunion \gedge\VarPool\VarClient)
  \end{array}
\]

The usual implementation of this coordination pattern requires the
programmer to keep track of the number of active workers using a
counter that is decremented each time a partial result is
collected~\cite{ImamSarkar14}. When the counter reaches zero, the
master knows that all the workers have finished their job and
notifies the client. In the mailbox calculus, we achieve the same
goal by means of a dedicated mailbox $\VarPool$ from which the
partial results are collected: when $\VarPool$ becomes disposable,
it means that no more active workers remain.\Luca{Questo esempio
  potrebbe essere criticato in quanto l'azione $\Free[]{}$ sembra
  implicare l'uso di un complesso garbage collector distribuito. Il
  garbage collector pu\`o essere realizzato in maniera efficiente,
  ma \`e difficile da spiegare e dunque poco convincente. Questo \`e
  anche l'unico esempio in cui $\Free[]{}$ \`e seguita da un'azione
  interessante e non da $\Done$. Usando una sola mailbox si ricade
  nel problema del type pollution.}

%%% Local Variables:
%%% mode: latex
%%% TeX-master: "main"
%%% End:

\subsection{Encoding of binary sessions}
\label{sec:example_binary_sessions}

\newcommand{\SessionProcess}[1]{\RecVar[Session]_{#1}}
\newcommand{\CreateSessionProcess}[1]{\RecVar[Create]_{#1}}
\newcommand{\EncodeSessionType}[1]{\mathcal{E}(#1)}
\newcommand{\TagOp}{\Tag[op]}

Session types~\cite{Honda93,HuttelEtAl16} have become a popular
formalism for the specification and enforcement of structured
protocols through static analysis.  A session is a private
communication channel shared by processes that interact through one
of its \emph{endpoint}. Each endpoint is associated with a
\emph{session type} that specifies the type, direction and order of
messages that are supposed to be exchanged through that endpoint. A
typical syntax for session types in the case of \emph{binary
  sessions} (those connecting exactly two peer processes) is shown
below:
\[
  \SessionTypeT, \SessionTypeS
  ~~::=~~
  \End
  ~~\mid~~
  \In\tmessage{}\Type.\SessionType
  ~~\mid~~
  \Out\tmessage{}\Type.\SessionType
  ~~\mid~~
  \SessionTypeT \SBranch \SessionTypeS
  ~~\mid~~
  \SessionTypeT \SChoice \SessionTypeS
\]

A session type $\In\tmessage{}\Type.\SessionType$ describes an
endpoint used for receiving a message of type $\Type$ and then
according to $\SessionType$.  Dually, a session type
$\Out\tmessage{}\Type.\SessionType$ describes an endpoint used for
sending a message of type $\Type$ and then according to
$\SessionType$.
An external choice $\SessionTypeT \SBranch \SessionTypeS$ describes
an endpoint used for receiving a selection (either $\TagLeft$ or
$\TagRight$) and then according to the corresponding continuation
(either $\SessionTypeT$ or $\SessionTypeS$). Dually, an internal
choice $\SessionTypeT \SChoice \SessionTypeS$ describes an endpoint
used for making a selection and then according to the corresponding
continuation.
Communication safety and progress of a binary session are guaranteed
by the fact that its two endpoints are linear resources typed by
\emph{dual} session types, where the dual of $\SessionType$ is
obtained by swapping inputs with outputs and internal with external
choices.

In this example we encode sessions and session types using mailboxes
and mailbox types.  We encode a session as a non-uniform, concurrent
object. The object is ``concurrent'' because it is accessed
concurrently by the two peers of the session. It is ``non-uniform''
because its interface changes over time, as the session
progresses. The object uses a mailbox $\VarSelf$ and its behavior is
defined by the equations for $\SessionProcess\SessionType(\VarSelf)$
shown below, where $\SessionType$ is the session type according to
which it must be used by one of the peers:
\[
  \begin{array}{@{}r@{~}c@{~}l@{}}
    \SessionProcess\End(\VarSelf) & \triangleq & \Free\VarSelf\Done
    \\
    \SessionProcess{
      \In\tmessage{}\Type.\SessionType
    }(\VarSelf) & \triangleq &
    \Receive\VarSelf\TagSend{x,s}
    \Receive\VarSelf\TagReceive{r}
    \\ & &
    \parens{
      \Send{s}\TagReply{x,\VarSelf}
      \parop
      \Send{r}\TagReply\VarSelf
      \parop
      \invoke{\SessionProcess\SessionType}\VarSelf
    }
    \\
    \SessionProcess{\Out\tmessage{}\Type.\SessionType}(\VarSelf) & \triangleq &
    \invoke{\SessionProcess{\In\tmessage{}\Type.\SessionType}}\VarSelf
    \\
    \SessionProcess{\SessionTypeT \SBranch \SessionTypeS}(\VarSelf) & \triangleq &
    \Receive\VarSelf\TagLeft{s}
    \Receive\VarSelf\TagReceive{r}
    \\ & &
    \parens{
      \Send{s}\TagReply\VarSelf
      \parop
      \Send{r}\TagLeft\VarSelf
      \parop
      \invoke{\SessionProcess\SessionTypeT}\VarSelf
    }
    \\ & \sumop &
    \Receive\VarSelf\TagRight{s}
    \Receive\VarSelf\TagReceive{r}
    \\ & &
    \parens{
      \Send{s}\TagReply\VarSelf
      \parop
      \Send{r}\TagRight\VarSelf
      \parop
      \invoke{\SessionProcess\SessionTypeS}\VarSelf
    }
    \\
    \SessionProcess{\SessionTypeT \SChoice \SessionTypeS}(\VarSelf) & \triangleq &
    \invoke{\SessionProcess{\SessionTypeT \SBranch \SessionTypeS}}\VarSelf
  \end{array}
\]

To grasp the intuition behind the definition of
$\SessionProcess\SessionType(\VarSelf)$, it helps to recall that
each stage of a session corresponds to an interaction between the
two peers, where one process plays the role of ``sender'' and its
peer that of ``receiver''. Both peers manifest their willingness to
interact by storing a message into the session's mailbox. The
receiver always stores a $\TagReceive$ message, while the sender
stores either $\TagSend$, $\TagLeft$ or $\TagRight$ according to
$\SessionType$. All messages contain a reference to the mailbox
owned by sender and receiver (respectively $s$ and $r$) where they
will be notified once the interaction is completed. A $\TagSend$
message also carries actual payload $\Var$ being exchanged. The role
of $\SessionProcess\SessionType(\VarSelf)$ is simply to forward each
message from the sender to the receiver. The notifications stored in
$s$ and $r$ contain a reference to the session's mailbox so that its
type reflects the session's updated interface corresponding to the
rest of the conversation.

Interestingly, the encoding of a session with type $\SessionType$ is
undistinguishable from that of a session with the dual type
$\co\SessionType$. This is natural by recalling that each stage of a
session corresponds to a single interaction between the two peers:
the order in which they store the respective messages in the
session's mailbox is in general unpredictable but also unimportant,
for both messages are necessary to complete each interaction.

\newcommand{\DefAlice}{\RecVar[Alice]}
\newcommand{\DefCarol}{\RecVar[Carol]}

As an example, suppose we want to model a system where $\DefAlice$
asks $\DefCarol$ to compute the sum of two numbers exchanged through
a session $s$. $\DefAlice$ and $\DefCarol$ use the session according
to the session types
$\SessionType \eqdef
\Out\tmessage{}\tint.\Out\tmessage{}\tint.\In\tmessage{}\tint.\End$
and
$\co\SessionType \eqdef
\In\tmessage{}\tint.\In\tmessage{}\tint.\Out\tmessage{}\tint.\End$,
respectively.
The system is modeled as the process
\begin{equation}
  \label{eq:session}
  \New\VarAlice
  \New\VarCarol
  \New{s}\parens{
    \invoke\DefAlice{\VarAlice,s}
    \parop
    \invoke\DefCarol{\VarCarol,s}
    \parop
    \invoke{\SessionProcess\SessionType}{s}
  }
\end{equation}
where $\DefAlice$ and $\DefCarol$ are defined as follows:
\[
  \begin{array}{@{}r@{~}c@{~}l@{}}
    \DefAlice(\VarSelf,s) & \triangleq &
    \phantom(
    \Send{s}\TagSend{4,\VarSelf}
    \parop
    \Receive\VarSelf\TagReply{s}
    \\ & &
    \parens{
      \Send{s}\TagSend{2,\VarSelf}
      \parop
      \Receive\VarSelf\TagReply{s}
      \\ & &
      \parens{
        \Send{s}\TagReceive\VarSelf
        \parop
        \Receive\VarSelf\TagReply{\Var,s}
        \\ & &
        \parens{
          \Send\VarSystem\TagPrintInt\Var
          \parop
          \Free\VarSelf
          \Done
        }
      }
    }
    \\
    \DefCarol(\VarSelf,s) & \triangleq &
    \phantom(
    \Send{s}\TagReceive\VarSelf
    \parop
    \Receive\VarSelf\TagReply{\VarX,s}
    \\ & &
    \parens{
      \Send{s}\TagReceive\VarSelf
      \parop
      \Receive\VarSelf\TagReply{\VarY,s}
      \\ & &
      \parens{
        \Send{s}\TagSend{\VarX + \VarY, \VarSelf}
        \parop
        \Receive\VarSelf\TagReply{s}
        \Free\VarSelf
        \Done
      }
    }
  \end{array}
\]

The process~\eqref{eq:session} and the definitions of $\DefAlice$
and $\DefCarol$ are well typed.  In general,
$\SessionProcess\SessionType$ is consistent with the declaration
$\SessionProcess\SessionType : (\VarSelf :
\In(\EncodeSessionType\SessionType \tmul
\EncodeSessionType{\co\SessionType}); \gempty)$ where
$\EncodeSessionType\SessionType$ is the pattern defined by the
following equations:
\[
  \begin{array}[b]{@{}r@{~}c@{~}l@{}}
    \EncodeSessionType\End & \eqdef & \tone
    \\
    \EncodeSessionType{\In\tmessage{}\Type.\SessionType} & \eqdef &
    \tmessage\TagReceive{\Out\tmessage\TagReply{\Type,\Out\EncodeSessionType\SessionType}}
    \\
    \EncodeSessionType{\Out\tmessage{}\Type.\SessionType} & \eqdef &
    \tmessage\TagSend{\Type,\Out\tmessage\TagReply{\Out\EncodeSessionType\SessionType}}
    \\
    \EncodeSessionType{\SessionTypeT \SBranch \SessionTypeS} & \eqdef &
    \tmessage\TagReceive{\Out\parens{
        \tmessage\TagLeft{\Out\EncodeSessionType\SessionTypeT}
        \tsum
        \tmessage\TagRight{\Out\EncodeSessionType\SessionTypeS}
      }
    }
    \\
    \EncodeSessionType{\SessionTypeT \SChoice \SessionTypeS} & \eqdef &
    \tmessage\TagLeft{\Out\tmessage\TagReply{\Out\EncodeSessionType\SessionTypeT}}
    \tsum
    \tmessage\TagRight{\Out\tmessage\TagReply{\Out\EncodeSessionType\SessionTypeS}}
  \end{array}
\]

% According to the encoding of the session types into mailbox types we
% have
% \[
%   \begin{array}{rcl}
%     \EncodeSessionType\SessionType & = &
%     \Out\tmessage\TagSend{
%       \tint,
%       \Out\tmessage\TagReply{
%         \Out\tmessage\TagSend{
%           \tint,
%           \Out\tmessage\TagReply{
%             \Out\tmessage\TagReceive{
%               \Out\tmessage\TagReply{
%                 \tint,
%                 \Out\tone
%               }
%             }
%           }
%         }
%       }
%     }
%     \\
%     \EncodeSessionType{\co\SessionType} & = &
%     \Out\tmessage\TagReceive{
%       \Out\tmessage\TagReply{
%         \tint,
%         \Out\tmessage\TagReceive{
%           \Out\tmessage\TagReply{
%             \tint,
%             \Out\tmessage\TagSend{
%               \tint,
%               \Out\tmessage\TagReply{
%                 \Out\tone
%               }
%             }
%           }
%         }
%       }
%     }
%   \end{array}
% \]
% and
% $\SessionProcess\SessionType : (\VarSelf :
% \In(\tmessage\TagSend\cdots \tmul \tmessage\TagReceive\cdots);
% \gempty)$. The type of $\VarSelf$, which contains the product of two
% patterns $\tmessage\TagSend\cdots$ and $\tmessage\TagReceive\cdots$,
% makes it clear that the session is meant to be used by two
% concurrent peers, each owning one endpoint of the session.

This encoding of binary sessions extends easily to internal and
external choices with arbitrary labels and also to recursive session
types by interpreting both the syntax of $\SessionType$ and the
definition of $\SessionProcess\SessionType$ coinductively. The usual
regularity condition ensures that $\SessionProcess\SessionType$ is
finitely representable. Finally, note that the notion of subtyping
for encoded session types induced by Definition~\ref{def:subt}
coincides with the conventional one~\cite{GayHole05}. Thus, the
mailbox type system subsumes a rich session type system where
Theorem~\ref{thm:soundness} corresponds to the well-known
communication safety and progress properties of sessions.

%%% Local Variables:
%%% mode: latex
%%% TeX-master: "main"
%%% End:

\subsection{Encoding of sessions with forks and joins}
\label{sec:example_forks_joins}

\newcommand{\TagMain}{\Tag[Then]}
\newcommand{\CollectProcess}[1]{\RecVar[Join]_{#1}}

We have seen that it is possible to share the output capability on a
mailbox among several processes. We can take advantage of this
feature to extend session types with forks and joins:
\[
  \SessionTypeT, \SessionTypeS
  ~~::=~~
  \End
  ~~\mid~~
  \In\tmessage{}\Type.\SessionType
  ~~\mid~~
  \Out\tmessage{}\Type.\SessionType
  ~~\mid~~
  \SessionTypeT \SBranch \SessionTypeS
  ~~\mid~~
  \SessionTypeT \SChoice \SessionTypeS
  ~~\mid~~
  \SJoin_{i\in I} \tmessage{\Tag_i}{\Type_i};\SessionType
  ~~\mid~~
  \SFork_{i\in I} \tmessage{\Tag_i}{\Type_i};\SessionType
\]

The idea is that the session type
$\SFork_{1\leq i\leq n} \tmessage{\Tag_i}{\Type_i};\SessionType$
describes an endpoint that can be used for sending \emph{all} of the
$\Tag_i$ messages, and then according to $\SessionType$. The
difference between
$\SFork_{1\leq i\leq n} \tmessage{\Tag_i}{\Type_i};\SessionType$ and
a session type of the form
$\Out\tmessage{}{\Type_1}\dots\Out\tmessage{}{\Type_n}.\SessionType$
is that the $\Tag_i$ messages can be sent by independent processes
(for example, by parallel workers) in whatever order instead of by a
single sender.
Dually, the session type
$\SJoin_{1\leq i\leq n} \tmessage{\Tag_i}{\Type_i};\SessionType$
describes an endpoint that can be used for collecting \emph{all} of
the $\Tag_i$ messages, and then according to $\SessionType$.
Forks and joins are dual to each other, just like simple outputs are
dual to simple inputs. The tags $\Tag_i$ need not be distinct, but
equal tags must correspond to equal argument types.

The extension of $\SessionProcess\SessionType$ to forks and joins is
shown below:
\[
  \begin{array}{@{}r@{~}c@{~}l@{}}
    \SessionProcess{
      \SFork_{i\in I} \tmessage{\Tag_i}{\Type_i};\SessionType
    }(\VarSelf) & \triangleq &
    \Receive\VarSelf\TagSend{s}
    \Receive\VarSelf\TagReceive{r}
    \invoke{
      \CollectProcess{
        \SFork_{i\in I} \tmessage{\Tag_i}{\Type_i};\SessionType
      }
    }{\VarSelf,s,r}
    \\
    \SessionProcess{
      \SJoin_{i\in I} \tmessage{\Tag_i}{\Type_i};\SessionType
    }(\VarSelf) & \triangleq &
    \invoke{
      \SessionProcess{
        \SFork_{i\in I} \tmessage{\Tag_i}{\Type_i};\SessionType
      }
    }{\VarSelf}
    \\
    \CollectProcess{
      \SFork_{i\in I} \tmessage{\Tag_i}{\Type_i};\SessionType
    }(\VarSelf,s,r) & \triangleq &
    \begin{cases}
      \Send{s}\TagReply\VarSelf
      \parop
      \Send{r}\TagReply\VarSelf
      \parop
      \invoke{\SessionProcess\SessionType}\VarSelf
      & \text{if $I = \emptyset$}
      \\
      \Receive\VarSelf{\Tag_i}{\Var_i}\parens{
        \Send{r}{\Tag_i}{\Var_i}
        \parop
        \invoke{
          \CollectProcess{
            \SFork_{i\in I\setminus\set{i}} \tmessage{\Tag_i}{\Type_i};\SessionType
          }
        }{\VarSelf,s,r}
      }
      & \text{if $i \in I$}
    \end{cases}
    \\
    \SessionProcess{
      \SJoin_{i\in I} \tmessage{\Tag_i}{\Type_i};\SessionType
    }(\VarSelf) & \triangleq &
    \invoke{
      \SessionProcess{
        \SFork_{i\in I} \tmessage{\Tag_i}{\Type_i};\SessionType
      }
    }\VarSelf
  \end{array}
\]

As in the case of simple interactions, sender and receiver manifest
their willingness to interact by storing $\TagSend$ and
$\TagReceive$ messages into the session's mailbox $\VarSelf$. At
that point, $\invoke{\CollectProcess\SessionType}{\VarSelf,s,r}$
forwards all the $\Tag_i$ messages coming from the sender side to
the receiver side, in some arbitrary order (case $i \in I$).
When there are no more messages to forward (case $I = \emptyset$)
both sender and receiver are notified with a $\TagReply$ message
that carries a reference to the session's endpoint, with its type
updated according to the rest of the continuation.

The encoding of session types extended to forks and joins follows
easily:
\[
  \begin{array}{@{}r@{~}c@{~}l@{}}
    \EncodeSessionType{\SFork_{i\in I} \tmessage{\Tag_i}{\Type_i};\SessionType}
    & \eqdef &
    \tmessage\TagSend{\Out\tmessage\TagReply{\Out\EncodeSessionType\SessionType}}
    \tmul
    \tMul_{i\in I}
    \tmessage{\Tag_i}{\Type_i}
    \\
    \EncodeSessionType{\SJoin_{i\in I} \tmessage{\Tag_i}{\Type_i};\SessionType}
    & \eqdef &
    \tmessage\TagReceive{
      \Out
      \parens{
        \tMul_{i\in I}
        \tmessage{\Tag_i}{\Type_i}
      }
      \tmul
      \tmessage\TagReply{\Out\EncodeSessionType\SessionType}
    }
  \end{array}
\]

An alternative definition of $\CollectProcess\SessionType$ that
fowards messages as soon as they become available can be obtained by
providing suitable input actions for each $i\in I$ instead of
picking an arbitrary $i\in I$.

%%% Local Variables:
%%% mode: latex
%%% TeX-master: "main"
%%% End:

%%% Local Variables:
%%% mode: latex
%%% TeX-master: "main"
%%% End:

% \input{inference}

\section{Related Work}
\label{sec:related}

\subparagraph*{Concurrent Objects.}
There are analogies between actors and concurrent objects. Both
entities are equipped with a unique identifier through which they
receive messages, they may interact with several concurrent clients
and their behavior may vary over time, as the entity interacts with
its clients. Therefore, static analysis techniques developed for
concurrent objects may be applicable to actors (and vice versa).
Relevant works exploring behavioral type systems for concurrent
objects include those of Najim \etal~\cite{NajmNimourStefani99},
Ravara and Vasconcelos~\cite{RavaraVasconcelos00}, and Puntigam
\etal~\cite{Puntigam01,PuntigamPeter01}.
As in the pure actor model, each object has a unique mailbox and the
input capability on that mailbox cannot be transferred. The mailbox
calculus does not have these constraints. A notable variation is the
model studied by Ravara and Vasconcelos~\cite{RavaraVasconcelos00},
which accounts for \emph{distributed objects}: there can be several
copies of an object that react to messages targeted to the same
mailbox.
Another common trait of these works is that the type discipline
focuses on sequences of method invocations and types contain
(abstract) information on the internal state of objects and on state
transitions. Indeed, types are either finite-state
automata~\cite{NajmNimourStefani99}, or terms of a process
algebra~\cite{RavaraVasconcelos00} or tokens annotated with state
transitions~\cite{PuntigamPeter01}.
In contrast, mailbox types focus on the content of a mailbox and
sequencing is expressed in the type of explicit continuations.
The properties enforced by the type systems in these works differ
significantly. Some do not consider deadlock
freedom~\cite{RavaraVasconcelos00,Puntigam01}, others do not account
for out-of-order message processing~\cite{Puntigam01}.
Details on the enforced properties also vary. For example, the
notion of protocol conformance used by Ravara and
Vasconcelos~\cite{RavaraVasconcelos00} is such that any message sent
to an object that is unable to handle that message, but can do so in
some future state is accepted. In our setting, this would mean
allowing to send a $\TagRelease$ message to a free lock if the lock
is acquired later on, or allowing to send a $\TagReply$ message to
an account if the account will later be involved in a transaction.

The most closely related work among those addressing concurrent
objects is the one by Crafa and Padovani~\cite{CrafaPadovani17}, who
propose the use of the Objective Join Calculus as a model for
non-uniform, concurrent objects and develop a type discipline that
can be used for enforcing concurrent object protocols.  Mailbox
types have been directly inspired by their types of concurrent
objects.
There are two main differences between the work of Crafa and
Padovani~\cite{CrafaPadovani17} and our own.
First, in the Objective Join Calculus every object is associated
with a single mailbox, just like in the pure actor
model~\cite{HewittEtAl73,Agha86}, meaning that mailboxes are not
first class. As a consequence, the types considered by Crafa and
Padovani~\cite{CrafaPadovani17} all have an (implicit) output
capability.
Second, in the Objective Join Calculus input operations are defined
atomically on molecules of messages, whereas in the mailbox calculus
messages are received one at a time. As a consequence, the type of a
mailbox in the work of Crafa and Padovani~\cite{CrafaPadovani17} is
invariant, whereas the same mailbox may have different types at
different times in the mailbox calculus
(Example~\ref{ex:lock_type}). Remarkably, this substantial
difference has no impact on the structure of the type language that
we consider.
% %
% Padovani~\cite{} has defined a type checking algorithm for the type
% system of Crafa and Padovani. We think that this work can be used in
% large part for defining a corresponding type checking algorithm for
% the mailbox calculus.

\subparagraph*{Static analysis of actors.}
Srinivasan and Mycroft~\cite{SrinivasanMycroft08} define a type
discipline for controlling the ownership of messages and ensuring
actor isolation, but consider only uniformly typed mailboxes and do
not address mailbox conformance or deadlock freedom.

Christakis and Sagonas~\cite{ChristakisSagonas11} describe a static
analysis technique whose aim is to ensure matching between send and
receive operations in actors.  The technique, which is described
only informally and does not account for deadlocks, has been
implemented in a tool called \textsf{dialyzer} and used for the
analysis of Erlang programs.

Crafa~\cite{Crafa12} defines a behavioural type system for actors
aimed at ensuring that the order of messages produced and consumed
by an actor follows a prescribed protocol. Protocols are expressed
as types and describe the behavior of actors rather than the content
of the mailboxes they use. Deadlock freedom is not addressed.

Charousset \etal~\cite{CharoussetHiesgenSchmidt16} describe the
design and implementation of CAF, the C++ Actor Framework. Among the
features of CAF is the use of \emph{type-safe message passing
  interfaces} that makes it possible to statically detect a number
of protocol violations by piggybacking on the C++ type system. There
are close analogies between CAF's message passing interfaces and
mailbox types with output capability: both are equipped with a
subset semantics and report only those messages that can be stored
into the mailbox through a mailbox reference with that type.
Charousset \etal~\cite{CharoussetHiesgenSchmidt16} point out that
this feature fosters the decoupling of actors and enables
incremental program recompilation.

Giachino \etal~\cite{GiachinoEtAl16,Mastandrea16} define a type
system for the deadlock analysis of actors making use of implicit
futures. Mailbox conformance and deadlocks due to communications are
not taken into account.

Fowler \etal~\cite{FowlerLindleyWadler17} formalize channel-based
and mailbox-based communicating systems, highlighting the
differences between the two models and studying type-preserving
encodings between them. Mailboxes in their work are uniformly typed,
but the availability of union types make it possible to host
heterogeneous values within the same mailbox. This however may lead
to a loss of precision in typing. This phenomenon, dubbed \emph{type
  pollution} by Fowler \etal~\cite{FowlerLindleyWadler17}, is
observable to some extent also in our typing discipline and can be
mitigated by the use of multiple mailboxes (\cf
Section~\ref{sec:example_master_workers}).
Finally, Fowler \etal~\cite{FowlerLindleyWadler17} leave the
extension of their investigation to behaviorally-typed language of
actors as future work. Our typing discipline is a potential
candidate for this investigation and addresses a more general
setting thanks to the support for first-class mailboxes.

\subparagraph*{Sessions and actors.}
The encoding of binary sessions into actors discussed in
Section~\ref{sec:example_binary_sessions} is new and has been
inspired by the encoding of binary sessions into the linear
$\pi$-calculus~\cite{Kobayashi07,DardhaGiachinoSangiorgi17}, whereby
each message is paired with a continuation. In our case, the
continuation, instead of being a fresh (linear) channel, is either
the mailbox of the peer or that of the session. This style of
communication with explicit continuation passing is idiomatic in the
actor model, which is based on asynchronous communications.
The encoding discussed in Section~\ref{sec:example_binary_sessions}
can be generalized to multiparty sessions by defining
$\SessionProcess\SessionType$ as a \emph{medium process} through
which messages are exchanged between the parties of the
session. This idea has been put forward by Caires and
P\'erez~\cite{CairesPerez16} to encode multiparty sessions using
binary sessions.

Mostrous and Vasconcelos~\cite{MostrousVasconcelos11} study a
session type system for enforcing ordered dyadic interactions in
core Erlang. They use \emph{references} for distinguishing messages
pertaining to different sessions, making use of the advanced pattern
matching capabilities of Erlang. Their type system guarantees a
weaker form of mailbox conformance, whereby junk messages may be
present at the end of a computation, and does not consider deadlock
freedom. Compared to our encoding of binary sessions, their approach
does not require a medium process representing the session itself.

Neykova and Yoshida~\cite{NeykovaYoshida16} propose a framework
based on multiparty session types for the specification and
implementation of actor systems with guarantees on the order of
interactions. This approach is applicable when designing an entire
system and both the network topology and the communication protocol
can be established in advance.
Fowler~\cite{Fowler16} builds upon the work of Neykova and Yoshida
to obtain a runtime protocol monitoring mechanism for Erlang.
Charalambides \etal~\cite{CharalambidesDingesAgha16} extend the
multiparty session approach with a protocol specification language
that is parametric in the number of actors participating in the
system.
In contrast to these approaches based on multiparty/global session
types, our approach ensures mailbox conformance and deadlock freedom
of a system compositionally, as the system is assembled out of
smaller components, and permits the modeling of systems with a
dynamic network topology or with a varying number of interacting
processes.

\subparagraph*{Linear logic.}
Shortly after its introduction, linear logic has been proposed as a
specification language suitable for concurrency.
Following this idea, Kobayashi and
Yonezawa~\cite{KobayashiYonezawa94,KobayashiYonezawa95} have studied
formal models of concurrent objects and actors based on linear
logic. More recently, a direct correspondence between propositions
of linear logic and session types has been
discovered~\cite{CairesPfenning10,Wadler14,LindleyMorris15}.
There are several analogies between the mailbox type system and the
proof system of linear logic.
Mailbox types with output capability are akin to positive
propositions, with $\Out\tzero$ and $\Out\tone$ respectively playing
the roles of $0$ and $1$ in linear logic and
$\Out(\PatternE \tsum \PatternF)$ and
$\Out(\PatternE \tmul \PatternF)$ corresponding to $\oplus$ and
$\otimes$. Mailbox types with input capability are akin to negative
propositions, with $\In\tzero$ and $\In\tone$ corresponding to
$\top$ and $\bot$ and $\In(\PatternE \tsum \PatternF)$ and
$\In(\PatternE \tmul \PatternF)$ corresponding to $\with$ and
$\parr$.
Rules~\refrule{t-fail}, \refrule{t-free} and \refrule{t-branch} have
been directly inspired from the rules for $\top$, $\bot$ and $\with$
in the classical sequent calculus for linear logic.
Subtyping corresponds to inverse linear implication and its
properties are consistent with those of the logic connectives
according to the above interpretation.
A noteworthy difference between our type system and those for
session types based on linear
logic~\cite{CairesPfenning10,Wadler14,LindleyMorris15} is the need
for dependency graphs to ensure deadlock freedom
(Section~\ref{sec:graphs}). There are two reasons that call for such
auxiliary mechanism in our setting. First, the rule \refrule{t-par}
is akin to a symmetric cut rule. Dependency graphs are necessary to
detect mutual dependencies that may consequently arise
(Example~\ref{ex:future_deadlock}). Second, unlike session endpoints
that are linear resources, mailbox references can be used
non-linearly. Thus, the multiplicity of dependencies, and not just
the presence or absence thereof, is relevant
(Example~\ref{ex:multiplicity}).

%%% Local Variables:
%%% mode: latex
%%% TeX-master: "main"
%%% End:

\section{Concluding Remarks}
\label{sec:conclusions}

We have presented a mailbox type system for reasoning about
processes that communicate through first-class, unordered
mailboxes. The type system enforces mailbox conformance, deadlock
freedom and, for a significant class of processes, junk freedom as
well.
In sharp contrast with session types, mailbox types embody the
unordered nature of mailboxes and enable the description of
mailboxes concurrently accessed by several processes, abstracting
away from the state and behavior of the objects/actors/processes
using these mailboxes. The fact that a mailbox may have different
types during its lifetime is entirely encapsulated by the typing
rules and not apparent from mailbox types themselves.
The mailbox calculus subsumes the actor model and allows us to
analyze systems with a dynamic network topology and a varying number
of processes mixing different concurrency abstractions.

There are two natural extensions of the mailbox calculus that we
have not incorporated in the formal development for the sake of
simplicity.
First, it is possible to relax the syntax of guarded processes to
accommodate actions referring to different mailboxes as well as
actions representing timeouts. This extension makes the typing rules
for guards more complex to formulate but enhances expressiveness and
precision of typing (see Appendix~\ref{sec:example_readers_writer}).
Second, it is possible to allow multiple processes to receive
messages from the same mailbox by introducing a distinguished
capability that identifies shared mailboxes.  The notion of type
combination (Definition~\ref{def:type_combination}) must be suitably
revised for deadling with shared mailboxes and avoid the soundness
problems discussed in Example~\ref{ex:linear_input}.  With this
extension in place, it might also be possible to replace recursion
with replication in the calculus.

Concerning further developments, the intriguing analogies between
the mailbox type system and linear logic pointed out in
Section~\ref{sec:related} surely deserve a formal investigation.
On the practical side, a primary goal to fulfil is the development
of a type checking/inference algorithm for the proposed typing
discipline.  Subtyping is decidable and a type checking algorithm
for a slightly simpler type language has already been
developed~\cite{CobaltBlue}. We are confident that a type checking
algorithm for the mailbox calculus can be obtained by reusing much
of these known results.
Concerning the applicability of the approach to real-world
programming languages, one promising approach is the development of
a tool for the analysis of Java bytecode, possibly with the help of
Java annotations, along the lines of what has already been done for
Kilim~\cite{SrinivasanMycroft08}.
Other ideas for further developments include the extension of
pattern atoms so as to accommodate Erlang-style matching of messages
and the implementation of optimal matching algorithms driven by type
information.

%%% Local Variables:
%%% mode: latex
%%% TeX-master: "main"
%%% End:

% \subparagraph*{Acknowledgements.}

% I want to thank \dots

\bibliography{main}

\appendix

\section{Supplementary Properties}
\label{sec:extra_properties}

\subsection{Properties of subtyping}
\label{sec:prop_subtyping}

\begin{proposition}
  \label{prop:subt_reliable_usable}
  If $\TypeT \subt \TypeS$, then $\TypeT$ reliable implies $\TypeS$
  reliable and $\TypeS$ usable implies $\TypeT$ usable.
\end{proposition}
\begin{proof}
  We only prove the part of the statement concerning reliability,
  since usability is analogous, and we prove the reverse implication.
  The interesting case is when $\TypeT = \In\PatternE$. From the
  hypothesis $\TypeT \subt \TypeS$ we deduce $\TypeS = \In\PatternF$
  and $\PatternE \subp \PatternF$.
  Now if $\PatternF \subp \tzero$ we have
  $\PatternE \subp \PatternF \subp \tzero$ by transitivity of
  $\subp$. Hence, $\TypeS$ unreliable implies $\TypeT$ unreliable.
\end{proof}

\begin{proposition}
  \label{prop:subt_der_subt}
  If $\tmessage\Tag\TypesT \tmul \PatternF \subp \PatternE$ and
  $\PatternF \not\subp \tzero$ and
  $\pder\PatternE{\tmessage\Tag\TypesS}$ is defined, then
  $\TypesT \subt \TypesS$ and
  $\PatternF \subp \pder\PatternE{\tmessage\Tag\TypesS}$.
\end{proposition}
\begin{proof}
  Let $\Bag \in \sem{\tmessage\Tag\TypesT \tmul \PatternF}$.
  From the hypothesis $\PatternF \not\subp \tzero$ we deduce that
  $\Bag = \bag{\tmessage\Tag\TypesT} \bunion \bag[i\in
  I]{\tmessage{\Tag_i}{\TypesT_i}}$ where
  $\bag[i\in I]{\tmessage{\Tag_i}{\TypesT_i}} \in \sem\PatternF$.
  From the hypothesis
  $\tmessage\Tag\TypesT \tmul \PatternF \subp \PatternE$ and the
  definition of $\subp$ we deduce that there exist $\TypesR$ and
  $\TypesR_i$ such that
  $\bag{\tmessage\Tag\TypesR} \bunion \bag[i\in
  I]{\tmessage{\Tag_i}{\TypesR_i}} \in \sem\PatternE$ and
  $\TypesT \subt \TypesR$ and $\TypesT_i \subt \TypesR_i$ for every
  $i\in I$.
  From the hypothesis that $\pder\PatternE{\tmessage\Tag\TypesS}$ is
  defined we deduce that $\TypesR \subt \TypesS$, hence
  $\TypesT \subt \TypesS$.
  Also, from the definition of derivative, we have
  $\bag[i\in I]{\tmessage{\Tag_i}{\TypesR_i}} \in
  \sem{\pder\PatternE{\tmessage\Tag\TypesS}}$.
  We conclude
  $\PatternF \subp \pder\PatternE{\tmessage\Tag\TypesS}$.
\end{proof}

We extend the terminology used for type classification
(Definition~\ref{def:classification}) to type environments as
well. We say that $\Context$ is reliable/irrelevant if all the types
in its range are reliable/irrelevant.

\begin{example}
  \label{ex:global_assumptions}
  The global assumptions we made on types are aimed at ensuring that
  in a judgment $\wtp{\Name : \Type, \Context}\Process{}$ where
  $\Type$ is relevant and $\Context$ is reliable the name $\Name$
  does indeed occur free in $\Process$.
  If we allowed unreliable arguments, then it would be possible to
  derive
  \[
    \wtp{
      \NameU : \Type,
      \NameV : \In\tmessage\Tag{\In\tzero}
    }{
      \Receive\NameV\Tag\Var
      \Fail\Var
    }{
      \gempty
    }
  \]
  meaning that $\NameU$ is not guaranteed to occur even if $\Type$
  is relevant.
  Analogously, if we allowed unusable arguments, then it would be
  possible to derive
  \[
    \wtp{
      \NameU : \Type,
      \NameV : \Out\tmessage\Tag{\Out\tzero}
    }{
      \New\Mailbox\parens{
        \Send\NameV\Tag\Mailbox
        \parop
        \Fail\Mailbox
      }
    }{
      \New\Mailbox\gedge\NameV\Mailbox
    }
  \]
  again meaning that $\NameU$ is not guaranteed to occur even if
  $\Type$ is relevant.
  \eoe
\end{example}

%%% Local Variables:
%%% mode: latex
%%% TeX-master: "main"
%%% End:

\subsection{Properties of dependency graphs}
\label{sec:prop_graphs}

\begin{proposition}[structure preserving transitions]
  The following properties hold:
  \begin{enumerate}
  \item If
    $\dgraphA_1 \gunion \dgraphA_2 \gred\NameU\NameV \dgraph$, then
    $\dgraph = \dgraphA_1' \gunion \dgraphA_2'$ for some
    $\dgraphA_1'$ and $\dgraphA_2'$.
  \item If $\New\Mailbox\dgraphA \gred\NameU\NameV \dgraphB$, then
    $\dgraphB = \New\Mailbox\dgraph'$ for some $\dgraph'$.
  \end{enumerate}
\end{proposition}
\begin{proof}
  A straightforward induction on the derivation of the transition.
\end{proof}

Because of the previous result, in the proofs that follow we only
consider transitions where the structure of the dependency graph is
preserved.

\begin{proposition}
  \label{prop:gred_left}
  If
  $\dgraphA_1 \gunion \dgraphA_2 \gred\NameU\NameV \dgraphA_1'
  \gunion \dgraphA_2'$, then
  $(\dgraphB \gunion \dgraphA_1) \gunion \dgraphA_2
  \gred\NameU\NameV (\dgraphB \gunion \dgraphA_1') \gunion
  \dgraphA_2'$.
\end{proposition}
\begin{proof}
  By induction on the derivation of
  $\dgraphA_1 \gunion \dgraphA_2 \gred\NameU\NameV \dgraphA_1'
  \gunion \dgraphA_2'$ and by cases on the last rule applied. We
  omit the discussion of \refrule{g-right}, which is symmetric to
  \refrule{g-left}.

  \proofrule{g-left}
  Then $\dgraphA_1 \gred\NameU\NameV \dgraphA_1'$ and
  $\dgraphA_2 = \dgraphA_2'$.
  We conclude by applying \refrule{g-right}, then \refrule{g-left}.

  \proofrule{g-trans}
  Then
  $\dgraphA_1 \gunion \dgraphA_2 \gred\NameU\NameW \dgraphA_1''
  \gunion \dgraphA_2'' \gred\NameW\NameV \dgraphA_1' \gunion
  \dgraphA_2'$.
  Using the induction hypothesis we derive
  $(\dgraphB \gunion \dgraphA_1) \gunion \dgraphA_2
  \gred\NameU\NameW (\dgraphB \gunion \dgraphA_1'') \gunion
  \dgraphA_2'' \gred\NameW\NameV (\dgraphB \gunion \dgraphA_1')
  \gunion \dgraphA_2'$ and we conclude with one application of
  \refrule{g-trans}.
\end{proof}

\begin{proposition}
  \label{prop:gred_gunion}
  If
  $\dgraph_1 \gunion (\dgraph_2 \gunion \dgraph_3) \gred\NameU\NameV
  \dgraph_1' \gunion (\dgraph_2' \gunion \dgraph_3')$, then
  $(\dgraph_1 \gunion \dgraph_2) \gunion \dgraph_3 \gred\NameU\NameV
  (\dgraph_1' \gunion \dgraph_2') \gunion \dgraph_3'$.
\end{proposition}
\begin{proof}
  By induction on the derivation of
  $\dgraph_1 \gunion (\dgraph_2 \gunion \dgraph_3) \gred\NameU\NameV
  \dgraph_1' \gunion (\dgraph_2' \gunion \dgraph_3')$ and by cases
  on the last rule applied.

  \proofrule{g-left}
  Then $\dgraph_1 \gred\NameU\NameV \dgraph_1'$ and
  $\dgraph_2 = \dgraph_2'$ and $\dgraph_3 = \dgraph_3'$.
  We conclude with two applications of \refrule{g-left}.

  \proofrule{g-right}
  Then $\dgraph_1 = \dgraph_1'$ and
  $\dgraph_2 \gunion \dgraph_3 \gred\NameU\NameV \dgraph_2' \gunion
  \dgraph_3'$.
  We conclude by Proposition~\ref{prop:gred_left}.

  \proofrule{g-trans}
  Then
  $\dgraph_1 \gunion (\dgraph_2 \gunion \dgraph_3) \gred\NameU\NameW
  \dgraph_1'' \gunion (\dgraph_2'' \gunion \dgraph_3'')
  \gred\NameW\NameV \dgraph_1' \gunion (\dgraph_2' \gunion
  \dgraph_3')$.
  Using the induction hypothesis we deduce
  $(\dgraph_1 \gunion \dgraph_2) \gunion \dgraph_3 \gred\NameU\NameW
  (\dgraph_1'' \gunion \dgraph_2'') \gunion \dgraph_3''
  \gred\NameW\NameV (\dgraph_1' \gunion \dgraph_2') \gunion
  \dgraph_3'$.
  We conclude with one application of \refrule{g-trans}.
\end{proof}

\begin{proposition}
  \label{prop:gred_new}
  If
  $\New\Mailbox\dgraph_1 \gunion \dgraph_2 \gred\NameU\NameV
  \New\Mailbox\dgraph_1' \gunion \dgraph_2'$ and
  $\Mailbox\not\in\fn(\dgraph_2)$, then
  $\New\Mailbox(\dgraph_1 \gunion \dgraph_2) \gred\NameU\NameV
  \New\Mailbox(\dgraph_1' \gunion \dgraph_2')$.
\end{proposition}
\begin{proof}
  By induction on the derivation of
  $\New\Mailbox\dgraph_1 \gunion \dgraph_2 \gred\NameU\NameV
  \New\Mailbox\dgraph_1' \gunion \dgraph_2'$ and by cases on the
  last rule applied. We do not discuss rule \refrule{g-right}, which
  is straightforward.

  \proofrule{g-left}
  Then
  $\New\Mailbox\dgraph_1 \gred\NameU\NameV \New\Mailbox\dgraph_1'$
  and $\dgraph_2 = \dgraph_2'$.
  From \refrule{g-new} we deduce
  $\dgraph_1 \gred\NameU\NameV \dgraph_1'$ and
  $\Mailbox\ne\NameU,\NameV$.
  We conclude with one application of \refrule{g-left} and one
  application of \refrule{g-new}.

  \proofrule{g-trans}
  Then
  $\New\Mailbox\dgraph_1 \gunion \dgraph_2 \gred\NameU\NameW
  \New\Mailbox\dgraph_1'' \gunion \dgraph_2'' \gred\NameW\NameV
  \New\Mailbox\dgraph_1' \gunion \dgraph_2'$.
  From the induction hypothesis we deduce
  $\New\Mailbox(\dgraph_1 \gunion \dgraph_2) \gred\NameU\NameW
  \New\Mailbox(\dgraph_1'' \gunion \dgraph_2'') \gred\NameW\NameV
  \New\Mailbox(\dgraph_1' \gunion \dgraph_2')$.
  We conclude with one application of \refrule{g-trans}.
\end{proof}

\begin{proposition}
  \label{prop:grel}
  The following properties hold:
  \begin{enumerate}
  \item\label{item:grel_empty}
    $\grel{\gempty \gunion \dgraph} = \grel\dgraph$
  \item\label{item:grel_comm}
    $\grel{\dgraph_1 \gunion \dgraph_2} = \grel{\dgraph_2 \gunion
      \dgraph_1}$
  \item\label{item:grel_assoc}
    $\grel{\dgraph_1 \gunion (\dgraph_2 \gunion \dgraph_3)} =
    \grel{(\dgraph_1 \gunion \dgraph_2) \gunion \dgraph_3}$
  \item\label{item:grel_new} If $\Mailbox\not\in\fn(\dgraph_2)$,
    then
    $\grel{\New\Mailbox\dgraph_1 \gunion \dgraph_2} =
    \grel{\New\Mailbox(\dgraph_1 \gunion \dgraph_2)}$.
  \end{enumerate}
\end{proposition}
\begin{proof}
  Items~\ref{item:grel_empty} and~\ref{item:grel_comm} are trivial.
  Items~\ref{item:grel_assoc} and~\ref{item:grel_new} respectively
  follow from Proposition~\ref{prop:gred_gunion} and
  Proposition~\ref{prop:gred_new}.
\end{proof}

%%% Local Variables:
%%% mode: latex
%%% TeX-master: "main"
%%% End:

\subsection{Properties of type environments}
\label{sec:prop_type_environments}

\begin{proposition}
  \label{prop:env_subt_reliable}
  If $\ContextA \subt \ContextB$, then $\ContextA$ reliable implies
  $\ContextB$ reliable and $\ContextB$ usable implies $\ContextA$
  usable.
\end{proposition}
\begin{proof}
  Immediate from Proposition~\ref{prop:subt_reliable_usable}.
\end{proof}

\begin{proposition}
  \label{prop:cmul_subt}
  If both $\ContextA_1 \cmul \ContextA_2$ and
  $\ContextB_1, \ContextB_2$ are defined and
  $\ContextA_i \subt \ContextB_i$, then
  $\ContextA_1 \cmul \ContextA_2 \subt \ContextB_1, \ContextB_2$.
\end{proposition}
\begin{proof}
  We discuss a few notable cases when
  $\Name \in (\dom(\ContextA_1) \cap \dom(\ContextA_2) \cap
  \dom(\ContextB_1)) \setminus \dom(\ContextB_1)$.

  \proofcase{$\ContextA_1(\Name) = \Out\PatternE$ and
    $\ContextA_2(\Name) = \Out\PatternF$ and
    $\ContextB_1(\Name) = \Out\PatternG$}
  Then $\Out\PatternE \subt \Out\PatternG$ and
  $\Out\PatternF \subt \Out\tone$, which means
  $\PatternG \subp \PatternE$ and $\tone \subp \PatternF$.
  From the properties of $\subp$ we deduce
  $\PatternG \subp \PatternG \tmul \tone \subp \PatternE \tmul
  \PatternF$.
  We conclude
  $(\Context_1 \cmul \Context_2)(\Name) = \Out\parens{\PatternE
    \tmul \PatternF} \subt \Out\PatternG = (\ContextB_1,
  \ContextB_2)(\Name)$.

  \proofcase{$\ContextA_1(\Name) = \Out\PatternE$ and
    $\ContextA_2(\Name) = \In\parens{\PatternE \tmul \PatternF}$
    and $\ContextB_1(\Name) = \Out\PatternG$}
  Then $\Out\PatternE \subt \Out\PatternG$ and
  $\In\parens{\PatternE \tmul \PatternF} \subt \Out\tone$, which
  means $\PatternG \subp \PatternE$ and
  $\PatternE \tmul \PatternF \tmul \tone \subp \PatternE \tmul
  \PatternF \subp \tone$.
  From the properties of $\subp$ we deduce
  $\PatternG \subp \PatternE \subp \tone$ and
  $\PatternF \subp \tone$.
  We conclude
  $\parens{\Context_1 \cmul \Context_2}(\Name) = \In\PatternF \subt
  \Out\PatternG = (\ContextB_1, \ContextB_2)(\Name)$.

  \proofcase{$\ContextA_1(\Name) = \In\parens{\PatternE \tmul
      \PatternF}$ and $\ContextA_2(\Name) = \Out\PatternE$ and
    $\ContextB_1(\Name) = \In\PatternG$}
  Then $\In\parens{\PatternE \tmul \PatternF} \subt \In\PatternG$
  and $\Out\PatternE \subt \Out\tone$, which means
  $\PatternE \tmul \PatternF \subp \PatternG$ and
  $\tone \subp \PatternE$.
  From the properties of $\subp$ we deduce
  $\PatternF \subp \tone \tmul \PatternF \subp \PatternE \tmul
  \PatternF \subp \PatternG$.
  We conclude
  $\parens{\ContextA_1 \cmul \ContextA_2}(\Name) = \In\PatternF
  \subt \In\PatternG = (\ContextB_1, \ContextB_2)(\Name)$.

  \proofcase{$\ContextA_1(\Name) = \In\parens{\PatternE \tmul
      \PatternF}$ and $\ContextA_2(\Name) = \Out\PatternE$ and
    $\ContextB_1(\Name) = \Out\PatternG$}
  Then $\In(\PatternE \tmul \PatternF) \subt \Out\PatternG$ and
  $\Out\PatternE \subt \Out\tone$, which means
  $\PatternE \tmul \PatternF \tmul \PatternG \subp \tone$ and
  $\tone \subp \PatternE$.
  From the properties of $\subp$ we deduce
  $\PatternF \tmul \PatternG \subp \tone$.
  We conclude
  $\parens{\ContextA_1 \cmul \ContextA_2}(\Name) = \In\PatternF
  \subt \Out\PatternG = (\ContextB_1, \ContextB_2)(\Name)$.
\end{proof}

\begin{proposition}
  \label{prop:env_mul_reliable}
  If $\ContextA$ is usable and $\ContextA \cmul \ContextB$ is
  reliable, then $\ContextB$ is reliable.
\end{proposition}
\begin{proof}
  We discuss only the interesting case when
  $\ContextA(\Name) = \Out\PatternE$ and
  $\ContextB(\Name) = \In(\PatternE \tmul \PatternF)$, the others
  being symmetric or simpler.
  Then $(\ContextA \cmul \ContextB)(\Name) = \In\PatternF$.
  From the hypothesis that $\ContextA$ is usable we deduce
  $\Out\tzero \not\subt \Out\PatternE$, that is
  $\PatternE \not\subp \tzero$.
  From the hypothesis that $\ContextA \cmul \ContextB$ is reliable
  we deduce $\In\PatternF \not\subt \In\tzero$, that is
  $\PatternF \not\subp \tzero$.
  We conclude $\PatternE \tmul \PatternF \not\subp \tzero$.
\end{proof}

%%% Local Variables:
%%% mode: latex
%%% TeX-master: "main"
%%% End:

%%% Local Variables:
%%% mode: latex
%%% TeX-master: "main"
%%% End:

\section{Proof of Theorem~\ref{thm:sr}}
\label{sec:proof_type_preservation}

\begin{lemma}
  \label{lem:subst}
  If $\wtp{\Context, \Var : \Type}\Process\dgraph$ and
  $\Mailbox \not\in\dom(\Context)$, then
  $\wtp{\Context, \Mailbox :
    \Type}{\Process\subst\Mailbox\Var}{\dgraph\subst\Mailbox\Var}$.
\end{lemma}
\begin{proof}
  A straightforward induction on the derivation of
  $\wtp{\Context, \Var : \Type}\Process\dgraph$.
\end{proof}

\begin{lemma}
  If $\wtp\Context\ProcessP\dgraph$ and
  $\ProcessP \equiv \ProcessQ$, then $\wtp\Context\ProcessQ\dgraph$.
\end{lemma}
\begin{proof}
  An easy induction on the derivation of
  $\ProcessP \equiv \ProcessQ$, reasoning by cases on the last rule
  applied. The most critical points concern dependency graphs, whose
  structure changes as processes and restrictions are
  rearranged. Proposition~\ref{prop:grel} provides all the arguments
  to conclude that these changes do not affect the semantics of
  dependency graphs and therefore that no cycles are introduced.
\end{proof}

\begin{reptheorem}{thm:sr}
  If $\Context$ is reliable and $\wtp\Context\ProcessP\dgraph$ and
  $\ProcessP \red \ProcessQ$, then
  $\wtp\Context\ProcessQ\dgraph$.
\end{reptheorem}
\begin{proof}
  By induction on the derivation of $\ProcessP \red \ProcessQ$ and by
  cases on the last rule applied.

  \proofrule{r-read}
  Then
  $\ProcessP = \Send\Mailbox\Tag\MailboxesC \parop
  \Receive\Mailbox\Tag\Vars\Process \sumop \Guard\red
  \ProcessP\subst\MailboxesC\Vars = \ProcessQ$.
  From \refrule{t-sub} and \refrule{t-par} we deduce that there
  exist $\Context_1$, $\Context_2$, $\dgraph_1$ and $\dgraph_2$ such
  that:
  \begin{itemize}
  \item $\Context \subt \Context_1 \cmul \Context_2$
  \item $\dgraph \gimplies \dgraph_1 \gunion \dgraph_2$
  \item
    $\wtp{\Context_1}{\Send\Mailbox\Tag{\MailboxesC}}{\dgraph_1}$
  \item
    $\wtp{\Context_2}{\Receive\Mailbox\Tag\Vars\Process \sumop
      \Guard}{\dgraph_2}$
  \end{itemize}

  From \refrule{t-sub} and \refrule{t-msg} we deduce that there
  exist $\TypesT$ such that:
  \begin{itemize}
  \item
    $\Context_1 \subt \Mailbox : \Out{\tmessage\Tag\Types},
    \MailboxesC : \Types$
  \item
    $\dgraph_1 \gimplies \gedge\Mailbox{\set\MailboxesC}$
  \end{itemize}

  From \refrule{t-sub} and \refrule{t-guard} we deduce that there
  exist $\ContextB$ and $\PatternE$ such that:
  \begin{itemize}
  \item $\Context_2 \subt \Mailbox : \In\PatternE, \ContextB$
  \item
    $\wtg{\Mailbox : \In\PatternE,
      \ContextB}{\Receive\Mailbox\Tag\Vars\Process \sumop \Guard}$
  \item $\df{}\PatternE$
  \item $\dgraph_2 \gimplies \gedge\Mailbox{\dom(\ContextB)}$
  \end{itemize}

  From \refrule{t-branch}, \refrule{t-in} and $\df{}\PatternE$ we
  deduce that there exist $\TypesS$, $\PatternF_1$, $\PatternF_2$
  and $\dgraphB$ such that:
  \begin{itemize}
  \item
    $\PatternE = \tmessage\Tag\TypesS \tmul \PatternF_1 \tsum
    \PatternF_2$
  \item
    $\wtp{\Mailbox : \In\PatternF_1, \ContextB, \Vars :
      \TypesS}{\Process}{\dgraphB}$
  \item $\PatternF_1 \eqp \pder\PatternE{\tmessage\Tag\TypesS}$
  \end{itemize}

  From the fact that $\dgraph$ is acyclic and the properties of
  $\dgraph_1$ and $\dgraph_2$ we deduce that
  $\set\MailboxesC \cap \dom(\ContextB) = \emptyset$. Indeed,
  suppose by contradiction that
  $\Name \in \set\MailboxesC \cap \dom(\ContextB)$ and observe that
  $\Name$ cannot be $\Mailbox$. Then
  $\dgraph_1 \gimplies \gedge\Mailbox\Name$ and
  $\dgraph_2 \gimplies \gedge\Name\Mailbox$, so
  $\dgraph \gimplies \gedge\Mailbox\Mailbox$ which is absurd.

  From the definition of subtype environment and of
  $\Context_1 \cmul \Context_2$ we deduce that:
  \begin{itemize}
  \item $\Context_1 = \Mailbox : \Out\PatternE', \Context_1'$ where
    $\Out\PatternE' \subt \Out{\tmessage\Tag\TypesT}$ and
    $\Context'_1 \subt \MailboxesC : \TypesT$
  \item
    $\Context_2 = \Mailbox : \In\parens{\PatternE' \tmul
      \PatternF'}, \Context_2'$ where
    $\In\parens{\PatternE' \tmul \PatternF'} \subt \In\PatternE$ and
    $\Context_2' \subt \ContextB$
  \end{itemize}

  From the definition of $\subt$, the precongruence and transitivity
  properties of $\subp$ we derive
  $\tmessage\Tag\TypesT \tmul \PatternF' \subp \PatternE' \tmul
  \PatternF' \subp \PatternE$.
  From the hypothesis that $\Context$ is reliable we deduce
  $\PatternF' \not\subp \tzero$.
  From Proposition~\ref{prop:subt_der_subt} and $\df{}\PatternE$ we
  deduce $\TypesT \subt \TypesS$ and
  $\PatternF' \subp \pder\PatternE{\tmessage\Tag\TypesS}$.
  From Lemma~\ref{lem:subst} we deduce
  $\wtp{\Mailbox : \In\PatternF_1, \ContextB, \MailboxesC :
    \TypesS}{\Process\subst\MailboxesC\Vars}{\dgraphB\subst\MailboxesC\Vars}$.\Luca{In
    questo punto serve sapere che in $\dgraphB$ non possono
    comparire nomi arbitrari, $\MailboxesC$ in particolare, per non
    generare cicli.}
  To conclude we apply \refrule{t-sub} observing that
  $\Context \subt \Context_1 \cmul \Context_2 = \Mailbox :
  \In\PatternF', (\Context_1' \cmul \Context_2') \subt \Mailbox :
  \In\PatternF_1, \ContextB, \MailboxesC : \TypesS$ using
  Proposition~\ref{prop:cmul_subt} and the fact that
  $\set\MailboxesC \cap \dom(\ContextB) = \emptyset$.

  \proofrule{r-free}
  Then
  $\ProcessP = \New\Mailbox(\Free\Mailbox\ProcessQ \sumop \Guard)
  \red \ProcessQ$.
  From \refrule{t-sub} and \refrule{t-new} we deduce that there
  exist $\ContextB$ and $\dgraphB$ such that:
  \begin{itemize}
  \item $\ContextA \subt \ContextB$
  \item $\dgraphA \gimplies \New\Mailbox\dgraphB$
  \item
    $\wtp{\ContextB, \Mailbox : \In\tone}{\Free\Mailbox\ProcessQ
      \sumop \Guard}{\dgraphB}$.
  \end{itemize}

  From \refrule{t-sub} and \refrule{t-guard} we deduce that there
  exist $\ContextB'$, $\Pattern$ and $\dgraphB'$ such that:
  \begin{itemize}
  \item
    $\ContextB, \Mailbox : \In\tone \subt \ContextB', \Mailbox :
    \In\Pattern$
  \item
    $\wtg{\ContextB', \Mailbox : \In\Pattern}{\Free\Mailbox\ProcessQ
      \sumop \Guard}$
  \item $\dgraphB \gimplies \gedge\Mailbox{\dom(\ContextB')}$.
  \end{itemize}

  From \refrule{t-branch} and \refrule{t-free} we deduce that there
  exists $\dgraphB'$ such that
  $\wtp{\ContextB'}{\ProcessQ}{\dgraphB'}$.
  From the well formedness condition on judgments and the definition
  of $\subt$ for type environments, we know that
  $\dom(\dgraphB') \subseteq \dom(\ContextB') \subseteq
  \dom(\ContextB) \subseteq \dom(\ContextA)$.
  From the fact that $\dgraphB$ implies a connected graph with
  domain $\set\Mailbox \cup \dom(\ContextB')$ we deduce
  $\dgraphA \gimplies \dgraphB \gimplies \dgraphB'$.
  We conclude with one application of \refrule{t-sub}.

  \proofrule{r-def}
  Then
  $\ProcessP = \invoke\RecVar\MailboxesC \red
  \ProcessR\subst\MailboxesC\Vars = \ProcessQ$ where
  $\define\RecVar\Vars\ProcessR$.
  From \refrule{t-sub} and \refrule{t-def} we deduce that there
  exist $\Types$ and $\dgraphB$ such that
  $\ContextA \subt \MailboxesC : \Types$ and
  $\RecVar : (\Vars : \Types; \dgraphB)$ and
  $\dgraph \gimplies \dgraphB\subst\MailboxesC\Vars$.
  From the hypothesis that all process definitions are well typed we
  know that $\wtp{\Vars : \Types}\ProcessR{\dgraphB}$.
  From Lemma~\ref{lem:subst} we deduce
  $\wtp{\MailboxesC :
    \Types}{\ProcessR\subst\MailboxesC\Vars}{\dgraphB\subst\MailboxesC\Vars}$.
  To conclude we apply \refrule{t-sub}.

  \proofrule{r-par}
  Then
  $\ProcessP = \ProcessP_1 \parop \ProcessP_2 \red \ProcessP_1' \parop
  \ProcessP_2 = \ProcessQ$ where $\ProcessP_1 \red \ProcessP_1'$.
  From \refrule{t-sub} and \refrule{t-par} we deduce that there
  exist $\Context_1$, $\Context_2$, $\dgraph_1$ and $\dgraph_2$ such
  that $\wtp{\Context_i}{\Process_i}{\dgraph_i}$ for $i=1,2$ where
  $\Context \subt \Context_1 \cmul \Context_2$ and
  $\dgraph \gimplies \dgraph_1 \gunion \dgraph_2$.
  From the hypothesis that $\Context$ is reliable and
  Proposition~\ref{prop:env_subt_reliable} we deduce that
  $\Context_1 \cmul \Context_2$ is reliable.
  From the assumption that types are usable and
  Proposition~\ref{prop:env_mul_reliable} we deduce that
  $\Context_1$ is reliable.
  From the induction hypothesis we deduce
  $\wtp{\Context_1}{\ProcessP_1'}{\dgraph_1}$.
  To conclude we apply \refrule{t-par} and then \refrule{t-sub}.

  \proofrule{r-new}
  Then
  $\ProcessP = \New\Mailbox\ProcessP' \red \New\Mailbox\ProcessP'' =
  \ProcessQ$ where $\ProcessP' \red \ProcessP''$.
  From \refrule{t-sub} and \refrule{t-new} we deduce that there
  exist $\ContextB$ and $\dgraphB$ such that
  $\wtp{\ContextB, \Mailbox : \In\tone}{\Process'}{\dgraphB}$ where
  $\Context \subt \ContextB$ and
  $\dgraph \gimplies \New\Mailbox\dgraphB$.
  From the hypothesis that $\ContextA$ is reliable and
  Proposition~\ref{prop:env_subt_reliable} we deduce that
  $\ContextB$ is reliable and so is
  $\ContextB, \Mailbox : \In\tone$.
  From the induction hypothesis we deduce
  $\wtp{\ContextB, \Mailbox : \In\tone}{\ProcessP''}{\dgraphB}$.
  To conclude we apply \refrule{t-new} and then \refrule{t-sub}.
\end{proof}

%%% Local Variables:
%%% mode: latex
%%% TeX-master: "main"
%%% End:

\section{Proofs of Theorems~\ref{thm:soundness} and~\ref{thm:fair_termination}}
\label{sec:proof_soundness}

\begin{reptheorem}{thm:soundness}
  If $\wtp\EmptyContext\ProcessP\dgraph$, then $\ProcessP$ is
  mailbox conformant and deadlock free.
\end{reptheorem}
\begin{proof}
  Immediate consequence of Lemma~\ref{lem:conformance} and
  Lemma~\ref{lem:deadlock_freedom}.
\end{proof}

\begin{reptheorem}{thm:fair_termination}
  If $\wtp\EmptyContext\ProcessP\dgraph$ and $\ProcessP$ is finitely
  unfolding, then $\ProcessP$ is fairly terminating.
\end{reptheorem}
\begin{proof}
  Consider a reduction $\ProcessP \red^* \ProcessQ$. From the
  hypothesis that $\ProcessP$ is finitely unfolding we know that
  there exists $\ProcessR$ such that
  $\ProcessP \red^* \ProcessQ \red^* \ProcessR$ and no reduction of
  $\ProcessR$ uses \refrule{r-def}. We conclude applying
  Theorem~\ref{thm:sr} and Lemma~\ref{lem:fair_termination}.
\end{proof}

\subsection{Proof of mailbox conformance}
\label{sec:proof_conformance}

\begin{lemma}
  \label{lem:fail}
  If $\wtp\Context{\ProcessContext[\Fail\Name]}\dgraph$, then
  $\Context$ is unreliable.
\end{lemma}
\begin{proof}
  We reason by induction on the derivation of
  $\wtp\Context{\ProcessContext[\Fail\Name]}\dgraph$ and by cases on
  the last rule applied, omitting symmetric cases and those ruled
  out by the syntax of $\ProcessContext$.

  \proofrule{t-guard}
  From \refrule{t-fail} we deduce that
  $\Context = \Name : \In\tzero, \ContextB$ and we conclude that
  $\Context$ is unreliable.

  \proofrule{t-par}
  Suppose, without loss of generality, that
  $\ProcessContext = \Process \parop \ProcessContext'$.
  Then there exist $\Context_1$, $\Context_2$, $\dgraph_1$ and
  $\dgraph_2$ such that $\Context = \Context_1 \cmul \Context_2$ and
  $\dgraph \gimplies \dgraph_1 \gunion \dgraph_2$ and
  $\wtp{\Context_1}{\Process}{\dgraph_1}$ and
  $\wtp{\Context_2}{\ProcessContext'[\Fail\Name]}{\dgraph_2}$.
  From the induction hypothesis we deduce that $\Context_2$ is
  unreliable.
  We know that $\Context_1$ is usable since all types are assumed to
  be usable, hence we conclude that $\Context_1 \cmul \Context_2$ is
  unreliable by Proposition~\ref{prop:env_mul_reliable}.

  \proofrule{t-new}
  Then $\ProcessContext = \New\Mailbox\ProcessContext'$ and
  $\wtp{\Context, \Mailbox :
    \In\tone}{\ProcessContext'[\Fail\Name]}\dgraphB$.
  By induction hypothesis we deduce that
  $\Context, \Mailbox : \In\tone$ is unreliable.
  Since $\In\tone$ is reliable, we conclude that $\Context$ is
  unreliable.

  \proofrule{t-sub}
  Then $\wtp\ContextB{\ProcessContext[\Fail\Name]}\dgraphB$ where
  $\ContextA \subt \ContextB$.
  By induction hypothesis we deduce that $\ContextB$ is unreliable
  and we conclude by Proposition~\ref{prop:env_subt_reliable}.
\end{proof}

\begin{lemma}
  \label{lem:conformance}
  If $\wtp\EmptyContext\ProcessP\dgraph$, then
  $\ProcessP \nred^* \ProcessContext[\Fail\Mailbox]$ for all
  $\ProcessContext$ and $\Mailbox$.
\end{lemma}
\begin{proof}
  Immediate consequence of Theorem~\ref{thm:sr} and
  Lemma~\ref{lem:fail}.
\end{proof}

%%% Local Variables:
%%% mode: latex
%%% TeX-master: "main"
%%% End:

\subsection{Proof of deadlock freedom}
\label{sec:proof_deadlock_freedom}

\begin{notation}
  We abbreviate $\wtp\Context\Process\dgraph$ with
  $\wtp\Context\Process{}$ when the dependency graph $\dgraph$ is
  irrelevant.
\end{notation}

\begin{notation}
  Let $\wtp\Context\Process{}$ and $\Name\in\fn(\Process)$. We write
  $\Name^\Type \in \Process$ if $\Name\in\fn(\Process)$ and there is
  a judgment of the form $\wtp{\ContextB, \Name : \Type}\ProcessQ{}$
  in the derivation tree of $\wtp\Context\Process{}$. Similarly for
  guards.
\end{notation}

% \begin{remark}
%   Without the implicit side-condition on judgments
%   $\wtp\Context\Process{}$ imposing that $\Context$ is usable, the
%   property $\wtp\Context\Process{}$ and $\Context$ is reliable
%   implies $\Context$ is usable would not hold. Indeed, take
%   $\ContextA = \MailboxA : \Out\tzero, \MailboxB : \In(\tzero \tmul
%   \tone)$ and
%   $\ContextB = \MailboxA : \In(\tzero \tmul \tone), \MailboxB :
%   \Out\tzero$. Then
%   $\ContextA \cmul \ContextB = \MailboxA : \In\tone, \MailboxB :
%   \In\tone$.
% \end{remark}

\begin{definition}
  We say that $\ProcessQ$ occurs \emph{unguarded} in $\ProcessP$ if
  $\Process = \ProcessContext[\ProcessQ]$ for some
  $\ProcessContext$.
\end{definition}

\begin{lemma}[output occurrence]
  \label{lem:output_occurrence}
  If $\wtp{\Mailbox : \Out\PatternE, \Context}\Process{}$ and
  $\Context$ is reliable, then there exist $n\geq 0$ and
  $\PatternF_1, \dots, \PatternF_n$ such that
  $\PatternF_1 \cdots \PatternF_n \subp \PatternE$ and
  $\Mailbox^{\Out\PatternF_i} \in \Process$ for all
  $1 \leq i\leq n$. Furthermore, the $\Mailbox^{\Out\PatternF_i}$
  occurrences account for all of the unguarded messages stored into
  $\Mailbox$ by $\Process$.
\end{lemma}
\begin{proof}
  By induction on the typing derivation and by cases on the last
  typing rule applied. The fact that all of the unguarded messages
  stored into $\Mailbox$ by $\Process$ are considered follows from
  the structure of the proof, which visits every sub-process of
  $\Process$ in which $\Mailbox$ may occur free.

  \proofrule{t-done}
  This case is impossible because $\Done$ is well typed in the empty
  context only.

  \proofrule{t-msg}
  Then $\Process = \Send\NameU\Tag\NamesV$ and
  $\Mailbox : \Out\PatternE, \Context = \NameU :
  \Out\tmessage\Tag\TypesT, \NamesV : \TypesT$ where $\Mailbox$ is
  either $\NameU$ or one of the $\NameV_i$.
  We conclude by taking $n \eqdef 1$ and
  $\PatternF_1 \eqdef \PatternE$.

  \proofrule{t-def}
  Similar to the previous case.

  \proofrule{t-par}
  Then $\Process = \Process_1 \parop \Process_2$ and
  $\Mailbox : \Out\PatternE, \Context = \Context_1 \cmul \Context_2$
  and $\wtp{\Context_i}{\Process_i}{}$ for every $i=1,2$.
  From the hypothesis that $\Context$ is reliable, the assumption
  that all types are usable and
  Proposition~\ref{prop:env_mul_reliable} we deduce that
  $\Context_1$ and $\Context_2$ are reliable.
  We discuss two interesting sub-cases:
  \begin{itemize}
  \item Suppose $\Context_1 = \Mailbox : \Out\PatternE, \Context_1'$
    and $\Mailbox \not\in\dom(\Context_2)$.
    Then $\Mailbox \not\in \fn(\Process_2)$ and we conclude from the
    induction hypothesis on $\Process_1$.
  \item Suppose $\Context_1 = \Mailbox : \Out\PatternE_1, \Context_1'$
    and $\Context_2 = \Mailbox : \Out\PatternE_2, \Context_2'$ and
    $\PatternE = \PatternE_1 \tmul \PatternE_2$.
    From the induction hypothesis we deduce that there exist
    $n_1 \geq 0$ and $n_2 \geq 0$ and
    $\PatternF_1', \dots, \PatternF_{n_1}'$ and
    $\PatternF_1'', \dots, \PatternF_{n_2}''$ such that
    $\PatternF_1' \cdots \PatternF_{n_1}' \subp \PatternE_1$ and
    $\PatternF_1'' \cdots \PatternF_{n_2}'' \subp \PatternE_2$ and
    $\Mailbox^{\Out\PatternF_i'} \in \Process_1$ for all
    $1\leq i\leq n_1$ and
    $\Mailbox^{\Out\PatternF_i''} \in \Process_2$ for all
    $1\leq i\leq n_2$.
    We conclude by taking $n \eqdef n_1 + n_2$ and
    $\PatternF_1, \dots, \PatternF_n \eqdef \PatternF_1', \dots,
    \PatternF_{n_1}', \PatternF_1'', \dots, \PatternF_{n_2}''$.
  \end{itemize}

  \proofrule{t-new}
  Then $\ProcessP = \New\MailboxC\ProcessQ$ and
  $\wtp{\Mailbox : \Out\PatternE, \Context, \MailboxC :
    \In\tone}\ProcessQ{}$.
  Since $\In\tone$ is reliable we conclude by the induction
  hypothesis.

  \proofrule{t-guard}
  Then $\Process = \Guard$ and $\wtg\Context\Guard$.
  We conclude by the induction hypothesis.

  \proofrule{t-sub}
  Then $\wtp\ContextB\Process{}$ where
  $\Mailbox : \Out\PatternE, \ContextA \subt \ContextB$.
  We have two possibilities:
  \begin{itemize}
  \item Suppose $\Mailbox\in\dom(\ContextB)$.
    Then $\ContextB = \Mailbox : \Out\PatternE', \ContextB'$ where
    $\PatternE' \subp \PatternE$ and $\ContextA \subt \ContextB'$.
    From Proposition~\ref{prop:env_subt_reliable} we deduce that
    $\ContextB'$ is reliable.
    We conclude by the induction hypothesis using transitivity of
    $\subp$.
  \item Suppose $\Mailbox\not\in\dom(\ContextB)$.
    Then $\Mailbox\not\in\fn(\Process)$ and $\Out\PatternE$ is
    irrelevant, meaning that $\tone \subp \PatternE$.
    We conclude by taking $n \eqdef 0$.
  \end{itemize}

  \proofrule{t-fail}
  This case is impossible because $\Context$ is reliable by
  hypothesis.

  \proofrule{t-free}
  Then $\Process = \Free\NameU\ProcessQ$ and
  $\Mailbox : \Out\PatternE, \Context = \Name : \In\tone, \ContextB$
  and $\wtp\ContextB\ProcessQ{}$.
  Clearly $\Mailbox \ne \Name$. Also, from the hypothesis that
  $\Context$ is reliable we deduce that $\ContextB$ is reliable.
  We conclude by the induction hypothesis.

  \proofrule{t-in}
  Then $\Process = \Receive\Name\Tag\Vars\ProcessQ$ and
  $\Mailbox : \Out\PatternE, \Context = \Name :
  \In(\tmessage\Tag\TypesT \tmul \PatternF), \ContextB$ and
  $\wtp{\Name : \In\PatternF, \ContextB, \Vars :
    \TypesT}\ProcessQ{}$.
  Clearly $\Mailbox \ne \Name$ therefore
  $\ContextB = \Mailbox : \Out\PatternE, \ContextB'$.
  From the hypothesis that $\Context$ is reliable we deduce that
  $\PatternF \not\subp \tzero$.
  From the assumption that argument types are reliable we deduce
  that $\Name : \In\PatternF, \ContextB, \Vars : \TypesT$ is
  reliable.
  We conclude by the induction hypothesis.

  \proofrule{t-branch}
  Then $\Process = \Guard_1 \sumop \Guard_2$ and
  $\Mailbox : \Out\PatternE, \Context = \NameU : \In(\PatternE_1
  \tsum \PatternE_2), \ContextB$ and
  $\wtg{\NameU : \In\PatternE_i, \ContextB}{\Guard_i}$ for every
  $i=1,2$.
  Clearly $\Mailbox \ne \NameU$, therefore
  $\ContextB = \Mailbox : \Out\PatternE, \ContextB'$.
  From the hypothesis that $\Context$ is reliable we deduce that
  $\PatternE_1 \tsum \PatternE_2 \not\subp \tzero$.  Suppose,
  without loss of generality, that $\PatternE_1 \not\subp \tzero$.
  We conclude by the induction hypothesis on $\Guard_1$.
\end{proof}

\begin{lemma}[input occurrence]
  \label{lem:input_occurrence}
  If $\wtp{\Mailbox : \In\PatternE, \Context}\Process{}$ and
  $\Context$ is reliable and
  $\MessageType \tmul \PatternF \subp \PatternE$ and
  $\PatternF \not\subp \tzero$, then there exist $\MessageType'$ and
  $\PatternF'$ such that $\MessageType \subp \MessageType'$ and
  $\PatternF' \not\subp \tzero$ and
  $\Mailbox^{\In\MessageType'\tmul\PatternF'} \in \Process$.
\end{lemma}
\begin{proof}
  By induction on the derivation of
  $\wtp{\Mailbox : \In\PatternE, \Context}\Process{}$ and by cases
  on the last rule applied, recalling that a guard is also a
  process.

  \proofrule{t-done}
  This case is impossible because $\Done$ is well typed in the empty
  context only.

  \proofrule{t-msg}
  Then $\Process = \Send\NameU\Tag\NamesV$ and
  $\Mailbox : \In\PatternE, \Context = \NameU :
  \Out\tmessage\Tag\TypesT, \NamesV : \TypesT$.
  Clearly $\Mailbox \ne \NameU$, so we conclude by taking
  $\MessageType \eqdef \MessageType'$ and
  $\PatternF' \eqdef \PatternF$.

  \proofrule{t-def}
  Similar to the previous case.

  \proofrule{t-par}
  Then $\Process = \Process_1 \parop \Process_2$ and
  $\Mailbox : \In\PatternE, \Context = \Context_1 \cmul \Context_2$
  and $\wtp{\Context_i}{\Process_i}{}$ for $i=1,2$.
  We only discuss one interesting case when
  $\Context_1 = \Mailbox : \Out\PatternE_1, \Context_1'$ and
  $\Context_2 = \Mailbox : \In(\PatternE_1 \tmul \PatternE),
  \Context_2'$.
  From the assumption that all types are usable we know that
  $\PatternE_1 \not\subp \tzero$ and both $\Context_1'$ and
  $\Context_2'$ are usable.
  From the hypothesis that $\Context$ is reliable we deduce that
  $\Context_1' \cmul \Context_2'$ is also reliable.
  From Proposition~\ref{prop:env_mul_reliable} we deduce that
  $\Context_1'$ and $\Context_2'$ are both reliable.
  From the hypothesis $\MessageType \tmul \PatternF \subp \PatternE$
  and the precongruence of $\subp$ we deduce
  $\MessageType \tmul \PatternE_1 \tmul \PatternF \subp \PatternE_1
  \tmul \PatternE$.
  From $\PatternE_1 \not\subp \tzero$ and the hypothesis
  $\PatternF \not\subp \tzero$ we deduce
  $\PatternE_1 \tmul \PatternF \not\subp \tzero$.
  We conclude by the induction hypothesis.

 \proofrule{t-new}
  Then $\ProcessP = \New\MailboxC\ProcessQ$ and
  $\wtp{\Mailbox : \In\PatternE, \Context, \MailboxC :
    \In\tone}\ProcessQ{}$.
  Observe that $\In\tone$ is reliable, so we can conclude by the
  induction hypothesis.

  \proofrule{t-sub}
  Then $\wtp\ContextB\Process$ for some $\ContextB$ such that
  $\Mailbox : \In\PatternE, \ContextA \subt \ContextB$, which
  implies $\ContextB = \In\PatternE', \ContextB'$ where
  $\PatternE \subp \PatternE'$ and $\ContextA \subt \ContextB'$.
  We have
  $\MessageType \tmul \PatternF \subp \PatternE \subp \PatternE'$ by
  transitivity of $\subp$.
  From the hypothesis that $\ContextA$ is reliable and
  Proposition~\ref{prop:env_subt_reliable} we deduce that
  $\ContextB'$ is reliable.
  We conclude by the induction hypothesis.

  \proofrule{t-guard}
  Straightforward application of the induction hypothesis.

  \proofrule{t-fail}
  This case is impossible because $\PatternE \not\subp \tzero$ and
  $\Context$ is reliable.

  \proofrule{t-free}
  Then $\ProcessP = \Free\Name\ProcessQ$ and
  $\Mailbox : \In\PatternE, \Context = \Name : \In\tone, \ContextB$
  and $\wtp\ContextB\ProcessQ{}$.
  From the hypotheses $\MessageType \tmul \PatternF \subp \PatternE$
  and $\PatternF \not\subp \tzero$ we deduce $\Mailbox \ne \Name$,
  hence $\ContextB = \Mailbox : \In\PatternE, \ContextB'$ for some
  $\ContextB'$ such that $\ContextB'$ is usable.
  We conclude by the induction hypothesis.

  \proofrule{t-in}
  Then $\Process = \Receive\Name\Tag\Vars\ProcessQ$ and
  $\Mailbox : \In\PatternE, \Context = \Name :
  \In(\tmessage\Tag\Types \tmul \PatternE'), \ContextB$ and
  $\wtp{\Name : \In\PatternE', \ContextB, \Vars :
    \Types}\ProcessQ{}$.
  We distinguish the following sub-cases:
  \begin{itemize}
  \item Suppose $\Mailbox = \Name$ and
    $\MessageType \subp \tmessage\Tag\Types$.
    Then $\PatternE = \tmessage\Tag\Types \tmul \PatternE'$.
    From the hypotheses
    $\MessageType \tmul \PatternF \subp \PatternE$ and
    $\PatternF \not\subp \tzero$ we deduce
    $\PatternE' \not\subp \tzero$.
    We conclude by taking $\MessageType' \eqdef \tmessage\Tag\Types$
    and $\PatternF' \eqdef \PatternE'$.
  \item Suppose $\Mailbox = \Name$ and
    $\MessageType \not\subp \tmessage\Tag\Types$.
    From the hypotheses
    $\MessageType \tmul \PatternF \subp \PatternE$ and
    $\PatternF \not\subp \tzero$ we deduce that
    $\MessageType \tmul \PatternE'' \subp \PatternE'$ for some
    $\PatternE'' \not\subp \tzero$.
    Observe that $\ContextB = \ContextA$.
    From the assumption that argument types are reliable we know
    that all the types in $\Types$ are reliable. Therefore we can
    conclude by the induction hypothesis.
  \item Suppose $\Mailbox \ne \Name$.
    From the hypothesis that $\Context$ is reliable we deduce
    $\PatternE' \not\subp \tzero$.
    From the assumption that argument types are reliable we know
    that all the types in $\Types$ are reliable. Therefore we can
    conclude by the induction hypothesis.
  \end{itemize}

  \proofrule{t-branch}
  Then $\Process = \Guard_1 \sumop \Guard_2$ and
  $\Mailbox : \In\PatternE, \Context = \Name : \In(\PatternE_1 \tsum
  \PatternE_2), \ContextB$ and
  $\wtg{\Name : \In\PatternE_i, \ContextB}{\Guard_i}$.
  We distinguish two sub-cases:
  \begin{itemize}
  \item If $\Mailbox = \Name$, then
    $\PatternE = \PatternE_1 \tsum \PatternE_2$ and
    $\ContextA = \ContextB$.
    From the hypothesis
    $\MessageType \tmul \PatternF \subp \PatternE$ we deduce
    $\MessageType \tmul \PatternF \subp \PatternE_i$ for some
    $i=1,2$.\Luca{Questa proproet\`a richiede un lemma non banale
      che usa anche l'ipotesi che $\PatternE$ \`e in forma normale.}
    We conclude by the induction hypothesis.
  \item If $\Mailbox \ne \Name$, then
    $\ContextB = \Mailbox : \In\PatternE, \ContextB'$.
    From the hypothesis that $\ContextA$ is reliable we deduce
    $\Pattern_1 + \Pattern_2 \not\subp \tzero$.
    Suppose, without loss of generality, that
    $\Pattern_1 \not\subp \tzero$.  We conclude by applying the
    induction hypothesis on $\Guard_1$.
    \qedhere
  \end{itemize}
\end{proof}

\begin{lemma}
  \label{lem:termination}
  If $\wtp\EmptyContext\ProcessP\dgraphB$ and $\ProcessP \nred$,
  then $\ProcessP \equiv \Done$.
\end{lemma}
\begin{proof}
  Using the hypothesis $\ProcessP \nred$ and the laws of structural
  congruence we deduce that
  \[
    \ProcessP \equiv \New\Mailboxes\ProcessQ
    \qquad
    \text{where}
    \qquad
    \ProcessQ =
    \prod_{i\in I} \Send{\Mailbox_i}{\Tag_i}{\MailboxesC_i}
    \parop
    \prod_{\mathclap{j\in J}}
    \Guard_j
  \]
  where $\wtp{\Mailboxes : \In\tone}{\ProcessQ}{\dgraph}$ and
  $\dgraph$ is acyclic and each guarded process $\Guard_j$ concerns
  some mailbox $\Mailbox_j$.
  We prove the result by contradiction, assuming that
  $I \cup J \ne \emptyset$. The proof proceeds in two steps. In the
  first step we show that, for every $m\in I\cup J$, there exist
  $n\in I\cup J$ such that
  $\dgraph \gimplies \gedge{\Mailbox_n}{\Mailbox_m}$ and $\dgraph$
  contains an edge where $\Mailbox_n$ and $\Mailbox_m$ occur
  precisely in this order. In the second step, we show that this
  ultimately leads to a cyclic graph.

  Suppose $\ProcessQ \equiv \ProcessR \parop \Guard$ where
  $\Guard \equiv \sum_{h\in H}
  \Receive\Mailbox{\Tag_h}{\Vars_h}\Process_h\set{{} \sumop
    \Free\Mailbox\Process'}$.
  From \refrule{t-sub} and \refrule{t-par} we deduce that there
  exist $\Context_1$, $\Context_2$, $\PatternE$, $\PatternF$,
  $\dgraph_1$ and $\dgraph_2$ such that:
  \begin{itemize}
  \item
    $\wtp{\Context_1, \Mailbox : \Out\PatternE}\ProcessR{\dgraph_1}$
  \item
    $\wtp{\Context_2, \Mailbox : \In(\PatternE \tmul
      \PatternF)}\Guard{\dgraph_2}$
  \item $\tone \subp \PatternF$
  \end{itemize}
  where $\Context_1$ and $\Context_2$ are both reliable.
  From \refrule{t-sub} and \refrule{t-guard} we deduce that
  $\PatternE \tmul \PatternF \subp \sum_{h\in H}
  \tmessage{\Tag_h}{\Types_h} \tmul \PatternE_h\set{{} \tsum
    \tone}$.
  From Lemma~\ref{lem:output_occurrence} we deduce that there exist
  $n\geq 0$ and $\PatternF_1, \dots, \PatternF_n$ such that
  $\PatternF_1 \cdots \PatternF_n \subp \PatternE$ and
  $\Mailbox^{\Out\PatternF_k} \in \ProcessR$ for all
  $1\leq k\leq n$.
  Suppose $n = 0$. Then $\tone \subp \PatternE$ and therefore
  $\tone \subp \PatternE \tmul \PatternF$, meaning that the
  $\Free\Mailbox\ProcessP'$ sub-term must be present in
  $\Guard$. This is absurd for we know that $\Process \nred$, hence
  there must be at least one occurrence of $\Mailbox$ in
  $\ProcessR$.
  In particular, one of the $\PatternF_k$ for $1\leq k\leq n$
  generates a non-empty subset of the $\tmessage{\Tag_h}{\Types_h}$
  for $h\in H$.
  We reason by cases on the places where the
  $\Mailbox^{\Out\PatternF_k}$ may occur:
  \begin{itemize}
  \item From the hypothesis $\ProcessP \nred$ we deduce that for
    every $i\in I$ either $\Mailbox_i \ne \Mailbox$ or
    $\Tag_i \ne \Tag_h$ for every $h\in H$. Therefore, the
    $\Mailbox^{\Out\PatternF_k}$ cannot be any of the $\Mailbox_i$.
  \item If $\Mailbox^{\Out\PatternF_k}$ occurs in
    $\set{\MailboxesC_i}$ for some $i\in I$, then
    $\Mailbox_i \ne \Mailbox$ and
    $\dgraph \gimplies \gedge{\Mailbox_i}{\Mailbox}$.
  \item All of the $\Mailbox_j$ have an input capability, therefore
    the $\Mailbox^{\Out\PatternF_k}$ cannot be any of them.
  \item If $\Mailbox^{\Out\PatternF_k}$ occurs in $\Guard_j$ for
    some $j\in J$, then $\Mailbox_j \ne \Mailbox$ because $\Mailbox$
    is already used for input in $\Guard$, so we have
    $\dgraph \gimplies \gedge{\Mailbox_j}{\Mailbox}$.
  \end{itemize}

  Suppose
  $\ProcessQ \equiv \Send\Mailbox\Tag\MailboxesC \parop \ProcessR$.
  From \refrule{t-sub}, \refrule{t-par} and \refrule{t-msg} we
  deduce that there exist $\Context_1$, $\Context_2$, $\PatternE$,
  $\PatternF$, $\dgraph_1$ and $\dgraph_2$ such that:
  \begin{itemize}
  \item
    $\wtp{\Context_1, \Mailbox :
      \Out\tmessage\Tag\Types}{\Send\Mailbox\Tag\MailboxesC}{\dgraph_1}$
  \item
    $\wtp{\Context_2, \Mailbox : \In(\PatternE \tmul
      \PatternF)}\ProcessR{\dgraph_2}$
  \item $\tmessage\Tag\Types \subp \PatternE$ and
    $\tone \subp \PatternF$
  \end{itemize}
  where $\Context_1$ and $\Context_2$ are both reliable.
  From Lemma~\ref{lem:input_occurrence} we deduce that there exist
  $\MessageType$ and $\PatternF'$ such that
  $\tmessage\Tag\Types \subp \MessageType$ and
  $\PatternF' \not\subp \tzero$ and
  $\Mailbox^{\In\MessageType \tmul \PatternF'} \in \ProcessR$.
  We reason by cases on the places where
  $\Mailbox^{\In\MessageType \tmul \PatternF'}$ may occur:
  \begin{itemize}
  \item All the $\Mailbox_i$ have an output capability, therefore
    $\Mailbox^{\In\MessageType\tmul\PatternF'}$ cannot be any of
    them.
  \item If $\Mailbox^{\In\MessageType\tmul\PatternF'}$ occurs in
    $\set{\MailboxesC_i}$ for some $i\in I$, then
    $\Mailbox_i \ne \Mailbox$ and
    $\dgraph \gimplies \gedge{\Mailbox_i}{\Mailbox}$.
  \item If $\Mailbox^{\In\MessageType\tmul\PatternF'} \in \Guard_j$
    where $\Mailbox_j \ne \Mailbox$, then
    $\dgraph \gimplies \gedge{\Mailbox_j}{\Mailbox}$.
  \item If
    $\Mailbox^{\In\MessageType\tmul\PatternF'} \in \Guard_j \equiv
    \sum_{h\in H} \Receive\Mailbox{\Tag_h}{\Vars_h}\Process_h\set{{}
      \sumop \Free\Mailbox\Process'}$, then we can reason as in the
    case for inputs and find some $\Mailbox_k \ne \Mailbox$ such
    that $\dgraph \gimplies \gedge{\Mailbox_k}\Mailbox$.
  \end{itemize}

  In summary, we have seen that starting from the assumption that
  $I \cup J \ne \emptyset$ it is possible to obtain a set $A$ of
  mailbox names that contains at least two distinct elements and
  such that for every $\MailboxA \in A$ there exists
  $\MailboxB \in A$ such that
  $\dgraph \gimplies \gedge\MailboxB\MailboxA$ and $\dgraph$
  contains an edge where $\MailboxB$ and $\MailboxA$ precisely occur
  in this order. The set $A$ is necessarily finite because
  $\ProcessQ$ is finite, therefore $\grel\dgraph$ must be reflexive,
  which contradicts the hypothesis that $\dgraph$ is acyclic. This
  is absurd, hence we conclude $I = J = \emptyset$.
\end{proof}

\begin{lemma}
  \label{lem:deadlock_freedom}
  If $\wtp\EmptyContext\ProcessP\dgraph$ and
  $\ProcessP \red^* \ProcessQ \nred$, then $\ProcessQ \equiv \Done$.
\end{lemma}
\begin{proof}
  Straightforward consequence of Theorem~\ref{thm:sr} and
  Lemma~\ref{lem:termination}.
\end{proof}

%%% Local Variables:
%%% mode: latex
%%% TeX-master: "main"
%%% End:

\subsection{Fair termination for finitely unfolding processes}
\label{sec:proof_termination}

\begin{lemma}
  \label{lem:fair_termination}
  If $\wtp\EmptyContext\Process\dgraph$ and no reduction of
  $\ProcessP$ uses \refrule{r-def}, then $\ProcessP \red^* \Done$.
\end{lemma}
\begin{proof}
  Let $\size(\Process)$ be the function inductively defined by the
  following equations
  \[
    \begin{array}{rcl}
      \size(\Done) = \size(\Fail\Name) = \size(\Send\Name\Tag\NamesV) = \size(\invoke\RecVar\Names) & \eqdef & 0
      \\
      \size(\Free\Name\ProcessQ) = \size(\Receive\Name\Tag\Vars\ProcessQ)
      & \eqdef & 1 + \size(\ProcessQ)
      \\
      \size(\Guard_1 \sumop \Guard_2) & \eqdef &
      \max\set{ \size(\Guard_1), \size(\Guard_2) }
      \\
      \size(\Process_1 \parop \Process_2) & \eqdef &
      \size(\Process_1) + \size(\Process_2)
      \\
      \size(\New\Mailbox\ProcessQ) & \eqdef & \size(\ProcessQ)
    \end{array}
  \]
  and observe that $\ProcessP \red \ProcessQ$ implies
  $\size(\ProcessQ) < \size(\ProcessP)$ from the hypothesis that no
  reduction of $\ProcessP$ uses \refrule{r-def}.  Then $\ProcessP$
  is strongly normalizing and we conclude by
  Lemma~\ref{lem:deadlock_freedom}.
\end{proof}

%%% Local Variables:
%%% mode: latex
%%% TeX-master: "main"
%%% End:

%%% Local Variables:
%%% mode: latex
%%% TeX-master: "main"
%%% End:

\section{Readers-Writer Lock}
\label{sec:example_readers_writer}

\newcommand{\DefFree}{\RecVar[Free]}
\newcommand{\DefWrite}{\RecVar[Write]}
\newcommand{\DefRead}{\RecVar[Read]}
\newcommand{\VarWriter}{\Var[w]}
\newcommand{\VarReader}{\Var[r]}
\newcommand{\TagAcquireWriter}{\Tag[acquireW]}
\newcommand{\TagAcquireReader}{\Tag[acquireR]}
\newcommand{\TagRead}{\Tag[read]}
\newcommand{\TagWrite}{\Tag[write]}

A readers-writer lock grants read-only access to an arbitrary number
of readers and exclusive write access to a single writer. Below is a
particular modeling of a readers-writer lock making use of
\emph{mixed guards}, where different actions may refer to different
mailboxes:
\[
  \begin{array}{@{}r@{~}c@{~}l@{}}
    \DefFree(\VarSelf) & \triangleq &
    \Free\VarSelf\Done
    \\ & \sumop &
    \Receive\VarSelf\TagAcquireWriter\VarWriter\parens{
      \Send\VarWriter\TagReply\VarSelf
      \parop
      \invoke\DefWrite\VarSelf
    }
    \\ & \sumop &
    \Receive\VarSelf\TagAcquireReader\VarReader
    \New\VarPool\parens{
      \Send\VarReader\TagReply\VarPool
      \parop
      \invoke\DefRead{\VarSelf,\VarPool}
    }
    \\ & \sumop &
    \Receive\VarSelf\TagRead{}
    \Fail\VarSelf
    \\ & \sumop &
    \Receive\VarSelf\TagWrite\VarWriter
    \Fail\VarSelf
    \\
    \DefWrite(\VarSelf) & \triangleq &
    \Receive\VarSelf\TagRelease{}
    \invoke\DefFree\VarSelf
    \\ & \sumop &
    \Receive\VarSelf\TagWrite\VarWriter\parens{
      \Send\VarWriter\TagReply\VarSelf
      \parop
      \invoke\DefWrite\VarSelf
    }
    \\ & \sumop &
    \Receive\VarSelf\TagRead{}
    \Fail\VarSelf
    \\
    \DefRead(\VarSelf,\VarPool) & \triangleq &
    \Receive\VarSelf\TagAcquireReader\VarReader\parens{
      \Send\VarReader\TagReply\VarPool
      \parop
      \invoke\DefRead{\VarSelf,\VarPool}
    }
    \\ & \sumop &
    \Receive\VarSelf\TagWrite\VarWriter
    \Fail\VarSelf
    \\ & \sumop &
    \Receive\VarPool\TagRead{}
    \invoke\DefRead{\VarSelf,\VarPool}
    \\ & \sumop &
    \Free\VarPool\invoke\DefFree\VarSelf
  \end{array}
\]

The readers-writer lock may be free, in which case it can be
acquired either by a reader or by a writer but it does not accept
read or write requests. When the lock has been acquired by a writer
process, the writer is granted exclusive access and read requests
are not accepted. When a reader acquires a free lock, the lock
creates $\VarPool$ of readers and moves into a state where more
readers may be granted access. Each reader manifests its intention
to read the resource by storing a $\TagRead$ message into $\VarPool$
whereas prospective readers manifest their intention of accessing
the resource by storing $\TagAcquireReader$ into $\VarSelf$. The
mixed guard in the $\DefRead$ process is key to serve both kinds of
requests. In particular, when all readers have terminated,
$\VarPool$ can be deleted and the lock moves back into the free
state.

\begin{table}
  \begin{center}
    \[
      \begin{array}{@{}c@{}}
        \inferrule{
          \wtg{
            \Names : \In\Patterns,
            \Context;
            \Names : \In\Patterns
          }{
            \Guard
          }
          \\
          \vDash \Patterns
        }{
          \textstyle
          \wtp{
            \Names : \In\Patterns,
            \Context
          }{
            \Guard
          }{
            \gedge{\set\Names}{\dom(\Context)}
          }
        }
        ~~\defrule{t-guard*}
        \\\\
        \inferrule{
          % \mathstrut
        }{
          \wtg{
            \ContextA;
            \Name : \In\tzero
          }{
            \Fail\Name
          }
        }
        ~~\defrule{t-fail*}
        \qquad
        \inferrule{
          \wtp\ContextA\Process\dgraph
        }{
          \wtg{
            \Name : \Type, \ContextA;
            \Name : \In\tone
          }{
            \Free\Name\Process
          }
        }
        ~~\defrule{t-free*}
        \\\\
        \inferrule{
          \wtp{
            \Name : \In\PatternE,
            \Context,
            \Vars : \Types
          }{
            \Process
          }{
            \dgraph
          }
          % \\
          % \text{$\Types$ usable}
        }{
          \wtg{
            \Name : \Type, \ContextA;
            \Name : \In(\tmessage\Tag\Types\tmul\PatternE)
          }{
            \Receive\Name\Tag\Vars\Process
          }
        }
        ~~\defrule{t-in*}
        \qquad
        \inferrule{
          \wtg{
            \ContextA;
            \ContextB_i
          }{
            \Guard_i
          }
          ~{}^{(i=1,2)}
        }{
          \wtg{
            \ContextA;
            \ContextB_1 + \ContextB_2
          }{
            \Guard_1 \sumop \Guard_2
          }
        }
        ~~\defrule{t-branch*}
      \end{array}
    \]
  \end{center}
  \caption{
    \label{tab:relaxed_typing} Relaxed typing rules for mixed guards.
  }
\end{table}

In order to deal with mixed guards, the typing rules for guards and
guarded processes must be generalized as shown in
Table~\ref{tab:relaxed_typing}. The basic idea is the same as for
restricted guards, namely the type of mailboxes different from the
one referred to by an action cannot be affected by the content (or
lack thereof) of that mailbox. To do so, the judgments for guards
have the form $\wtg{\ContextA;\ContextB}\Guard$ where $\ContextA$ is
the type environment of the guarded process as a whole whereas
$\ContextB$ keeps track of the types of the mailboxes referred to by
the actions, which are split exactly as in the restricted typing
rules (Table~\ref{tab:inference}). The $+$ operation on type
environments is defined thus:
\[
  \Context_1 + \Context_2 =
  \begin{cases}
    \Context_1, \Context_2 & \text{if
      $\dom(\Context_1) \cap \dom(\Context_2) = \emptyset$}
    \\
    \Name : \In(\Pattern_1 \tsum \Pattern_2), \Context_1' +
    \Context_2' & \text{if
      $\Context_1 = \Name : \In\Pattern_1, \Context_1'$ and
      $\Context_2 = \Name : \In\Pattern_2, \Context_2'$}
  \end{cases}
\]

With the relaxed rules, it is possible to show that the above
process definitions are consistent with the following declarations
\[
  \begin{array}{@{}r@{~}c@{~}l@{}}
    \DefFree & : &
    \parens{
      \VarSelf : \In\parens{
        \tmessage\TagAcquireWriter\TypeT\tstar
        \tmul
        \tmessage\TagAcquireReader\TypeR\tstar
      };
      \gempty
    }
    \\
    \DefWrite & : &
    \parens{
      \VarSelf : \In\parens{
        \tmessage\TagAcquireWriter\TypeT\tstar
        \tmul
        \tmessage\TagAcquireReader\TypeR\tstar
        \tmul
        \parens{
          \TagRelease
          \tsum
          \tmessage\TagWrite\TypeT
        }
      };
      \gempty
    }
    \\
    \DefRead & : &
    \parens{
      \VarSelf : \In\parens{
        \tmessage\TagAcquireWriter\TypeT\tstar
        \tmul
        \tmessage\TagAcquireReader\TypeR\tstar
      },
      \VarPool : \In\TagRead\tstar;
      \gedge\VarSelf\VarPool
    }
  \end{array}
\]
where
$\TypeT \eqdef \Out\tmessage\TagReply{ \Out\parens{ \TagRelease
    \tsum \tmessage\TagWrite\TypeT } }$ and
$\TypeR = \Out\tmessage\TagReply{ \Out\TagRead\tstar }$.
The remarkable aspect of this typing is that it guarantees
statically the key properties of the readers-writer lock: that there
are no readers nor writers when the lock is free; that there are no
readers if there is a single writer; that there is no writer if
there are one or more readers. Once again, the modeling makes key
use of multiple mailboxes so as to keep the typing as precise as
possible.

%%% Local Variables:
%%% mode: latex
%%% TeX-master: "main"
%%% End:

\end{document}